\documentclass[11pt]{article}
\usepackage{setspace,graphicx,amssymb,amsmath,mathtools,latexsym,amsfonts,amscd,amsthm,multirow,mathdots,caption,array,dsfont,wrapfig}
\usepackage{fancyhdr,tabularx,cite,mathrsfs,enumitem,graphicx,pifont,booktabs}
\usepackage{stmaryrd}
\usepackage{rotating}
\usepackage{authblk}
\usepackage{hyperref}
\usepackage{subcaption}
\usepackage{diagbox}

\usepackage{comment,enumitem}
\usepackage{geometry}
\usepackage{arydshln}
\usepackage{textgreek}

\usepackage{hyperref}
\usepackage[font=small,labelfont=bf]{caption}
\usepackage{yfonts}

\usepackage{lscape}
\usepackage{tikz}
\usepackage{ctable}
\usetikzlibrary{calc,trees,positioning,arrows,chains,shapes.geometric,
	decorations.pathreplacing,decorations.pathmorphing,shapes,
	matrix,shapes.symbols}

\usepackage{tikz-cd}

\tikzset{
	>=stealth',
	punktchain/.style={
		rectangle,
		draw=black, very thick,
		text width=15em,
		minimum height=2em,
		text centered,
		on chain},
	line/.style={draw, thick, <-},
	element/.style={
		tape,
		top color=white,
		bottom color=blue!50!black!60!,
		minimum width=8em,
		draw=blue!40!black!90, very thick,
		text width=10em,
		minimum height=3.5em,
		text centered,
		on chain},
	every join/.style={->, thick,shorten >=1pt},
	decoration={brace},
	tuborg/.style={decorate},
	tubnode/.style={midway, right=7pt},
}

\usepackage[all,cmtip]{xy}

\newcommand{\Bf}[0]{\boldsymbol}

\newcommand{\RR}{\mathbb{R}}

\newcommand{\ZZ}{\mathbb{Z}}

\newcommand{\xmod}{{\rm \;mod\;}}

\newcommand{\bea}{\begin{eqnarray}}
\newcommand{\eea}{\end{eqnarray}}
\newcommand{\bee}{\begin{eqnarray*}}
	\newcommand{\eee}{\end{eqnarray*}}
\newcommand{\al}{\begin{align*}}
\newcommand{\eal}{\end{align*}}
\newcommand{\be}{\begin{equation}}
\newcommand{\ee}{\end{equation}}

\newcommand{\bem}{\begin{pmatrix}}
	\newcommand{\eem}{\end{pmatrix}}

\def\a{\alpha}
\def\b{\beta}

\def\l{\lambda}

\def\n{\nu}

\def\o{\omega}

\def\L{\Lambda}

\def\Tr{\tr}

\def\id{\mathds{1}}

\newcommand{\ex}{\operatorname{e}} 

\newcommand{\tr}{\operatorname{tr}}

\def\be{\begin{equation}}
\def\ee{\end{equation}}
\def\bea{\begin{eqnarray}}
\def\eea{\end{eqnarray}}

\newcommand{\lV}{log-${\cal V}_{\bar \Lambda}(m)$}
\newcommand{\lVb}{log-${\cal V}_{\bar\Lambda}(p,p')$}

\newcommand{\lVs}{log-${\cal V}_{\bar \Lambda}^0(m)$}
\newcommand{\lVbs}{log-${\cal V}_{\bar\Lambda}^0(p,p')$}

\usepackage{pifont}
\newcommand{\cmark}{\ding{51}}%
\newtheorem{thm}{Theorem}[section]
\newtheorem{cor}[thm]{Corollary}
\newtheorem{lem}[thm]{Lemma}
\newtheorem{prop}[thm]{Proposition}

\theoremstyle{definition}
\newtheorem{defn}[thm]{Definition}

\theoremstyle{remark}

\numberwithin{equation}{section}

\DeclareMathOperator{\Ck}{Coker}

\def\CN {{\cal N}}

\def\CH {{\cal H}}

\begin{document}
	
\title{\huge \bf 3-Manifolds and VOA Characters}
\author{Miranda C. N. Cheng$^{a,b,c}$,
Sungbong Chun$^{d}$, Boris Feigin$^{e,f,g}$, Francesca Ferrari$^{h}$, Sergei Gukov$^{i}$, Sarah M. Harrison$^{j,k}$, Davide Passaro$^b$}

\date{%
  $^a$Korteweg-de Vries Institute for Mathematics
  University of Amsterdam, Amsterdam, the Netherlands\\%
  $^b$Institute of Physics, University of Amsterdam, Amsterdam, the Netherlands\\
  $^c$Institute for Mathematics, Academia Sinica, Taipei, Taiwan\\
  $^d$New High Energy Theory Center, Rutgers University, Piscataway, NJ 08854, USA\\
  $^e$National Research University Higher School of Economics, Russia\\
  $^f$International Laboratory of Representation Theory, Mathematical Physics, Myasnitskaya ul., 20, Moscow, 101000, Russia\\
  $^g$Landau Institute for Theoretical Physics, pr. Akademika Semenova, 1a, Chernogolovka, 142432, Russia\\
  $^h$International Centre for Theoretical Physics, Strada Costiera 11, Trieste 34151, Italy\\
  $^i$Division of Physics, Mathematics and Astronomy, California Institute of Technology,
  1200 E. California Blvd., Pasadena, CA 91125, USA\\
  $^j$Department of Mathematics and Statistics, McGill University, Montreal, QC, Canada\\
  $^k$Department of Physics, McGill University, Montreal, QC, Canada\\[2ex]
}
\maketitle
\begin{abstract}
  By studying the properties of $q$-series $\widehat Z$-invariants, we develop a dictionary between 3-manifolds and vertex algebras. In particular, we generalize previously known entries in this dictionary to Lie groups of higher rank, to 3-manifolds with toral boundaries, and to BPS partition functions with line operators. This provides a new physical realization of logarithmic vertex algebras in the framework of the 3d-3d correspondence and opens new avenues for their future study. For example, we illustrate how invoking a knot-quiver correspondence for $\widehat{Z}$-invariants leads to many infinite families of new fermionic formulae for VOA characters.

\end{abstract}
\newpage
\tableofcontents

\newpage
\section{Introduction and Summary of Results}

The main goal of this paper is to build a new bridge between different areas of mathematics and mathematical physics. Namely, we explore the relation between characters of vertex operator algebras (VOAs), on the one hand, and $q$-series invariants of manifolds in low-dimensional topology, on the other hand.

A prototypical example of such a duality (or, correspondence) that goes back to the mid-90s involves Vafa-Witten invariants of 4-manifolds \cite{Vafa:1994tf}. Starting with the seminal work of Nakajima \cite{Nakajima:1994nid}, these $q$-series invariants of 4-manifolds can be interpreted as VOA characters, see {\it e.g.} \cite{Feigin:2018bkf} for a recent account and identification of cutting-and-gluing operations on both sides of the correspondence. The Vafa-Witten invariants have the following general form:
\begin{equation}
  Z_b (X;\tau) \; = \; q^{\Delta_b} \big( c_0^{(b)} + c_1^{(b)} q + c_2^{(b)} q^2 + \ldots \big) \,, \qquad q= e^{2\pi i \tau}
\end{equation}
Namely, they depend on the choice of the 4-manifold $X$, a compact Lie group $G$, the variable $\tau$ with values in the upper-half plane, $\tau \in \mathbb H$, and the extra ``decoration'' data on the 4-manifold $b \in H^2 (X; \pi_1 (G))$. The physical definition of the invariants exists, at least in principle, for general 4-manifolds. However, the corresponding moduli spaces turn out to be non-compact and, as a result, rigorous mathematical definitions are currently limited to K\"ahler surfaces, see {\it e.g.} \cite{Dijkgraaf:1997ce,Gholampour:2017bxh,Tanaka:2017bcw} for some recent work in this direction.

The correspondence explored in this paper can be considered as a 3-dimensional analogue of the long-studied relation between Vafa-Witten invariants and VOA characters \cite{Cheng:2018vpl}. With its roots in the 3d-3d correspondence, the 3-manifold analogue of the Vafa-Witten $q$-series invariant, called $\widehat Z^G_{\underline{\vec{b}}} (X;\tau)$\footnote{We will refer to it as the $\widehat Z$-invariant, $q$-series invariant, BPS invariant, and homological block interchangeably.}, also depends on the choice of gauge group $G$, 3-manifold $X$, variable $q = e^{2\pi i \tau}$ in the unit disk, $|q|<1$, and extra data ${\underline{\vec{b}}}$ given by the generalised {Spin$^c$}-structure. In part due to the fact that topology of 3-manifolds is much simpler than the topology of 4-manifolds, the $q$-series invariants in dimension three are easier to define and compute. For example, $\widehat Z^G_{\underline{\vec{b}}} (X,\tau)$ can be defined \cite{Gukov:2020lqm} via Rozansky-Witten theory \cite{Rozansky:1996bq} based on affine Grassmannians, and Rozansky-Witten theory has rigorous mathematical definitions due to Kapranov \cite{MR1671737} and Kontsevich \cite{MR1671725}.

Alternatively, $\widehat Z^G_{\underline{\vec{b}}} (X;\tau)$ for closed 3-manifolds can be defined via surgery techniques \cite{Gukov:2019mnk}, which require invariants of knot (link) complements as input. The latter, in turn, can be obtained via a very efficient $R$-matrix approach \cite{Park:2020edg,Park:2021ufu} that involves quantum groups at generic $q$ and Verma modules. It allows one to compute $F_K (x,\tau) := \widehat Z (S^3 \setminus K; \tau)$ for many infinite families of knots and links, so that explicit expressions for $F_K (x,\tau)$ are now available for all knots up to 8 crossings and many knots with 8 and 9 crossings, for $G=SU(2)$.\footnote{The difficulties start at around 9 crossings \cite{Costantino:2021yfd}, where one needs to use combination of different methods in order to obtain  $\widehat Z (S^3 \setminus K; \tau)$, {\it e.g.} known relations to other 3-manifold invariants labelled by Spin or Spin$^c$ structures.} For example, for $G=SU(2)$ and $X = - S^3_{+5} ({\bf 10_{145}})$:
\be
\begin{matrix}
  b=2: &~& q^{14/5} \left( -1 +2 q +2 q^2 + q^3 +\ldots \right) \\
  b=1: &~&  q^{11/5} \left( -1 -2 q^2 -2 q^3 -4 q^4
  +\ldots \right) \\
    b=0: &~& 2 q^4 +2 q^7 +2 q^8 +2 q^9 +4 q^{10} +\ldots \\
    b=-1: &~& q^{11/5} \left( -1 -2 q^2 -2 q^3 -4 q^4
    +\ldots \right) \\
      b=-2: &~& q^{14/5} \left( -1 +2 q +2 q^2 + q^3 +\ldots \right)
    \end{matrix}
    \ee

    These techniques provide much more data for the explicit form of $\widehat Z^G_{\underline{\vec{b}}} (X;\tau)$ than we can hope to match with VOA characters. In fact, starting with rather simple classes of surgeries, such as Seifert 3-manifolds or plumbed 3-manifolds, our main goal will be to identify the $\widehat Z$ 
    \begin{equation}
      \widehat Z^G_{\underline{\vec{b}}} (X;\tau) \; \sim  \; \chi_{\underline{\vec{b}}} (\tau)
      \label{ZhatVOA}
    \end{equation}
    with characters of vertex algebras, when overall powers of $q$ and $\eta(\tau)$ are ignored. 
    Compared to the four-dimensional version mentioned earlier, one interesting feature of this relation for 3-manifolds is that in most examples the vertex algebra is logarithmic. We will refer to them as logarithmic vertex operator algebras (log VOAs). In particular, the algebras we discuss in the paper all have irreducible but indecomposable modules. For a review, see for example \cite{Creutzig:2013hma, Flohr2003, adamovic2012c2cofinite}. 

    \begin{center}
      \begin{table}
        \begin{tabular}{ccccc}
          \toprule
    & & triplet/singlet& $G=SU(2)$ & $G\neq SU(2)$  \\ \midrule
    \multirow{4}*{\parbox[c]{2cm}{\centering Three \\exceptional\\ fibers}} 
    &	{(Pseudo-)Spherical} & triplet & \cmark  & \cmark\\ 
    &$D=1$&singlet&\cmark & \cmark\\\cmidrule{2-5}
    &{Non Pseudo-Spherical} & triplet & $\bigcirc$ &$\bigcirc$ \\ 
    &$D>1$&singlet &  Cor \ref{cor:A1A2}   &  $\bigcirc$ \\\midrule
  \multirow{2}*{\parbox[c]{2cm}{\centering Four\\ exceptional\\ fibers}} &	\multirow{2}*{\parbox[c]{1.5cm}{\vspace{9pt}Spherical}}  & \parbox[c]{1.5cm}{\centering\vspace{4.5pt}triplet} & \parbox[c]{1.5cm}{\centering\vspace{4.5pt}Thm \ref{thm:4fibchimatching}} & \parbox[c]{0.5cm}{\centering\vspace{4.5pt}\,?}  \\[9pt] 
                                                                                                         &	 & \parbox[c]{1.5cm}{\centering \ singlet \vspace{4.5pt}}  &\parbox[c]{1.5cm}{\centering \ Cor \ref{cor:4fib} \vspace{4.5pt}}   &\parbox[c]{0.5cm}{\centering \ ? \vspace{4.5pt}}  
                                                                                                         \\\bottomrule
        \end{tabular}
        \caption{\label{tab:3sinfibers}The relation between $\widehat Z^G_{\underline{\vec{b}}}(X_\Gamma)$,  its integrands $\tilde \chi_{\underline{\vec{b}}}$, and the triplet and singlet characters. See Theorem \ref{3fiber_sphere} and Corollary \ref{cor:psudo} for the entries with \cmark, Theorem  \ref{thm:combining_into_characters} for the entries with $\bigcirc$.
        }
        \label{tab:scope}
      \end{table}
    \end{center}

    Specifically, the results of this paper mainly focus on Seifert manifolds with three or four exceptional fibers. 
    Apart from studying the relation \eqref{ZhatVOA}  between $\widehat Z^G_{\underline{\vec{b}}} (X)$ and VOA characters, we also explore a more refined relation among the {\em integrand} of the contour integral leading to $\widehat Z^G_{\underline{\vec{b}}} (X)$, and the ``triplet" type VOA, containing the ``singlet" type VOA as a subalgebra. Namely, we have 
    \be\label{IntegrandVOA}
    \widehat Z^G_{\underline{\vec{b}}} (X;\tau) \; \sim  \; \int_{\cal C} d\vec\xi \, \tilde \chi_{\underline{\vec{b}}} (\tau,\vec\xi\,) , ~~~~ \tilde \chi_{\underline{\vec{b}}} (\tau,\vec\xi\,) \sim \chi_{\underline{\vec{b}}} (\tau,\vec\xi\,)
    \ee 
    and $\chi_{\underline{\vec{b}}} (\tau,\vec\xi\,)$ is given by characters of certain triplet vertex algebras. 
    The results are summarised in Table \ref{tab:scope}. 
    We say a Seifert manifold with $N$ exceptional fibers, with Seifert data 
    $X_\Gamma =M(b;\{q_i/p_i\}_{i=1,\dots,N})$, is {\em pseudo-spherical} if 
    \be
    {1\over \textfrak{e}\, p_i} \in \ZZ ~{\rm for\, all}~i=1,\dots,N, 
    \ee
    where $\textfrak{e} =	b +\sum_k {q_k\over p_k} $ is the orbifold Euler characteristic, which is specifically always the case when $X_\Gamma$ is an integral homological sphere.  In the present work we focus on the cases of {\em negative Seifert} manifolds, namely those with $\textfrak{e} <0$. 

    \bigskip
    To summarize, in the present paper we show the following.  
    \begin{itemize}
      \item Let $X$ be any pseudo-spherical negative Seifert manifold with three exceptional fibers, $G$ be any choice of simply-laced Lie group, and ${\underline{\vec{b}}}$ be any choice of the generalised {Spin$^c$} structure. Then the integrand  $\tilde \chi_{\underline{\vec{b}}}$ of the three-manifold invariant $\widehat Z^G_{\underline{\vec{b}}} (X,\tau)$, up to  overall powers of $\eta(\tau)$ and $q$,  is given by a virtual generalised character of the so-called triplet vertex algebra corresponding to the Lie algebra ${\mathfrak g}$, reviewed in \S\ref{subsec:1m}. 

        As a consequence, again up to  overall powers of $\eta(\tau)$ and $q$,  the three-manifold invariant $\widehat Z^G_{\underline{\vec{b}}} (X;\tau)$ is given by a virtual generalised character\footnote{Following the commonly used definition of virtual characters, we say that an integral linear combination of generalised characters is a virtual generalised character. } of the so-called singlet vertex algebras corresponding to the Lie algebra ${\mathfrak g}$. 

        The above corresponds to the entries with \cmark in Table \ref{tab:scope} and precised in 
        Theorem \ref{3fiber_sphere} and Corollary \ref{cor:psudo}. 

      \item 
        Let $X$ be any negative Seifert manifold with three exceptional fibers, $G$ be any choice of Lie group, and ${\underline{\vec{b}}}$ be any choice of the generalised {Spin$^c$} structure. 
        Then there is a sum over the generalised {Spin$^c$} structures including ${\underline{\vec{b}}}$, such that the corresponding sum of the integrand $\chi_{\underline{\vec{b}}} (\tau,\vec\xi\,)$, up to an overall power of $\eta(\tau)$ and $q$ as well as a rescaling of  $\tau$,  is given by a virtual generalised character of the so-called triplet vertex algebras corresponding to the Lie algebra ${\mathfrak g}$. 

        As a consequence, again up to an overall power of $\eta(\tau)$ and $q$ as well as a rescaling of $\tau$, the corresponding sum of three-manifold invariants $\widehat Z^G_{\underline{\vec{b}}} (X;\tau)$ is given by a virtual generalised character of the so-called singlet vertex algebras corresponding to the Lie algebra ${\mathfrak g}$. 

        The above corresponds to the entries with $\bigcirc$ in Table \ref{tab:scope} and precised in 
        Theorem   \ref{thm:combining_into_characters}. 

      \item 
        Let $X$ be any negative Seifert manifold with four exceptional fibers that is an integral homological sphere, $G=SU(2)$, and ${\underline{\vec{b}}}$ be any choice of the  {Spin$^c$} structure. Then the integrand  $\chi_{\underline{\vec{b}}} (\tau,\vec\xi\,)$ of the three-manifold invariant $\widehat Z^G_{\underline{\vec{b}}} (X;\tau)$, up to overall multiplicative constants and powers of $\eta(\tau)$ and $q$,  is given by an integral linear combination of generalised characters of the so-called $(p,p')$ triplet vertex algebras corresponding to the Lie algebra ${\mathfrak g}$, reviewed in \S\ref{sec:pp'}. 

        As a consequence, again up to  overall powers of $\eta(\tau)$ and $q$,  three-manifold invariant $\widehat Z^G_{\underline{\vec{b}}} (X;\tau)$ is given by a virtual generalised character of the so-called $(p,p')$ singlet vertex algebras. 

        The precise version of the above is the content of Theorem  \ref{thm:4fibchimatching} and Corollary \ref{cor:4fib}. 

      \item Following the consideration of \S4 of \cite{Gukov:2017kmk}, we also investigate the effect of including Wilson operators in the theory on the relation between the BPS partition function $\widehat Z$ and the characters of log VOAs. 
        We found that the relation continues to exist in the presence of Wilson operators, but gets modified in ways that depend on to which node the Wilson operator is associated to. 

        Upon the including of a Wilson operator associated with an end node, Theorem \ref{thm:combining_into_characters} and as a result Theorem \ref{3fiber_sphere} and Corollary \ref{cor:A1A2} continue to hold, as well as Theorem \ref{thm:4fibchimatching}, with a modification of parameters that is given by \eqref{dfn:AW}. 
        When the Wilson operator is associated with the central node, Corollary  \ref{cor:psudo} gets modified into \eqref{eqn:shift_wilson}, and similarly for Corollary \ref{cor:A1A2}. This 
        ``shifting" phenomenon  has been observed for the special Lens space example  in \cite{Gukov:2017kmk}. The relation undergoes a more drastic modification when the added Wilson operator is associated to an intermediate node (a vertex in the plumbing graph with two other vertices connected to it). Namely, the statements in Corollaries  \ref{cor:psudo}-\ref{cor:A1A2} are modified into \eqref{eqn:intermediate_Wilson}, and homological blocks are no longer given by a virtual VOA character
        up to an overall multiplicative constant and powers of $\eta(\tau)$ and $q$. Instead, they are given by a virtual generalised character of a log VOA, each modified by an individual rational $q$-power,  up to an overall multiplicative constant and powers of $\eta(\tau)$. 

      \item Until recently, various ways\footnote{{\it e.g.} based on the combination of surgery and $R$-matrix techniques mentioned earlier} of computing $\widehat{Z}$-invariants would typically produce an explicit form of the $q$-series up to any desired order in $q$, but not a closed form expression. This makes the study of modular properties and other related questions quite challenging in general. Recent insights from enumerative geometry and the knot-quiver correspondence provide a new and surprising solution to this problem, which we discuss in \S \ref{sec:fermionic} and expect to be a powerful tool in the future work on $\widehat{Z}$-invariants. In particular, the closed form expressions produced by a version of the knot-quiver correspondence are perfectly suited for identifying the spectrum of quasiparticles in integrable massive deformations of the 2d logarithmic CFTs.

    \end{itemize}

    \newpage
    \section{Notation Guide}
    \begin{footnotesize} 
      \begin{list}{}{ \itemsep -1pt \labelwidth 23ex \leftmargin 13ex } 

      \item[$\Lambda$] The root lattice associated to the simply-laced Lie algebra $\mathfrak{g}$. 

      \item[$\Lambda^\vee$] The dual root lattice.

      \item[$\bar \Lambda$] The rescaled root lattice. 

      \item[$\Phi_{\rm s}$] The set of simple roots $\{\vec\a_i\}_{i=1}^{\text{rank}G}$. 

      \item[$\Phi_\pm$] The set of positive and negative roots. 

      \item[$\vec{\rho}$] The Weyl vector of the root system, $\vec{\rho}:= \frac{1}{2} \sum_{\vec{\a}\in \Phi_+} \vec{\a}$.

      \item[$\langle \cdot , \cdot \rangle$] The scalar product in the dual space of the Cartan subalgebra.

      \item[$\{\vec \omega_i\}_{i=1}^{\text{rank}G}$] The set of fundamental weights, satisfying $\langle\vec \omega_i, \vec\alpha_j \rangle = \delta_{i,j}$.

      \item[$P^+$] The set of dominant integral weights. See \eqref{eqn:P+}.

      \item[$W$] The Weyl group of the root system.   

      \item[$w(\cdot)$] The action of the element $w\in W$. The length of $w$ is denoted as $l(w)$.

      \item[$c$] The central charge of a Virasoro algebra.   

      \item[log-$\mathcal{V}_{\bar \Lambda}(m)$] The logarithmic VOA associated to $\bar \Lambda$, also known as (1,m) Log VOA. See \S \ref{subsec:1m}.

      \item[\lVs] The charge zero subsector of \lV. See \eqref{eqn:lVs}.

      \item[$\mathcal F_{\vec{\lambda}}$] The Fock modules of the Heisenberg algebra generated by the vertex operators $ V_{ {\vec{\lambda}}} (z)$.

      \item[$\mathcal Y_{\vec{\lambda}'}$] The lattice VOA  ${\cal V}_{\bar \Lambda}$ irreducible modules, which can be decomposed as
        $\mathcal Y_{\vec{\lambda}'}= \bigoplus_{\vec{\a}\in \Lambda} \mathcal F_{\vec{\lambda}'+\vec{\a}}\,.$ See \S \ref{subsec:1m}.

      \item[\lVb] The logarithmic VOA associated to $\bar \Lambda$, also known as (p,p') Log VOA. See \S \ref{sec:pp'}.

      \item[\lVbs] The charge zero subsector of \lVb. See \eqref{eqn:lVbs}.

      \item[$\mathcal F_{r,s;n}$] The Fock modules of the Heisenberg algebra generated by the vertex operators $V_{r,s;n}(z)$. For brevity $\mathcal F_0:= \mathcal F_{1,1;0}$.

      \item[$\mathcal Y_{r,s}^{\pm}$] The lattice VOA $\mathcal{V}_{\bar \Lambda}$ irreducible modules with $1\leq r\leq p$ and $1\leq s\leq p'$. 
        They can be decomposed into Fock modules $\mathcal F_{r,s;n}$ as
        $\mathcal Y_{r,s}^+= \bigoplus_{n\in \ZZ} \mathcal F_{p-r,p'-s;2n},~~\mathcal Y_{r,s}^-= \bigoplus_{n\in \ZZ} \mathcal F_{p-r,p'-s;2n+1}.$ See \S \ref{sec:pp'}.

      \item[$\mathcal K_{r,s}^{\pm}$] Virasoro modules defined in \eqref{eqn:Krs}.

      \item[$\mathcal J_{r,s;n}$] The irreducible Virasoro module of highest weight $\Delta_{r,s;n}$.

      \item[$\mathcal X_{r,s}^{\pm}$] Virasoro modules defined in \eqref{eqn:Xrs}.

      \item[$\Delta(\vec{\xi})$] The Weyl denominator of the Lie algebra ${\mathfrak g}$. See \eqref{eqn:denom_id}.

      \item[$\chi_{\vec{\lambda}}$, $\chi^0_{\vec{\lambda}}$] Characters of modules labelled by ${\vec{\lambda}}$ of \lV and \lVs respectively. See \S \ref{subsec:1mchar}.

      \item[ch$^{\pm}$, ch$^{0,\pm}$] Characters of modules of \lVb and \lVbs respectively. See \S \ref{subsec:pp'char}.

      \item[$\Psi_{m,r}(\tau)$] The holomorhpic Eichler integral of unary theta functions of weight 3/2. See \eqref{eqn:false}.

      \item[$\Phi_{m,r}(\tau)$] The holomorhpic Eichler integral of unary theta functions of weight 1/2. See \eqref{eqn:false2}.

      \item[$\Gamma$] The plumbing graph associated to the plumbed manifold $X_{\Gamma}$ with vertex set $V$. 

      \item[$v_0$] The central node of a plumbing graph $\Gamma$.

      \item[$M$] The adjacency matrix of $\Gamma$. The number of positive eigenvalues of $M$ is $\pi_M$ and its signature is given by $\sigma_M$.

      \item[$\Gamma_{M,G}$] The lattice defined as  $\Gamma_{M,G}:=M\ZZ^{|V|}\otimes_{\ZZ} \Lambda$ with norm given as in (\ref{def:notation}).

      \item[$\textfrak{e}$] The orbifold Euler characteristic of $X_\Gamma =M(b;\{q_i/p_i\}_{i})$, given by $ \textfrak{e} = 	b +\sum_k {q_k\over p_k}$.

      \item[$D$] The smallest positive integer such that ${D\over \textfrak{e} p_i}\in \ZZ$ for $i=1,\dots,N$ associated to an $N$-leg star graph corresponding to a negative Seifert manifold $X_\Gamma =M(b;\{q_i/p_i\}_{i})$.

      \item[$m^{(\vec \nu)}_{\vec\sigma}$] The multiplicity of the weight $\vec\sigma$ in the highest weight module with highest weight $\vec \nu$. See \eqref{LieAlgcharacter}.

      \item[$\underline{\vec{x}}$] An element in $\mathbb{R}^{|V|}  \otimes_\ZZ\Lambda$, whose norm squared is given by 
        $|| \underline{\vec{x}} ||^2 := \sum_{v,v'\in V} M^{-1}_{v,v'} \langle \vec{x}_v , \vec{x}_{v'} \rangle$.

      \item[${\underline{\vec{b}}}$] An element in $(\ZZ^{|V|}\otimes_{\ZZ}\Lambda+\underline{\vec{b}_{0}} )  /  \Gamma_{M,G}$, corresponding to a generalized {Spin$^c$} structure. 

      \item[ $\widehat Z^G_{\underline{\vec{b}}}(X_\Gamma)$] The homological block, a quantum invariant of $X_\Gamma$ associated to a simply-laced Lie group $G$. See Definition \ref{dfn:blocks}. 

      \item[ $\widehat{Z}^G_{\underline{\vec{b}}}(X_\Gamma, W_{\vec \nu_{v_\ast}})$ ] The homological blocks, modified by Wilson operators $W_{\vec \nu_{v_\ast}}$ associated to a  node $v_\ast \in V$ in the plumbing graph and highest weight representations with highest weight $\vec \nu \in \Lambda^\vee$. See equation \eqref{dfn:ZhatW}.

      \item[$\tilde \chi_{{\hat{w}};\underline{\vec{b}}}$] Integrand in the definition of $\widehat Z^G_{\underline{\vec{b}}}(X_\Gamma)$, $\hat{w}:=(w_1,\dots, w_N)$, ${\underline{\vec{b}}}\in (\ZZ^{|V|}\otimes_{\ZZ}\Lambda+\underline{\vec{b}_{0}} )  /  \Gamma_{M,G}$. See \ref{prop:main_generalSeif}.

      \item[$C_\Gamma^G(q)$] The constant factor in the definition of $\widehat Z^G_{\underline{\vec{b}}}(X_\Gamma)$, defined as $C_\Gamma^G(q):= (-1)^{|\Phi_+|\pi_M}q^{{3\sigma_M-{\rm Tr} M\over 2} |\vec \rho|^2}$. See equation \eqref{eqn:C}.

      \item[${S}_{w,w_1,w_2,\dots,w_N;\vec{\underline{b}}}$] The set given by
        $ \{ \vec \ell_0 \, \lvert (\vec{\ell}_0, \,{-w_1(\vec \rho), \cdots, -w_N( \vec\rho)} , {0, \cdots, 0}) \in \Gamma_{M,G} + w(\underline{\vec b}) \}$. See \eqref{dfn:theset}.

      \end{list}
    \end{footnotesize}

    \newpage
    \section{Various Log VOAs and Their Characters}
    In this section we will briefly review the  logarithmic vertex operator algebras relevant for our study of homological blocks and in particular their characters.

    We take $G $ to be a simply-laced Lie group, use ${\mathfrak g}$ to denote the associated Lie algebra and let $\Lambda=\Lambda_G$ be the corresponding root lattice. 
    We will denote by $\Phi_{\rm s}=\{\vec \alpha_i\}$ a set of simple roots and $\{\vec \omega_i\}$ the corresponding 
    fundamental weights,
    $\Phi_\pm$ the set of positive resp. negative roots, 
    and by 
    \be
    \label{eqn:P+}
    P^+:=\{\vec \lambda \in \Lambda^\vee 
      \lvert \langle \vec \lambda, \vec \alpha\rangle {>}  0  ~\forall~\vec\alpha \in \Phi_+ 
    \}
    \ee
    the set of dominant integral weights, where $\langle \cdot , \cdot \rangle$ is a quadratic form given by the Cartan matrix of $G$.   For $\vec{x} \in \mathbb{C} \otimes_\mathbb{Z} \Lambda $, we define the norm 
    $| \vec{x} |^2 := \langle \vec{x}, \vec{x} \rangle$ as usual. 

    In \S \ref{subsec:1m} we review the VOAs {\lV}   and \lVs, associated with $\mathfrak g$. They are also often referred to as the triplet and singlet $(1,m)$ log VOAs, respectively. In \S \ref{sec:pp'} we review the VOAs {\lVb} and \lVbs, associated with $\mathfrak g=A_1$. They are often referred to as the triplet and singlet $(p,p')$ log VOAs, respectively.

    \subsection{ log-\texorpdfstring{${\cal V}_{\bar \Lambda}(m)$}{V\_L(m)}}
    \label{subsec:1m}

    The logarithmic vertex operator algebra log-${\cal V}_{\bar \Lambda_{A_1}}(m)$, also known as the triplet model and sometimes denoted as ${\cal W}$(m) in the literature, was first constructed in \cite{Kausch:1990vg,Kausch:1995py}. 
    The analogous algebras are later defined for arbitrary simple-laced semi-simple Lie algebra ${\mathfrak g}$ in \cite{FeiginTipunin}. 
    We mainly follow \cite{FeiginTipunin,Sugimoto}.

    Let $\varphi_{\vec{\a}_i}(z)$ be the chiral scalar field associated to the root $\vec{\a}_i\in \Phi_{\rm s}$, 
    satisfying the following operator product expansion
    \be
    \varphi_{\vec{\a}}(z)\varphi_{\vec{\b}}(w) =  \langle \vec{\a}, \vec{\b} \rangle \log (z-w)\,.
    \ee
    In terms of the mode expansion
    \be
    \varphi_{\vec{\a}_i}(z)= (\bar{\varphi}_{\vec{\a}_i})_0 + (\bar{\varphi}_{\vec{\a}_i})_0 \log (z) - \sum_{n \neq 0} ({\varphi}_{\vec{\a}_i})_n z^{-n}, 
    \ee
    we modify the commutation rule of the zero modes such that 
    \be
    [(\bar{\varphi}_{\vec{\a}_i})_0,(\bar{\varphi}_{\vec{\a}_j})_0]=b_{ij}
    \ee
    where 
    \be
    b_{ij}=-b_{ji}=\begin{cases} 1 & i<j ~{\rm and}~ C_{ij} \neq 0\\ 0 & \text{otherwise}\end{cases}.
    \ee

    The vertex operators are defined as
    \be
    V_{\vec{\lambda}}(z):= e^{\frac{1}{\sqrt{m}}\varphi_{\vec{\lambda}}(z)}
    \ee
    with $\vec{\lambda} \in \Lambda^{\vee}$.

    The lattice VOA  ${\cal V}_{\bar \Lambda}$ can be constructed directly from these fields. The irreducible modules of this VOA are specified by an element $\vec \lambda' \in \bar \Lambda^\vee/ \bar \Lambda$, where $\bar \Lambda := \sqrt{m} \Lambda$. It is convenient to decompose $\vec \lambda'$ into
    the following two parts \cite{FeiginTipunin}: 
    \begin{equation}\label{eqn:decompose_lambda}
      \vec \lambda' = \sqrt{m} \vec \lambda +  \vec \mu  
    \end{equation}
    where $\vec \lambda \in \Lambda^\vee/\Lambda$ and  
    \begin{equation}\label{eqn:mu}
      \vec \mu = {1\over \sqrt{m}}\sum_{i=1}^{\text{rank}G} (1-s_i)\vec \omega_i 
    \end{equation}
    for $s_i\in \{1,\dots,m\}$. 

    The ${\cal V}_{\bar \Lambda}$ irreducible modules can be written in terms of the Fock module $\mathcal F_{\vec{\lambda}}$, corresponding to the vertex operator $ V_{ {\vec{\lambda}}} (z)$, as 
    $$\mathcal Y_{\vec{\lambda}'}= \bigoplus_{\vec{\a}\in \Lambda} \mathcal F_{\vec{\lambda}'+\vec{\a}}\,.$$

    Now we choose
    an energy momentum tensor 
    \be T(z):= {1\over 2} C^{ij}\partial \varphi_{\vec{\a}_i}(z) \partial \varphi_{\vec{\a}_j}(z) + {Q_0} \partial ^2\varphi_{\vec{\rho}}(z),\ee
    where  
    \be
    Q_0 = \sqrt{m} - \frac{1}{\sqrt{m}},
    \ee
    $\vec{\rho}$ is the Weyl vector, and $C^{ij}$ denotes the $ij$-entry of the inverse of the Cartan matrix. 
    The Virasoro algebra has thus central charge $$ c= \text{rank}G + 12|\vec \rho|^2 \biggl(2-m- {1 \over m}\biggr),$$ and the vertex operator $V_{ {\vec{\lambda}}} (z)$ with $\lambda \in \bar{\Lambda}^\vee$ has conformal dimension 
    $$\Delta_{\vec{\lambda}} = \frac{1}{2} | \vec{\lambda} - Q_0\vec{\rho}|^2 + \frac{c-\text{rank}G}{24}.$$

    The log-${\cal V}_{\bar \Lambda}(m)$ algebra can be described as a subalgebra of ${\cal V}_{\bar \Lambda}$ by considering the intersection of 
    the kernels of the screening charges in the original lattice VOA.

    Let 
    \be
    e_i := \frac{1}{2\pi i} \oint dz e^{\sqrt{m}\varphi_{\vec{\a}_i}(z)}, \quad f_i := \frac{1}{2\pi i} \oint dz e^{-\frac{1}{\sqrt{m}}\varphi_{\vec{\a}_i}(z)}
    \ee
    be the screening operators. They commute with the energy momentum tensor and in addition $e_i$ commutes with $f_j$ for $i,j = 1, \dots, \text{rank}G$.

    The log-${\cal V}_{\bar \Lambda}(m)$ algebra is the vertex operator subalgebra of ${\cal V}_{\bar \Lambda}$ defined by \cite{FeiginTipunin, Sugimoto}
    \be
    \label{eqn:lV}
    \text{log-}{\cal V}_{\bar \Lambda}(m) := \bigcap\limits_{i=1}^{\text{rank}G} \text{ker}_{{\cal V}_{\bar \Lambda}} f_{i}\,,
    \ee
    Instead of taking the kernel over the whole lattice algebra ${\cal V}_{\bar \Lambda}$, in order to define the $\text{log-}{\cal V}_{\bar \Lambda}^0(m)$ ``singlet" algebra, we restrict to the charge zero subalgebra $\mathcal F_{0}$
    \be
    \label{eqn:lVs}
    \text{log-}{\cal V}_{\bar \Lambda}^0(m) := \bigcap\limits_{i=1}^{\text{rank}G} \text{ker}_{ \mathcal F_{0}} f_{i}\,.
    \ee

    \subsubsection{Characters}
    \label{subsec:1mchar}
    As described in the previous section, a module $\mathcal X_{\vec\lambda'}$ for log-${\cal V}_{\bar \Lambda}(m)$ is  specified by an element $\vec \lambda' \in \bar \Lambda^\vee/ \bar \Lambda$, parametrized as in  \eqref{eqn:decompose_lambda}.
    Let 
    \begin{equation}\label{eqn:denom_id}
      \Delta(\vec{\xi})=  \ex^{\langle \vec{\xi},\vec \rho \rangle} \prod_{\vec\alpha \in \Phi_-} (1-\ex^{\langle \vec{\xi},\vec \alpha \rangle}) = \sum_{w\in W} (-1)^{l(w)} \ex^{\langle \vec{\xi}, w(\vec \rho) \rangle}
    \end{equation}
    be the usual Weyl denominator of the Lie algebra ${\mathfrak g}$, where $l(w)$ is the length of $w$ in the Weyl group $W$.
    With an abuse of notation of using $\vec \lambda$ below to denote any arbitrary representative of the  $\vec \lambda \in \Lambda^\vee/\Lambda$, the character of $\mathcal X_{\vec\lambda'}$ 
    is given as follows \cite{FeiginTipunin}: 
    \begin{gather}
      \label{eqn:char}
      \begin{split}
        \chi_{\vec{\lambda}'}(\tau, \vec{\xi}) &= \frac{1}{\eta^{\text{rank}G}(\tau)}{1\over\Delta(\vec\xi) } \sum_{\vec{\tilde{\lambda}} - \vec{\lambda} \in \Lambda} q^{\frac{1}{2} | \sqrt{m}\vec{\tilde{\lambda}} + \vec {\mu} +Q_0 \vec{\rho}|^2}\left( {\sum_{w \in \mathcal{W}} (-1)^{l(w)} \ex^{\langle  \vec{\xi}, w(\vec{\rho} + \vec{\tilde{\lambda}}) \rangle}}\right)\\
                                                         &=  \frac{1}{\eta^{\text{rank}G}(\tau)}\sum_{\substack{\vec{\tilde{\lambda}} - \vec{\lambda} \in \Lambda \\{ \vec\rho +\vec{\tilde{\lambda}} \in P^+}}} \chi^{\mathfrak g}_{\vec{\tilde{\lambda}}}(\vec\xi)  
                                                         \sum_{w\in W}(-1)^{l(w)} q^{\frac{1}{2} | \vec \mu +\sqrt{m}w ( \vec{\tilde{\lambda}} +\vec\rho )-{1\over \sqrt{m}} \vec{\rho}|^2} . 
      \end{split}
    \end{gather}

    Note that, to go from the first line to the second line in (\ref{eqn:char}), we have used $w(\vec \lambda) \equiv \vec \lambda \,{\rm mod}\, \Lambda$ if $\vec \lambda \in \Lambda^\vee$, and the fact that:
    \be
    {\sum_{w \in W} (-1)^{l(w)} \ex^{\langle  \vec{\xi}, w(\vec{\rho} + \vec{{\lambda}}) \rangle}} =   0 
    \ee
    if $\vec{\rho} + \vec{{\lambda}}$ lies on the boundary of a Weyl chamber.

    We have also used the Weyl character formula to write the multiplicity function ${\rm dim}(V_{\vec \lambda}(\vec \beta)) $ in a highest weight module of the Lie algebra $\mathfrak g$: 
    \be\label{highest_weight_module}
    \chi^{\mathfrak g}_{\vec\lambda}(\vec\xi)= \frac{\sum_{w\in W} (-1)^{l(w)} \ex^{\langle \vec{\xi}, w(\vec\lambda +\vec \rho) \rangle}}{\Delta(\vec{\xi})} = \sum_{\vec \beta\in \Lambda^\vee} {\rm dim}(V_{\vec \lambda}(\vec \beta)) \ex^{\langle \vec{\xi}, \vec \beta \rangle}.
    \ee

    For later use, we will also introduce the {\em generalised characters}, given by (\ref{eqn:char}) but with $\vec \mu \in \bar \Lambda^\vee$ lying outside the range indicated in (\ref{eqn:mu}). 

    Taking $\vec \xi =0$, we have 
    \be
    \chi_{\vec{\lambda}'}(\tau) =\chi_{\vec{\lambda}'}(\tau, \vec{\xi})\lvert_{\vec\xi=0} = \frac{1}{\eta^{\text{rank}G}(\tau)}\sum_{\substack{\vec{\tilde{\lambda}} - \vec{\lambda} \in \Lambda\\ \vec{\tilde{\lambda}}+\vec \rho  \in P^+}}   {\rm dim}(V_{\vec{\tilde{\lambda}} }) 
    \sum_{w\in W}
    (-1)^{l(w)} q^{\frac{1}{2} | \vec \mu +\sqrt{m}w ( \vec{\tilde{\lambda}}  +\vec\rho )-{1\over \sqrt{m}} \vec{\rho}|^2}. 
    \ee

    Another way to view the log-${\cal V}_{\bar \Lambda}(m)$ characters in (\ref{eqn:char}) is as a generating function for the singlet characters. 
    Assuming that the log-${\cal V}_{\bar \Lambda}(m)$ module $\mathcal X_{\vec\lambda'}$ is completely reducible as the module for the corresponding singlet model  (cf \S 5 of \cite{MR3624911}), we can write the log-${\cal V}_{\bar \Lambda}(m)$ character in terms of the 
    singlet characters $\chi^0_{\vec{\tilde \lambda}'}$, labelled by $\vec{\tilde \lambda}'  \in \bar \Lambda^\vee$, as 
    \be\label{singlet_triplet}
    \chi_{\vec{\lambda}'}(\tau, \vec{\xi}) =\sum_{\vec{\tilde{\lambda}} - \vec{\lambda} \in \Lambda}  \ex^{\langle  \vec{\tilde{\lambda}} , \vec \xi \rangle} \chi^0_{\vec{\tilde \lambda}'}(\tau) , 
    \ee 
    where we have written $\vec{\tilde \lambda}' = \sqrt{m} \vec{\tilde \lambda} + \vec \mu $ analogously to (\ref{eqn:decompose_lambda}).

    In particular, 
    the corresponding atypical singlet characters then read
    \begin{gather}\label{singlet_char}
      \begin{split}
        \chi_{\vec{\tilde \lambda}'}^0(\tau)
&={\rm CT}_{\vec z}\left(\ex^{-\langle  \vec{\tilde{\lambda}} , \vec \xi \rangle} \chi_{\vec{\lambda}'}(\tau, \vec{\xi})\right) 
 \\&= \frac{1}{\eta^{\text{rank}G}(\tau)}
 \sum_{\substack{\vec{\bar{\lambda}} + \vec\rho \in P^+\\ \vec{\bar{\lambda}}-\vec{\tilde \lambda} \in \Lambda  }}  {\rm dim}(V_{\vec{\bar{\lambda}}}(\vec{\tilde{\lambda}})) 
 \sum_{w\in W}
 (-1)^{l(w)} q^{\frac{1}{2} | \vec \mu +\sqrt{m}w ( \vec{\bar{\lambda}}  +\vec\rho )-{1\over \sqrt{m}} \vec{\rho}|^2}.
      \end{split}\end{gather}
      In the above, we have used ``${\rm CT}_{x}(f(x))$" to denote the ``$x^0$ (constant) terms of the polynomial $f(x)$ in $x$".  
      In what follows we will often use the notation $\vec\xi = \sum_i \xi_i \vec\alpha_i$,  $\quad z_i := \ex^{\langle \vec{\xi}, \vec\omega_i \rangle} = \ex^{\xi_i},$ and hence $e^{\langle \vec{\xi}, \vec{\alpha} \rangle} = \prod_{\vec{\alpha_i} \in \Phi_{\rm s}} z_i^{\langle \vec{\alpha_i}, \vec{\alpha}\rangle}$, and denote the corresponding vector by $\vec z$.
      In \eqref{singlet_char}, we say the left hand side is a generalised character of {log-${\cal V}_{\bar \Lambda}^0(m)$} when $ \chi_{\vec{\lambda}'}(\tau, \vec{\xi})$ is a generalised character of {log-${\cal V}_{\bar \Lambda}(m)$}. 

      \vspace{15pt}
      \noindent {\bf Example: $G=SU(2)$}\\

      We have
      \be
      \Delta(\vec \xi) =\Delta(\xi)= z-z^{-1},
      \ee
      where $\vec \xi = \xi \vec \alpha$ and $z=e^\xi$, 
      and the log-${\cal V}_{\bar \Lambda_{A_1}}(m)$ modules are labelled by
      \be
      \vec \lambda' =\begin{cases} {1-s\over 2\sqrt{m}} \vec \alpha\\( {\sqrt{m}\over 2}+ {1-s\over 2\sqrt{m}} )\vec\alpha \end{cases}\hspace{-0.4cm}, ~~~ s\in\{1,2,\dots,m\},
      \ee
      where we use $\vec \alpha$ to denote the simple root. Their characters are given by
      \begin{gather}\begin{split}\label{eqn:A1characters}
        \chi_{{1-s\over 2\sqrt{m}}\vec\alpha}
        (\tau,\xi) &= {1\over \eta(\tau)}
        \sum_{n\in \ZZ} q^{(-s+m+2mn)^2\over 4m} {z^{1+2n}-z^{-1-2n}\over {z-z^{-1}}}\\
        \chi_{{1-s+m\over 2\sqrt{m}}\vec\alpha}
        (\tau,\xi) &= {1\over \eta(\tau)}
        \sum_{n\in \ZZ} q^{(-s+2m+2mn)^2\over 4m} {z^{2+2n}-z^{-2-2n}\over {z-z^{-1}}}.
      \end{split}\end{gather}
      The corresponding singlet characters are 
      \begin{gather}\begin{split}
        \chi^{0}_{{1-s+2nm \over 2\sqrt{m}}\vec\alpha}(\tau) &={1\over \eta(\tau)}
        \left(\sum_{k \geq  |n| }q^{(-s+m+2mk)^2\over 4m}
          -\sum_{k \leq  -|n| }q^{(-s+m+2mk)^2\over 4m}
          +\delta_{n,0}q^{(-s+m)^2\over 4m}
        \right)\\
        \chi^{0}_{{1-s+(1+2n)m \over 2\sqrt{m}}\vec\alpha}(\tau) &={1\over \eta(\tau)}
        \left(\sum_{k \geq{\rm{max}}(n,-n-1) }q^{(-s+2m(1+k))^2\over 4m}
          -\sum_{k \leq{\rm{min}}(n,-n-1) }q^{(-s+2m(1+k))^2\over 4m}
        \right).
      \end{split}\end{gather}

      Taking the $z$-constant term in the first equation in \eqref{eqn:A1characters} we obtain
      \be\label{A1_outofrange}
      \eta(\tau)\chi^{0}_{{1-s \over 2\sqrt{m}}\vec\alpha}(\tau)  = \Psi_{m,m-s}(\tau) - \sum_{k\equiv m-s \xmod{2m}} ({\rm sgn}(k) - {\rm sgn}(k-m+s))\,q^{k^2\over 4m},
      \ee
      where 
      \be
      \label{eqn:false}
      \Psi_{m,r}(\tau) = \sum_{k\equiv r \xmod{2m}} {\rm sgn}(k)\,q^{k^2\over 4m}
      \ee
      is the false theta function \cite{MR557539}. 
      Namely, the above singlet characters are given by the false theta functions when multiplied by the eta function, and differ from the false theta function by a finite polynomial in $q$ in the case of generalised characters, for which $s$ does not lie in the range $s=1,2,\dots, m$.  

      Interestingly, note that terms of the form $q^{c}\over \eta(\tau)$ also have the interpretation as a character of the {\em typical} module character of the singlet $(1,m)$ model {log-${\cal V}_{\bar \Lambda}^0(m)$}. As a result, one can interpret the above identity  (\ref{A1_outofrange}) as the fact that a generalised $SU(2)$ singlet  character can always be expressed as an integral linear combination of the typical and atypical module characters.  
      We will see in (\ref{eqn:charID}) a somewhat similar phenomenon for the case $G=SU(3)$.

      \vspace{15pt}
      \noindent{\bf Example: $G=SU(3)$}
      \\

      The Weyl denominator is
      \begin{gather}\label{A2_Weyl} \begin{split}
        \Delta(\vec \xi) &= \Delta(\xi_1,\xi_2)   = z_1z_2(1-z_1^{-2}z_2)(1-z_2^{-2}z_1)(1-z_1^{-1}z_2^{-1}),
      \end{split}\end{gather} 
      and the modules are labelled by 
      \be\label{eqn:labelMods}
      \vec \lambda' = \begin{cases} \vec \lambda'_{0,0,s_1,s_2}:= (1-s_1)\frac{\vec \o_1}{\sqrt{m}} +(1-s_2)\frac{\vec \o_2}{\sqrt{m}} , & \\ \vec \lambda'_{1,0,s_1,s_2}:= (1-s_1+m)\frac{\vec \o_1}{\sqrt{m}} +(1-s_2)\frac{\vec \o_2}{\sqrt{m}} , & s_1,\, s_2\in\{1,2,\dots,m\} \\ \vec \lambda'_{0,1,s_1,s_2}:= (1-s_1)\frac{\vec \o_1}{\sqrt{m}} +(1-s_2+m)\frac{\vec \o_2}{\sqrt{m}} , &  \end{cases},
      \ee
      where $\vec \o_1, \vec \o_2$ denote the fundamental weights, and $\vec\a_1$, $\vec\alpha_2$ denote the
      roots of $A_2$.
      The characters are given by
      \begin{multline}
        \chi^{}_{\vec \lambda'_{a,b,s_1,s_2}} (\tau,\vec \xi) = \frac{1}{\eta(\tau)^2}\sum_{n_1,n_2\in \ZZ}q^{\frac{1}{3m}((mN_1-s_1)^2+(mN_2-s_2)^2+(mN_1-s_1)(mN_2-s_2))} \\\label{eqn:A2triplet}
        \times \frac{z_1^{N_1}z_2^{N_2}+z_1^{N_2}z_2^{-N_1-N_2}+z_1^{-N_1-N_2}z_2^{N_1} - (z_{1,2}\leftarrow  z_{2,1}^{-1}) }{\Delta(\vec \xi)}
      \end{multline}
      where we have written in the sum 
      \[
        N_1 = 2n_1-n_2 +a, ~ N_2 = 2n_2-n_1 +b. 
      \]
      For $SU(3)$ the multiplicity function is given by 
      \be\label{dim0}
      {\rm dim}(V_{n_1\vec \omega_1 +n_2\vec \omega_2}(0)) ={\rm min}(n_1,n_2) +1,
      \ee
      when $n_1\vec \omega_1 +n_2\vec \omega_2 \in \Lambda\cap P^+$, 
      and
      \be
      {\rm dim}(V_{n_1\vec \omega_1 +n_2\vec \omega_2}) ={1\over 2} (n_1+1)(n_2+1)(n_1+n_2+2).
      \ee
      Putting (\ref{dim0}) in (\ref{singlet_char}), we obtain the singlet characters 
      $\chi^0_{a,b;s_1,s_2} :=\chi^0_{\vec\lambda'_{a,b;s_1,s_2}} $. In particular, 
      for $\vec{\tilde \lambda}=0$ and 
      $\vec{\tilde \lambda}'  =   \vec \mu= {1\over \sqrt{m}}\sum_{i=1}^{2} (1-s_i)\vec \omega_i   $ we have
      \begin{gather}\label{eqn:chiMin}\begin{split}
        \eta^{2}(\tau)\chi^0_{0,0;s_1,s_2}(\tau) &=\sum_{\vec{\bar \lambda} \in \Lambda \cap P^+} {\rm{dim}}(V_{\bar \lambda}(0)) \sum_{w\in W} 
        (-1)^{l(w)} q^{{1\over 2m} |-s_1 \vec\omega_1-s_2 \vec\omega_2+m w(\vec{\bar \lambda} + \vec \rho)|^2}  \\
                                                           &  \sum_{w\in W} 
                                                           (-1)^{l(w)} q^{\langle  (s_1 \vec\omega_1+s_2 \vec\omega_2) -w  (s_1 \vec\omega_1+s_2 \vec\omega_2) ,   n_1 \vec\omega_1+n_2 \vec\omega_2 \rangle} \\
                                                                                                     &=   \sum_{\substack{n_1,n_2 \in {\mathbb N} \\n_1 \equiv n_2 \;(3)} } {\rm min}(n_1,n_2) q^{{1\over 3}m(n_1^2+n_1n_2+n_2^2)-{n_1(2s_1+s_2)+n_2(s_1+2s_2)\over 3} +{s_1^2+s_1s_2+s_2^2 \over 3m}} \\& 
                                                                                                     (1-q^{s_1 n_1}-q^{s_2 n_2}+q^{s_1 n_1 + n_2(s_1+s_2) } +q^{  n_1(s_1+s_2)+s_2 n_2 } - q^{(s_1+s_2) (n_1+n_2)}).
      \end{split}\end{gather}

      Interestingly, for $G=SU(3)$,  generalised singlet characters are integral linear combinations of the actual irreducible characters (and thus  virtual characters), via the identity (Lemma \ref{lem:charID})
      \begin{equation}\label{eqn:charID}
        \chi^{0}_{0,0;s_1,s_2+m}= 3 \chi^{0}_{1,0;s_1,s_2}
        - \chi^{0}_{0,0;m-s_1-s_2,s_2} + \chi^{0}_{0,0;s_1+s_2,m-s_2},
      \end{equation}
      which holds for all $s_1, s_2\in \ZZ$. This identity will be useful later when we establish the relation between singlet characters and $\widehat Z$ invariants, and its proof can be found in Appendix  \ref{charc_id}.

      \subsection{log-\texorpdfstring{${\cal V}_{\bar \Lambda}(p,p')$}{V\_L(p,p')}}\label{sec:pp'}

      Apart from the log-${\cal V}_{\bar \Lambda}(m)$ algebra, often referred to as the $(1,m)$ log VOA, one can consider a more general family of algebras. They are often referred to as the $(p,p')$  log VOA, labelled by two coprime integers $p$ and $p'$, and their definition reduces to the one for log-${\cal V}_{\bar \Lambda}(m)$ when setting $p=1$, $p'=m$. In the following we will focus on the case when $\Lambda=\Lambda_{A_1}$ is given by  $A_1$ root lattice. In this section, we mostly follow \cite{FGST}.\footnote{Note: our definition of $\varphi$ differs by $\sqrt 2$ from that of \cite{FGST}.} See also \cite{FGST2,BM1,Creutzig:2014nua,Ridout:2013pwa}. 

      Let $\bar \Lambda =\sqrt{pp'}\Lambda_{A_1} \cong \sqrt{2pp'}\ZZ$, and consider the lattice VOA ${\cal V}_{\bar \Lambda}$. It will turn out that the algebra {log-${\cal V}_{\bar \Lambda}(p,p')$} is a subalgebra of  ${\cal V}_{\bar \Lambda}$.
      For the rest of this section we restrict to the case of $\Lambda=A_1.$
      In this case we define a general vertex operator of $\mathcal{V}_{\bar \Lambda}(p,p')$ as,
      $$ V_{r,s;n}(z) := e^{{p'(1-r)-p(1-s) + pp'n\over\sqrt{pp'}}\varphi(z)}, ~ n\in \ZZ,~ 1\leq r \leq p, ~ 1\leq s \leq p'.$$

      The lattice VOA $\mathcal{V}_{\bar \Lambda}$ has $2pp'$ irreducible modules $\mathcal Y_{r,s}^\pm$ for $1\leq r\leq p$ and $1\leq~s\leq~p'$. They can be decomposed into Fock modules $\mathcal F_{r,s;n}$ corresponding to the vertex operator $V_{r,s;n}(z)$ as follows:
      $$\mathcal Y_{r,s}^+= \bigoplus_{n\in \ZZ} \mathcal F_{p-r,p'-s;2n},~~~\mathcal Y_{r,s}^-= \bigoplus_{n\in \ZZ} \mathcal F_{p-r,p'-s;2n+1},$$
      where $\mathcal Y^+_{p-1,p'-1}$ is the vacuum module, and 
      $$\mathcal{V}_{\bar \Lambda}(p,p')=\bigoplus_{r=1}^p\bigoplus_{s=1}^{p'}(\mathcal Y_{r,s}^+\oplus \mathcal Y_{r,s}^-).$$
      Let $\mathcal F_0:= \mathcal F_{1,1;0}$.
      Now we define
      \be
      \alpha_+:=\sqrt{p'\over p},\,~ \alpha_-:= -\sqrt{p\over p'},\, ~\alpha_0=\alpha_++\alpha_-,
      \ee
      and choose an energy momentum tensor 
      $$T(z):= {1\over 4} \partial \varphi(z) \partial \varphi(z) + {\alpha_0\over 2} \partial ^2\varphi(z),$$
      such that its modes span a Virasoro algebra with $$ c= 1- 6{(p-p')^2\over pp'}.$$ The general vertex operators $V_{r,s;n}$
      have conformal dimension
      $$\Delta_{r,s;n}:= {(ps-p'r+pp'n)^2- (p-p')^2\over 4pp'}$$
      with respect to this choice of $T(z)$. Finally, let $\mathcal J_{r,s;n}$ denote the irreducible Virasoro module of highest weight $\Delta_{r,s;n}.$

      In order to define the log-$\mathcal{V}_{\bar \Lambda}(p,p')$ VOA, we can start with the screening operators
      $$e_+:= \oint dz e^{\alpha_+\varphi(z)},~~f_-:= \oint dz e^{\alpha_-\varphi(z)},$$
      which commute with the energy-momentum tensor, $[e_+,T(z)]=[f_-,T(z)]=0.$ Let 
      \begin{align}
        \mathcal K_{r,s}^\pm= {\rm ker} e_+^s \cap {\rm ker} f_-^r ~~ {\rm in}~~\mathcal Y^{\pm}_{p-r,p'-s}  \label{eqn:Krs},
      \end{align}
      where $1 \leq r \leq p-1$ and $1 \leq s\leq p'-1$. Then the \lVb~  VOA is defined to be the subalgebra of $\mathcal{V}_{\bar \Lambda}$ with underlying vector space $\mathcal K_{1,1}^+.$
      It is strongly generated by the energy momentum tensor $T(z)$ and two Virasoro primaries $W^\pm(z)$ of conformal dimension $(2p-1)(2p'-1)$. In terms of screening charges and vertex operators, these primaries are given by
      \begin{equation}
        W^{+}\left( z \right)=\left( f_{-}^{p-1} \right)V_{1,p-1;3}\left( z \right),\quad 
        W^{-}\left( z \right)=\left( e_{+}^{p'-1} \right)V_{1,p'-1;3}\left( z \right).
      \end{equation}

      There are $2pp' + {1\over 2} (p-1)(p'-1)$ irreducible modules of the log-${\cal V}_{\bar \Lambda}(p,p')$ VOA. These come in two categories:
      \begin{itemize}
        \item ${1\over2}(p-1)(p'-1)$ number of Virasoro modules $\mathcal J_{r,s}:= \mathcal J_{r,s;0}$,  for $$(r,s) \in \{(r,s)| 1 \leq r\leq p-1,1\leq s\leq p'-1, p'r+ps\leq pp'\}.$$ 
          These modules for the $(p,p')$ Virasoro minimal model are also modules of the log-${\cal V}_{\bar \Lambda}(p,p')$ VOA and are annihilated by the maximal VOA ideal. 
        \item $2pp' $ number of  irreducible modules ${\mathcal X}_{r,s}^\pm$ for $1\leq r\leq p$ and $1\leq s \leq p'$. These can be described as 
          \begin{align}
      &\mathcal X_{r,s}^\pm= {\rm im} e_+^{p-r} \cap {\rm im} f_-^{p'-s} ~~ {\rm in}~~\mathcal Y^{\pm}_{p-r,p'-s} \notag\\
      &\mathcal X_{r,p'}^\pm= {\rm im} e_+^{p-r} ~~ {\rm in}~~\mathcal Y^{\mp}_{p-r,p'} \label{eqn:Xrs}\\
      &\mathcal X_{p,s}^\pm= {\rm im} f_-^{p'-s} ~~ {\rm in}~~\mathcal Y^{\mp}_{p,p'-s} \notag\\
      &\mathcal X_{p,p'}^\pm= \mathcal Y^{\pm}_{p,p'} \notag
          \end{align}
          for $1\leq r\leq p-1$ and $1\leq s \leq p'-1$.
      \end{itemize}

      The \lVb~ VOA admits an $sl(2, \mathbb{C})$ action which commutes with the Virasoro algebra generated by $T(z)$ and the currents $W^\pm (z)$ are highest and lowest-weight components of an $sl(2,\mathbb{C})$ triplet.  This is where the name ``triplet algebra" derives from. As Virasoro and $sl(2,\mathbb{C})$ bimodules, the ${\mathcal X}^\pm_{r,s}$ decompose as
      \bea
      {\mathcal X}^+_{r,s}&=&\bigoplus_{n\in \mathbb N} \mathcal J_{r,p'-s;2n-1} \otimes \ell_{2n-1} \label{eqn:chiRS1}, \\
      {\mathcal X}^-_{r,s}&=&\bigoplus_{n\in \mathbb N} \mathcal J_{r,p'-s;2n} \otimes \ell_{2n} \label{eqn:chiRS2},
      \eea
      for 
      \be\label{pprime_range} 1\leq r\leq p, ~~ 1\leq s\leq p',\ee and where $\ell_n$ is the $n$-dimensional irreducible representation of $sl(2,\mathbb{C})$. The ${\cal J}_{r,s}$ are $sl(2,\mathbb{C})$ singlets.

      One can also define the closely related \lVbs ~VOA (also called the $(p,p')$ singlet model) as a subalgebra of $\mathcal F_0$ via (see, e.g., \cite{Creutzig:2014nua})
      \be
      \label{eqn:lVbs}
      \text{\lVbs} := \text{ker}_{\mathcal F_0} e_+ \cap \text{ker}_{\mathcal F_0}f_-.
      \ee

      \subsubsection{Characters}
      \label{subsec:pp'char}

      From the discussion in the previous section, one can compute the characters 
      \be
      {\rm ch}_{r,s}^\pm (\tau, \xi):= {\rm Tr}_{\chi^\pm_{r,s}} q^{L_0-c/24}z^{J_0}= \begin{cases}\sum_{n \geq 0} {\rm ch}\mathcal J_{r,p'-s;2n+1}(\tau){\rm ch}\ell_{2n+1}(\xi)\\
      \sum_{n \geq 1} {\rm ch}\mathcal J_{r,p'-s;2n}(\tau){\rm ch}\ell_{2n}(\xi)\end{cases}
      \ee
      using 
      $${\rm ch}\mathcal J_{r,s;n}(\tau)={q^{1-c\over 24}\over \eta(\tau)} \sum_{k\geq 0} \left ( q^{\Delta_{r,s;n+2k}}+ q^{\Delta_{r,s;-n-2-2k}}- q^{\Delta_{r,p'-s;n+2k+1}}- q^{\Delta_{r,p'-s;-n-2k-1}}\right )$$
      and
      $${\rm ch}\ell_n(z)= {z^n - z^{-n} \over z-z^{-1}}= z^{n-1} + z^{n-3} + \ldots + z^{-n+1}.$$
      Beginning with the case of ${\rm ch}_{r,s}^+ (\tau, \xi)$, we have
      \begin{multline}\label{chiStep1}
        {\rm ch}_{r,s}^+(\tau,\xi)={q^{1-c\over 24}\over \eta(\tau)}\sum_{n=0}^\infty \ell_{2n+1}(z)\sum_{k=0}^\infty \left( q^{\Delta_{r,p'-s;2n+2k+1}}+ q^{\Delta_{r,p'-s;-2n-3-2k}}\right.\\\left.- q^{\Delta_{r,s;2n+2k+2}}- q^{\Delta_{r,s;-2n-2k-2}}\right ).
      \end{multline}
      Now we use the identity 
      $$\sum_{n\geq 0 }f(n) \sum_{k\geq 0} g(k,n)= \sum_{k\geq 0}g(0,k) \sum_{n=0}^k f(n)$$
      for $f(n) = \ell_{2n+1}(z)$ and $g(k,n)= q^{\Delta(n,k)}$ for each of the four terms in \eqref{chiStep1} to rewrite this as
      \begin{multline}
        {\rm ch}_{r,s}^+(\tau,\xi)={q^{1-c\over 24}\over \eta(\tau)}\sum_{k=0}^\infty {z^{2k+2}-2 + z^{-2k-2}\over (z-z^{-1})^2} \left( q^{\Delta_{r,p'-s;2k+1}}+ q^{\Delta_{r,p'-s;-2k-3}}\right.\\\left.- q^{\Delta_{r,s;2k+2}}- q^{\Delta_{r,s;-2k-2}}\right ),
      \end{multline}
      where we have used the fact that $$\sum_{n=0}^k \ell_{2n+1}(z)= {z^{2k+2}-2 + z^{-2k-2}\over (z-z^{-1})^2}.$$ Taking $k \to -k-2$ in the second and fourth terms, we can rewrite this as
      \be
      {\rm ch}_{r,s}^+(\tau,\xi)={q^{1-c\over 24}\over \eta(\tau)}\sum_{k\in \ZZ} {z^{2k+2}-2 + z^{-2k-2}\over (z-z^{-1})^2} \left( q^{\Delta_{r,p'-s;2k+1}}- q^{\Delta_{r,s;2k+2}}\right ),
      \ee
      which, after plugging in the explicit forms of $\Delta_{r,s;n}$ and $c(p,p')$ and shifting $k\to k-1$, finally yields
      \be\label{eqn:ppprime_tripletchar1}
      {\rm ch}_{r,s}^+(\tau, \xi)={1\over \eta(\tau)}\sum_{k\in \ZZ} {z^{2k}-2 + z^{-2k}\over (z-z^{-1})^2} \left( q^{(2pp'k+\tilde \mu_{r,s,1})^2\over 4pp'}- q^{(2pp'k+\tilde \mu_{r,s,2})^2\over 4pp'}\right ),
      \ee
      where 
      \be\label{mu_pprime}
      \tilde \mu_{r,s,1} := ps+p'r ~,~ \tilde \mu_{r,s,2} := ps-p'r.
      \ee
      In the next section, we will also consider the {\em generalised characters}, namely functions defined as in (\ref{eqn:ppprime_tripletchar1}) with $r, s \in {\mathbb Z}$ that are not necessarily in the range 
      (\ref{pprime_range}). When considered as such generalised characters, we see that they have the symmetry property
      \be\label{pprime_reflection}
      {\rm ch}_{r,s} = \epsilon_r \epsilon_s  
      {\rm ch}_{\epsilon_rr ,\epsilon_ss} ~~{\rm for}~~ \epsilon_r,\epsilon_s \in \{1,-1\}.
      \ee

      A similar computation leads to a formula for the character ${\rm ch}_{r,s}^- (\tau, \xi)$, which is given by
      \begin{multline}
        {\rm ch}_{r,s}^-(\tau, \xi)={1\over \eta(\tau)}\sum_{k\in \ZZ} {z^{2k+1}-z -z^{-1} + z^{-2k-1}\over (z-z^{-1})^2}\\\times \left( q^{(2pp'k+pp'-(ps+p'r))^2\over 4pp'}- q^{(2pp'k+pp'+(ps-p'r))^2\over 4pp'}\right ),
      \end{multline}
      where we have used that
      $$\sum_{n=1}^k {\rm ch}\ell_{2n}(\xi)= {z^{2k+1}-z -z^{-1} + z^{-2k-1}\over (z-z^{-1})^2}.$$
      Using
      $$\lim_{z\to 1} \left ({z^{2k}-2 + z^{-2k}\over (z-z^{-1})^2}\right) = k^2$$
      and 
      $$\lim_{z\to 1} \left ({z^{2k+1}-z -z^{-1} + z^{-2k-1}\over (z-z^{-1})^2}\right)= k(k+1),$$
      we obtain
      \be
      {\rm ch}_{r,s}^+(\tau, 0)={1\over \eta(\tau)} \sum_{k\in \ZZ} k^2\left( q^{(2pp'k-(ps+p'r))^2\over 4pp'}- q^{(2pp'k+(ps-p'r))^2\over 4pp'}\right )
      \ee
      and
      \be
      {\rm ch}_{r,s}^-(\tau, 0)={1\over \eta(\tau)}\sum_{k\in \ZZ} k(k+1) \left( q^{(2pp'k+pp'-(ps+p'r))^2\over 4pp'}- q^{(2pp'k+pp'+(ps-p'r))^2\over 4pp'}\right ).
      \ee
      Note that the above characters can be expressed in terms of sums of theta functions and their derivatives \cite{FGST}. 

      Taking the $z$-constant term of \eqref{eqn:ppprime_tripletchar1} gives the corresponding singlet character
      \be\label{eqn:pprime_singletChar1}
      {\rm ch}^{0,+}_{r,s}(\tau)={1\over \eta(\tau)}\sum_{k\in \ZZ} |k|  \left( q^{(2pp'k-ps-p'r)^2\over 4pp'}- q^{(2pp'k-ps+p'r))^2\over 4pp'}\right ).
      \ee
      Notice the relation between the above characters and Eichler integrals of theta functions. 
      In \cite{Cheng:2018vpl} we propose that for four exceptional fibers the following building blocks $\Xi_{m,r}$ play a role analogous to that of the false theta functions in the case of Seifert manifolds with 3 exceptional fibers:
      \be
      {\Xi}_{m,r}(\tau)  := \sum_{\substack{n\in \ZZ\\ n=r\xmod{2m}}} {\rm sgn}(n+m-1) \lfloor \tfrac{n+m-1}{2m} \rfloor \, q^{{n^2\over 4m}},
      \ee
      where $r\in \ZZ/2m$, satisfying
      \be {\Xi}_{m,r} = {\Xi}_{m,-r}. \ee
      It relates to the weight 3/2 and 1/2 holomorphic Eichler integrals as follows. 
      Let 
      \be\label{def:buildingblock4leg}
      B_{m,r}(\tau) :={1\over 2m}\left (\Phi_{m,r}(\tau)-r\Psi_{m,r}(\tau)\right ) =\sum_{k\in \ZZ} |k|\, q^{(2mk+r)^2\over 4m} , 
      \ee
      where $\Psi_{m,r}$ and $\Phi_{m,r}$ are the holomorphic Eichler integrals of  unary theta functions of weight 3/2 and 1/2 respectively, with the definitions \eqref{eqn:false} and
      \be
      \label{eqn:false2}
      \Phi_{m,r}(\tau) =
      \sum_{\substack{n\in \ZZ\\ n=r \xmod{2m}}} {\rm sgn}(n)\,n  \, q^{{n^2\over 4m}}= \sum_{\substack{n\geq 0 \\ n= r \bmod 2m}} nq^{n^2/4m}+\sum_{\substack{n\geq 0 \\ n= -r \bmod 2m}}n q^{n^2/4m}.
      \ee
      We see that 
      \be\label{shift_B}
      B_{m,r+2nm}=  {\Xi}_{m,r} - n \Psi_{m,r}
      \ee
      when $-m<r\leq m$.
      Compared with (\ref{eqn:pprime_singletChar1}), we have
      \be\label{eqn:rel_B_pprimechars}
      \eta(\tau)\,{\rm ch}^{0,+}_{r,s}(\tau)= \left(B_{pp', ps+p'r}- B_{pp', ps-p'r}\right)(\tau), 
      \ee
      which moreover coincides with $\left(\Xi_{pp', ps+p'r}- \Xi_{pp', ps-p'r}\right) (\tau)$ for $r$ and $s$ in the range \eqref{pprime_range}. 

      \subsection{Spectral Curves for Vertex Algebras}
      \label{subsec:spectral}

      When a log VOA has affine Kac-Moody symmetry $\hat{\mathfrak{g}}$, its module $V$ is naturally graded by this symmetry and the character $\chi_V (\tau,\xi) := \text{Tr} q^{L_0 - \frac{c}{24}} z^J$ is a function of $z$ that takes values in the maximal torus of $G$, such that $\mathfrak{g} = \text{Lie} (G)$. We wish to explore the $z$-dependence of characters in log VOAs, in particular $q$-difference operators that annihilate $\chi_V (\tau,\xi)$:
      \be
      \hat A \, \chi_V (\tau,\xi) \; = \; 0
      \label{qspectral}
      \ee

      \subsubsection{Example: the Triplet Algebra log-\texorpdfstring{${\cal V}_{\bar \Lambda}(m)$}{VL(m)}}

      Recall that the character of the triplet 
      $(1,m)$ model log-${\cal V}_{\bar \Lambda}(m)$
      is \eqref{eqn:A1characters}:

      \be
      \chi_{{1-s\over 2\sqrt{m}}\vec\alpha}
      (\tau,\xi) = {1\over \eta(\tau)}
      \sum_{n\in \ZZ} q^{(-s+m+2mn)^2\over 4m} {z^{1+2n}-z^{-(1+2n)} \over z-z^{-1}}
      \ee
      where we made explicit the dependence on the parameters $m$ and $s \in \{ 1, \ldots, m \}$.
      We claim that $\chi_{{1-m\over 2\sqrt{m}}\vec\alpha} (\tau,\xi)$
      is annihilated by the following $q$-difference operator:
      \be
      \hat A \chi_{{1-m\over 2\sqrt{m}}\vec\alpha} (\tau,\xi) = 0 , \ee
      where
      \be
      \hat A \; = \;
      \hat y^m + \frac{q^m({\hat z}^2-1)-q^{-m}({\hat z}^4-{\hat z}^{-2})}{q^{2m} {\hat z}^2 - 1} + \frac{{\hat z}^2 - q^{2m}}{q^{2m} {\hat z}^2 - 1} \hat y^{-m},
      \label{Amm}
      \ee
      where $\hat z$ and $\hat y$ form the algebra $\hat y \hat z = q \hat z \hat y$, usually called the {\it quantum torus}. On a function $f(\tau,\xi)$ these operators act as $$\hat z f(\tau,\xi) = z f(\tau,\xi), ~~\hat y f(\tau,\xi) = f(\tau,\xi+\tau).$$ It is easy to see that they indeed satisfy the desired $q$-commutation relation.

      Let us sketch the derivation of \eqref{Amm}. First, it is convenient to remove the denominator and introduce an auxiliary function:
      \be
      F_{m,s} (\tau,\xi) \; = \; (z - z^{-1}) \, \chi_{{1-s\over 2\sqrt{m}}\vec\alpha} (\tau,\xi)
      \ee
      We then observe that it has a structure similar to the unknot (and, more generally, torus knots), and so as in \cite{Awata:2012vzi} we make the following ansatz for the operator that annihilates $F_{m,m} (\tau,\xi)$:
      \be
      \hat A_0 \; = \; \hat y^m + \hat y^{-m} + R(\xi,\tau)
      \ee
      with some rational function $R(\xi,\tau)$ that needs to be determined. Then, it is easy to show that
      \be
      \hat A_0 F_{m,m} \; = \; 0
      \ee
      indeed holds with
      \be
      R(\xi,\tau) \; = \; 1 - q^{-2m} (1 + z^2 + z^{-2}).
      \ee
      The final step in getting to \eqref{Amm} requires passing from $F_{m,m}$ to $\chi_{{1-m\over 2\sqrt{m}}\vec\alpha}
      (\tau,\xi)$, which at the level of $q$-difference is achieved by conjugating with ${\hat z} - {\hat z}^{-1}$:
      \be
      \frac{1}{{\hat z}-{\hat z}^{-1}} \hat A_0 ({\hat z}-{\hat z}^{-1}) \; = \; \frac{q^m {\hat z} - q^{-m} {\hat z}^{-1}}{{\hat z}-{\hat z}^{-1}} \hat y^m + \frac{q^{-m} {\hat z} - q^{m} {\hat z}^{-1}}{{\hat z}-{\hat z}^{-1}} \hat y^{-m} + R(\xi,\tau).
      \ee
      Multiplying by $\frac{{\hat z}-{\hat z}^{-1}}{q^m {\hat z} - q^{-m} {\hat z}^{-1}}$ from the left, we get \eqref{Amm}.

      In the classical limit $q \to 1$, the quantum curve \eqref{Amm} becomes a hyperelliptic curve
      \be
      y^m + y^{-m} \; = \;z^{-2}+z^{2}
      \ee
      It would be interesting to extend this calculation to other values of $s$ and to more general logarithmic vertex algebras.

      \subsection{Fermionic forms of log VOA characters}
      \label{subsec:fermform}

      Later, in section \ref{sec:fermionic} we explain how connections between knot theory and physics (or, knot theory and enumerative geometry) can teach us useful lessons about the structure of the $\widehat{Z}$-invariants for many closed hyperbolic 3-manifolds.

      Until recently, this was the major obstacle in understanding the modular properties of $\widehat{Z}_b (X)$ for hyperbolic $X$ and identifying vertex algebras dual to 3-manifolds in the sense of \eqref{ZhatVOA}. By connecting knot theory to physics and enumerative geometry, this obstacle can be removed and one finds new avenues for exploring the modular properties of the BPS $q$-series invariants and connections to vertex algebras. Relegating a more complete account of these developments to section \ref{sec:fermionic}, here we briefly recall the relevant structure in the triplet vertex algebra \cite{Feigin:2007sp}.

      Let us consider the logarithmic vertex operator algebra $\text{log-}\mathcal{V}_{\bar{\Lambda}_{A_1}}(m)$ discussed in \S\ref{subsec:1m}. We are interested in the $2m$ irreducible representations whose characters are given in \eqref{eqn:A1characters}. In particular, we are interested in the linear combination of modules (in a notation consistent with \eqref{eqn:chiRS1}-\eqref{eqn:chiRS2}):
      \begin{equation}
        \mathcal{X}_{1,s} = \mathcal{X}^+_{\frac{1-s}{2\sqrt{m}} \vec{\alpha}} \oplus \mathcal{X}^-_{\frac{\sqrt{m}}{2} + \frac{1-s}{2\sqrt{m}} \vec{\alpha}}.
      \end{equation}
      One may write characters of these modules in the ``bosonic form,'' as in \eqref{eqn:A1characters}.
      Another way to write characters of these modules is via embedding the local chiral algebra of the $(1,m)$ model into a larger algebra $\mathcal{A}(m)$:
      \begin{equation}
        \text{log-}\mathcal{V}_{\bar{\Lambda}_{A_1}}(m) \xhookrightarrow{} \mathcal{A}(m).
      \end{equation}
      To obtain the algebra $\mathcal{A}(m)$, we first consider the $sl(2,\mathbb{C})$ doublet of fields:
      \begin{equation}
        a^+(z) = e^{-\sqrt{\frac{m}{2}} \varphi(z)}, \quad a^-(z) = [e,a^+], \quad e = \frac{1}{2 \pi i} \oint dz e^{\sqrt{2m} \varphi(z)},
      \end{equation}
      where the operators $a^\pm$ have the same conformal dimension $(3p-2)/4$. The OPE of $a^\pm$ has the form
      \begin{equation}
        a^+(z)a^-(w) = (z-w)^{-\frac{3p-2}{2}}\sum_{n \geq 0}H^n(w)
      \end{equation}
      where each $H^n(w)$ has conformal dimension $n$.
      The algebra $\mathcal{A}(m)$ generated by these operators is graded by the weight lattice of $sl(2,\mathbb{C})$:
      \begin{equation}
        \mathcal{A}(m) = \bigoplus_{\beta \in \sqrt{\frac{m}{2}}\mathbb{Z}} \mathcal{A}(m)^\beta
      \end{equation}
      which can be viewed as the origin of the $sl(2,\mathbb{C})$ symmetry in the triplet log VOA. In particular, this leads to the decomposition of the highest-weight irreducible modules\footnote{These modules have conformal dimension $\frac{s^2-1}{4m}+\frac{1-s}{2}$.} $\mathcal{X}_{1,s}$ generated from the vector $|s,m \rangle$,
      \begin{equation}
        a^\pm_{-\frac{3m-2s}{4}+n} | s,m \rangle = 0, \quad s \in \{1,\cdots,m\}, n \in \mathbb{N}.
      \end{equation}
      into {Vir$ \,\oplus\, sl(2,\mathbb{C})$} modules:
      \begin{equation}
        \mathcal{X}_{1,s} = \oplus_{n \in \mathbb{N}} \mathcal{J}_{m-s;n} \otimes \ell_n
      \end{equation}
      This allows to identify the irreducible modules of $\mathcal{A}(m)$ with those of the $(1,m)$ model $\text{log-}\mathcal{V}_{\bar{\Lambda}_{A_1}}(m)$.

      The fermionic form of characters then comes from the filtration on the graded algebra $\mathcal{A}(m)$. Relegating the details to \cite{Feigin:2007sp}, we reproduce here the resultant character formula for the irreducible module $\mathcal{X}_{1,s}$, with the fugacity $z$ set to $z=1$:
      \begin{equation}
        \begin{gathered}
          \chi_{s}(q) = q^{\frac{s^2-1}{4m} + \frac{1-s}{24} - \frac{c}{24}} \sum_{\vec{n} \in \mathbb{Z}_{\geq 0}^{m+1}} \frac{q^{\frac{1}{2} \vec{n} \tilde{C} \vec{n} + \vec{v}_s \cdot \vec{n}}}{(q)_{\vec{n}}} \\
          \tilde{C} = 
          \begin{pmatrix} 
            m/2 & m/2 & 1 & 2 & 3 & \cdots & m-1 \\
            m/2 & m/2 & 1 & 2 & 3 & \cdots & m-1 \\
            1 & 1 & 2 & 2 & 2& \cdots & 2 \\
            2 & 2 & 2 & 4 & 4 & \cdots & 4 \\
            3 & 3 & 2 & 4 & 6 & \cdots & 6 \\
            \cdots & \cdots & \cdots & \cdots & \cdots & \cdots & \cdots \\
            m-1 & m-1 & 2 & 4 & 6 & \cdots & 2(m-1) 
          \end{pmatrix}  \\
          \vec{v}_s = \left(\frac{m-s}{2}, \frac{m-s}{2}, 0, \cdots, 0, 1,2,\cdots,m-s \right)
        \end{gathered}
      \end{equation}
      Here, $c$ denotes the central charge of the logarithmic $(1,m)$ model and $0$ in $\vec{v}_s$ occurs $s-1$ times.

      Next, we explain how many elements of this paper find their natural home in the framework of 3d supersymmetric quantum field theory, related to the study of 3-manifold invariants via 3d-3d correspondence.

      \subsection{Log VOAs and 3d \texorpdfstring{${\cal N}=2$}{N=2} Theories}
      \label{sec:physics}

      Two-dimensional logarithmic VOAs and CFTs are relevant to many physical phenomena, including quantum Hall effect (QHE) plateau phase transition, percolation, and self-avoiding walks~\cite{Cardy:2003zr}.
      Curiously, their characters arise from supersymmetric theories in one extra dimension, in a way akin to holography, namely as half-indices of 3d $\CN=2$ theories with 2d $(0,2)$ boundary conditions~\cite{Gadde:2013wq}.

      The half-indices are basically 3d analogues of elliptic genera \cite{Witten:1986bf} in two-dimensional systems that count local operators in $\bar Q_+$-cohomology on the boundary of the 2d-3d combined system. Such combined systems naturally appear in the study of a 6d fivebrane theory partially twisted along a 4-manifold, especially in operations involving cutting and gluing~\cite{Gadde:2013sca}. The resulting 3d $\CN=2$ theory is then topologically twisted along one of its directions and holomorphically twisted along two other directions.\footnote{In particular, this clarifies a holographic-like nature in which characters of log VOAs arise \cite{Cheng:2018vpl} from three dimensions: a topological twist along one of the directions effectively projects the system to two dimensions, which then is further projected to its chiral sector by taking $\bar Q_+$-cohomology.}
      The study of such partial and holomorphic twists in supersymmetric QFTs goes back to \cite{Johansen:1994aw,Bershadsky:1995vm,Silverstein:1995re,Katz:2004nn} and has been an area of active research in recent years.

      There are several ways to formulate the half-index of the 2d-3d combined system. In radial quantization, counting local operators in QFT$_d$ requires surrounding such operators by a sphere $S^{d-1}$ and studying the Hilbert space $\CH (S^{d-1})$. In the case of 3d theory with 2d boundary, a local operator on the boundary is surrounded by a disk $D^2$, and so the analogue of radial quantization involves taking the trace over $\CH (D^2)$ or, equivalently, computing the partition function on $S^1 \times_q D^2$
      \begin{eqnarray}
        I\!\!I \Big( { \text{3d }\CN=2 \text{ theory} \atop + \text{ 2d } (0,2) \text{ bdry}} \Big) & = & \Tr_{\CH (D^2)} (-1)^F q^{R/2 + J_3} \nonumber \\
                                                                                                                                  & = & \text{partition function on } S^1 \times_q D^2
                                                                                                                                  \label{halfindex} \\
                                                                                                                                                                                                                                                                & = & \Tr_{H^* (\CH, \bar Q_+)} (-1)^F q^{R/2 + J_3} \nonumber
      \end{eqnarray}
      where this way of writing $S^1 \times_q D^2$ reminds us that the result depends on complex structure $\tau = {1\over 2\pi i}\log q$ of the boundary torus $T^2 = \partial \left( S^1 \times D^2 \right)$.
      Sometimes, the $S^1 \times_q D^2$ partition function is also called {\it K-theoretic vortex partition function} (with Omega-background along $D^2_q$).

      The basic ingredients of 3d $\CN=2$ Lagrangian theories include two types of supermultiplets: chiral and vector. Similarly, there are two types of matter supermultiplets in 2d $(0,2)$ theories, Fermi and chiral, so that basic elements of 2d $(0,2)$ gauge theories are Fermi, chiral, and vector multiplets. Below we summarize their contribution to the index \eqref{halfindex}:

      \begin{itemize}

        \item
          The contribution of a 2d $(0,2)$ Fermi multiplet to the elliptic genus and, hence, also to the index \eqref{halfindex} is basically a theta-function,
          \be
          \text{Fermi}
          \; = \; \theta \; = \;
          \overbrace{(z;q)_{\infty}}^{\partial_-^n \psi_-}
          \, \overbrace{\textcolor{red}{(qz^{-1};q)_{\infty}}}^{\partial_-^n \bar \psi_-}
          \label{Fermi}
          \ee
          where $x$ is the fugacity for the global $U(1)$ symmetry and we also indicate which modes of the Fermi multiplet contribute to various terms. Half of this contribution, shown in red, is the contribution to \eqref{halfindex} of a 3d $\CN=2$ chiral multiplet with Dirichlet boundary conditions.

        \item
          The contribution of a 2d $(0,2)$ chiral multiplet to the elliptic genus and, hence, to the index \eqref{halfindex} is the inverse theta-function,
          \begin{eqnarray}
            \text{Chiral}
& = & \frac{1}{\theta} \; = \;
\frac{\overbrace{1-z^{-1}}^{\text{0-mode of } \psi_+}}{\prod\limits_{n=0}^{\infty} \underbrace{(1-zq^n)}_{\partial_-^n \phi} \, \underbrace{(1-z^{-1}q^n)}_{\partial_-^n \bar \phi}} \\
& = & \frac{1}{\textcolor{red}{(z;q)_{\infty}} \, (qz^{-1};q)_{\infty}} \nonumber
          \end{eqnarray}
          Shown in red is the contribution to \eqref{halfindex} of a 3d $\CN=2$ chiral multiplet with Neumann boundary conditions.

        \item
          Finally, gauging a $U(1)$ symmetry with fugacity $z$ in the index \eqref{halfindex} has the effect of integrating over $z$. This operation has a clear physical meaning as it picks out gauge-invariant operators, {\it i.e.} the ``constant term'' in the $x$-dependent part of the integrand. To summarize, a 2d $(0,2)$ vector multiplet or, equivalently, a 3d $\CN=2$ vector multiplet with Neumann boundary conditions corresponds to the simple rule:
          \be
          \text{Vector} \; = \; \oint \frac{dz}{z} \; = \;
          \text{``constant term"}
          \label{vector}
          \ee

      \end{itemize}
      We can use these ingredients to (re)produce characters of older and more familiar logarithmic CFTs/VOAs.
      The best known examples of log VOAs include the following three infinite families:

      \begin{itemize}

        \item
          {\bf Symplectic fermions} are basically $\beta\gamma$-systems, labelled by an integer $d>0$ (the number of symplectic fermions) and with central charge
          \be
          c \; = \; -2d.
          \label{csymplf}
          \ee
          Note, that negative values of the central charge signal non-unitarity, which is a general feature of logarithmic theories.

        \item
          {\bf Triplet} $(1,m)$ models, denoted as {log-${\cal V}_{\bar\Lambda_{A_1}}(m)$} in \S\ref{subsec:1m},  are labelled by an integer $m>1$ and have central charge:
          \be
          c \; = \; 13 - 6 \left( m+\frac{1}{m} \right).
          \label{ctriplet}
          \ee

        \item
          {\bf Singlet} $(1,m)$ models, denoted as {log-${\cal V}_{\bar\Lambda_{A_1}}^0(m)$} in \S\ref{subsec:1m},  are particular subsectors of triplet $(1,m)$ models (and so have the same central charge):
          \be
          c \; = \; 13 - 6 \left( m+\frac{1}{m} \right).
          \label{csinglet}
          \ee

      \end{itemize}
      The last two families correspond to $\mathfrak{sl}(2)$ Lie algebra, which is not obvious in a short summary given here. They admit generalizations to other Lie algebras $\mathfrak{g}$ and to $(p,p')$ models labelled by two integers $p$ and $p'$, all of which are less studied.

      Note that, for $d=1$ and $m=2$ the central charges \eqref{csymplf} and \eqref{csinglet} take equal value $c=-2$. This is a manifestation of the relation between the simplest symplectic fermions with $d=1$ and the singlet $(1,2)$ model, which are equivalent.\footnote{To be more precise, they are related by gauging a $\mathbb{Z}_2$ symmetry.}
      The corresponding character is easily obtained by writing the invariant combinations of $\psi_n$ and $\tilde \psi_m$, the modes of two fermions \cite{RochaCaridi,Kausch:1995py,Guruswamy:1996rk,Flohr:2006id}:
      \begin{eqnarray}
        \chi (q) & = & \sum_{n=0}^{\infty} \frac{q^{n^2+n}}{(q;q)_n^2} \label{char} \\
                 & = & 1 + q^2 + 2 q^3 + 3 q^4 + 4 q^5 + 6 q^6 + 8 q^7 + 12 q^8 + 16 q^9 + \ldots \nonumber
      \end{eqnarray}
      Naturally, this is called the {\it fermionic form} of the character, which we already encountered in section \ref{subsec:fermform} and that will be discussed in more detail in section \ref{sec:fermionic}. It also has a {\it bosonic form}\footnote{In the literature, this form of characters is usually called {\it bosonic} because the alternating signs come from the subtraction of null vectors.}
      \begin{eqnarray}
        \chi (\tau) & = & \frac{1}{(q)_{\infty}}
        \sum_{n=0}^{\infty} (-1)^n q^{n(n+1)/2}
        \; = \; \frac{q^{-1/8}}{(q)_{\infty}} \Psi_{2,1} (\tau) \label{sfbos} \\
                    & = & \frac{1}{(q)_{\infty}} \sum_{n=0}^{\infty} \frac{(-q)^n (-q;q^2)_n}{(-q^2;q^2)_n} \nonumber
      \end{eqnarray}
      which will be useful in what follows. In particular, we will demonstrate how this character arises from a 3d $\CN=2$ theory.  

      Since the character of this $c=-2$ logarithmic model is constructed as the space of neutral (charge-0) states of two fermions $\psi$ and $\tilde \psi$ with charges $-1$ and $+1$, it is already in the form that can be easily converted to the supersymmetric index \eqref{halfindex} of a 2d-3d combined system. Namely, the modes of the 2d chiral fermions $\psi$ and $\tilde \psi$ each comprise the field content of a 2d $(0,2)$ ``half-Fermi'' multiplet. Since they carry charges $+1$ and $-1$, respectively, they contribute to the half-index factors
      \be
      \chi_{\pm} (\tau,\xi)
      \; = \; \prod_{n=1}^{\infty} (1 - z^{\pm 1} q^n)
      \; = \; (z^{\pm 1} q;q)_{\infty}
      \ee
      where $e^{2\pi i\xi} = z$ is the fugacity for the global symmetry (that we are about to gauge). 
      Therefore, the elliptic genus of two such multiplets (complex fermions) with charges $-1$ and $+1$ is
      $(z^{-1} q;q)_{\infty} (zq;q)_{\infty}$.
      Introducing a 2d $(0,2)$ vector multiplet and gauging this symmetry of the fermions means taking the constant term in this infinite product or, equivalently, integrating over $z$, {\it cf.} \eqref{vector}:
      \be
      \int_{|z|=1} \frac{dz}{z}
      (z^{-1} q;q)_{\infty} (zq;q)_{\infty}
      \; = \; 1 + q^2 + 2 q^3 + 3 q^4 + 4 q^5 + 6 q^6 + 8 q^7 + \ldots
      \label{chiinte}
      \ee
      This clearly agrees with \eqref{char}--\eqref{sfbos}. So, we have our first result: we managed to find a 2d $(0,2)$ physical system whose elliptic genus equals the character of the $c=-2$ log VOA.

      More precisely, our realization of symplectic fermions in supersymmetric QFT involves a 2d $(0,2)$ theory on a boundary of 3d $\CN=2$ theory. Indeed, a pure two-dimensional gauge theory with half-Fermi multiplets carrying charges $+1$ and $-1$ has gauge anomaly and, by itself, would be inconsistent. It has $-\frac{1}{2} - \frac{1}{2} = -1$ units of gauge anomaly, which can be compensated by anomaly inflow from 3d $\CN=2$ gauge theory with $G=U(1)$ and supersymmetric Chern-Simons term at level $k=+1$. In fact, this model is just a special case of a more general class of 2d-3d coupled systems in~\cite{Gadde:2013sca}.

      In particular, half-Fermi multiplets naturally arise from 3d $\CN=2$ chiral multiplets with Dirichlet boundary conditions~\cite{Gadde:2013sca}. So, we conclude that the character \eqref{char}--\eqref{sfbos} of the symplectic fermions is equal to 2d-3d half-index of the following system:
      \begin{center}
        \begin{tabular}{c || c}
          3d $\CN=2$ multiplet & boundary condition \\ \hline\hline
          $U(1)$ vector with $k=+1$ super-CS  & Neumann \\
          chiral with charge $+1$ & Dirichlet \\
          chiral with charge $-1$ & Dirichlet
        \end{tabular}
      \end{center}
      This theory is a special instance of
      \begin{center}
        \begin{tabular}{r l}
          {\bf Theory A:}
& 3d $\CN=2$ gauge theory with gauge group $U(N_c)$, \\
& Chern-Simons level $k>0$, \\
& and $N_f$ pairs of charged chirals with R-charge $R$
        \end{tabular}
      \end{center}
      which by the famous Giveon-Kutasov duality~\cite{Giveon:2008zn} is dual to
      \begin{center}
        \begin{tabular}{r l}
          {\bf Theory B:}
& 3d $\CN=2$ gauge theory with gauge group $U(N_f + k - N_c)$, \\
& Chern-Simons level $-k$, \\
& $N_f$ pairs of charged chirals $(q_a, \tilde q^a)$ with R-charge $1-R$, \\
& and $N_f^2$ uncharged chirals $M_{ab}$ with superpotential
        \end{tabular}
        \be
        W \; = \; \sum_{a,b} q^a {M_{a}}^b \tilde q_b
        \label{WtheoryB}
        \ee
      \end{center}
      Namely, our original Theory A has $N_c=1$, $k=1$, $N_f=1$, $R=0$. Therefore, its dual Theory B is a $U(1)_{-1}$ gauge theory with the following field content
      \begin{center}
        \begin{tabular}{c || c}
          3d $\CN=2$ multiplet & boundary condition \\ \hline\hline
          $U(1)$ vector with $k=-1$ super-CS  & Neumann \\
          chiral with charge $+1$ and $R=1$ & Neumann \\
          chiral with charge $-1$ and $R=1$ & Neumann \\
          chiral with charge $0$ and $R=0$ & Dirichlet		
        \end{tabular}
      \end{center}
      and the cubic superpotential \eqref{WtheoryB}.
      Here we also wrote the dual boundary conditions for all the fields.\footnote{Similar dualities for boundary conditions were also discussed {\it e.g.} in \cite{Chung:2016pgt,Dimofte:2017tpi,Franco:2019pum}.}
      Using the rules \eqref{Fermi}--\eqref{vector}, it is easy to see that the combined 2d-3d half-index \eqref{halfindex} with these boundary conditions produces another integral expression for our character \eqref{char}--\eqref{sfbos}, similar to \eqref{chiinte}:
      \begin{eqnarray}
        \chi (q) & = & (q)_{\infty} \int_{|z|=1} \, \frac{dz}{z}
        \frac{1}{(z^{-1} q^{1/2};q)_{\infty} (z q^{1/2};q)_{\infty}} \label{chiint2} \\
                 & = & 1 + q^2 + 2 q^3 + 3 q^4 + 4 q^5 + 6 q^6 + 8 q^7 + 12 q^8 + 16 q^9 + \ldots \nonumber
      \end{eqnarray}

      Here, we intentionally focused on the simplest non-trivial example of a log VOA character realized as the combined 2d-3d half-index, to illustrate how the failure of classical modular properties and the logarithmic nature of the VOA originate from three dimensions. If our system was entirely two-dimensional as a consistent QFT, its elliptic genus would be well-defined and exhibit familiar modular properties. However, as we saw in this simple example, the two-dimensional part of our system by itself is anomalous and requires three-dimensional ``bulk'' which, in turn, spoils modular properties. This simple example can be easily extended to more general systems related to characters of other logarithmic VOAs, old and new. In particular, via 3d-3d correspondence, many $\widehat{Z}$-invariants of 3-manifolds provide such examples.

      \section{\texorpdfstring{$\widehat{Z}^G$}{\^Z\^{}G}-invariants for Seifert Manifolds}
      \label{sec:homoblocks}

      In this section we will define our main object $\widehat Z^G_{\underline{\vec b}}$, labelled by a simply-laced Lie group $G$, a weakly negative plumbed  three-manifold $X_\Gamma$, a choice of generalized {Spin$^c$} structure $\underline{\vec b}$, and potential a Wilson line operator $W_{\vec \nu_{v_\ast}}$. Subsequently, we will study its relation to log VOAs reviewed in the previous section. 

      Denote the plumbing graph by $\Gamma$ and the resulting plumbed manifold by $X_\Gamma$. We write its adjacency matrix as $M$  and denote by $V$ its vertex set. 
      For a given simply-laced Lie group $G$, we  introduce the following notation:
      \begin{equation}
        \quad\underline{\vec{x}} \in \mathbb{R}^{|V|}  \otimes_\ZZ\Lambda   . 
      \end{equation}
      Sometimes we will write $\underline{\vec{x}} = (\vec{x}_v, \vec{x}_{v'} , \vec{x}_{v''} ,\dots )$ with $v,v',v'' ,\dots \in V$ and $\vec{x}_v\in \RR \otimes_\ZZ \Lambda$.
      Also, we define  the norms:
      \begin{equation}\label{def:notation} || \underline{\vec{x}} ||^2 := \sum_{v,v'\in V} M^{-1}_{v,v'} \langle \vec{x}_v , \vec{x}_{v'} \rangle.
      \end{equation} 
      We also define the lattice \be \Gamma_{M,G} :=M\ZZ^{|V|}\otimes_{\ZZ} \Lambda \ee
      with norm given as in (\ref{def:notation}).
      In what follows,  we choose the set ${\cal B}$ of $\underline{\vec{b}}$ to be isomorphic $\ZZ^{|V|}\otimes_{\ZZ} \Lambda/  \Gamma_{M,G}$. In particular, 
      we let 
      \be \label{dfn:setofb}
      {\cal B} =  (\ZZ^{|V|}\otimes_{\ZZ}\Lambda+\underline{\vec{b}_{0}} )  /  \Gamma_{M,G},
      \ee 
      where $\underline{\vec{b}_{0}}\in \ZZ^{|V|}\otimes_{\ZZ}\Lambda^\vee/\Lambda$ is given by $\vec b_{0,v} = {\rm d}_2(v) \vec \rho$, where $ {\rm d}_2(v) = {\rm deg}(v)\xmod{2}$. 
      Then, following 
      \cite{Park:2019xey}, we define homological blocks for the three manifold $X_\Gamma$. 

      \begin{defn}\label{dfn:blocks}
        Given  a simply-laced Lie group $G$,  the  homological blocks for a weakly negative plumbed three-manifold $X_\Gamma$ are defined as:
        \begin{multline}
          \widehat{Z}^G_{\underline{\vec{b}}}(X_\Gamma;\tau) =C_\Gamma^G(q)  \int_{\cal C} d \underline{\vec{\xi}}\,\left( \prod_{v \in V} \Delta(\vec{\xi}_v)^{2 - \deg v} \right)\\\times
          \sum_{w \in {W}} \sum_{\underline{\vec{\ell}} \in \Gamma_{M,G} +w(\underline{\vec{b}}) } q^{-\frac{1}{2}||  \underline{\vec{\ell}}||^2} 
          \left( \prod_{v' \in V} \ex^{\langle \vec{\ell}_{v'},  \vec{\xi}_{v'} \rangle} \right). 
        \end{multline}
        In the above equation, $W$ is the Weyl group and
        $w(\underline{\vec{b}}) $ denotes the diagonal action $w(\underline{\vec{b}}) =(w(\vec b_v), w(\vec b_{v'}),\dots)$.
        The integration measure is given by 
        $$ \int_{\cal C}  d \underline{\vec{\xi}}  :={\rm p.v.}  \int\prod_{v\in V} \prod_{i =1}^{\text{rank}G} {dz_{i,v}\over 2\pi i z_{i,v}},$$ and the contour ${\cal C}$ is given by the Cauchy principal value integral around the unique circle in the $z_{i,v}$-plane. 
        Recall that weakly negative means that $M^{-1}$ defines a negative-definite subspace in ${\mathbb C}^{|V|}$ spanned by the so-called high-valency vertices with deg$(v)>2$ \cite{Cheng:2018vpl}.

        Letting $\pi_M$ be the number of positive eigenvalues of $M$ and $\sigma_M$ the signature of $M$, according to \cite{Park:2019xey},
        \be\label{eqn:C} C_\Gamma^G(q)= (-1)^{|\Phi_+|\pi_M}q^{{3\sigma_M-{\rm Tr} M\over 2} |\vec \rho|^2},\ee
        where $ \Phi_+$ is a set of positive roots for $G$ and $\vec \rho$ is a Weyl vector for $G$.
        In the case of $G=SU(r+1)$, 
        we have 
        \be\label{eqn:rho2}
        |\vec \rho|^2 = {r(r+1)^2\over 12}, 
        \ee
        and the factor becomes
        \be C^G_\Gamma(q)= (-1)^{{r(r+1)\over 2}{\bf\pi}_M}q^{{3\sigma_M-{\rm Tr}M\over2}{r(r+1)^2\over 12}}.\ee

      \end{defn}
      \vspace{10pt}

      In what follows, we will  specialise the above definition to our main cases of interest in this paper. First, we specialise to the ``$N$-leg star graphs" which contain only a single 
      node with degree $N$ larger than two, which we will refer to as the central node $v=v_0$. 
      The resulting plumbed manifolds are Seifert three-manifolds. 
      We will  restrict to weakly negative star graphs, and the weak negativity simply means $M^{-1}_{v_0,v_0}<0$ in this case. We say that the corresponding manifold is  negative Seifert.  From \eqref{eqn:M00}, we see that these are precisely the manifolds 
      $X_\Gamma =M(b;\{q_i/p_i\}_{i=1,\dots,N})$ 
      with 
      \be \textfrak{e} = 	b +\sum_k {q_k\over p_k} < 0 . 
      \ee

      \begin{prop}\label{prop:main_generalSeif}

        Fix a simply-laced Lie group $G$. Consider an $N$-leg star graph that corresponds to a negative Seifert manifold $X_\Gamma =M(b;\{q_i/p_i\}_{i=1,\dots,N})$ with $N$ exceptional fibers, with the orbifold Euler characteristic given by $\textfrak{e} =b +\sum_k {q_k\over p_k}$.
        Let  $D$ be the smallest positive integer such that ${D\over 
        \textfrak{e} p_i}\in \ZZ$ for $i=1,\dots,N$ and set $m =  -D M_{v_0,v_0}^{-1}$. The homological blocks, defined in Definition \ref{dfn:blocks}, are given by
        \begin{gather}\begin{split}
          \widehat{Z}^G_{\underline{\vec{b}}}(\tau) &=C_\Gamma^G(q) \sum_{\hat{w} \in {W}^{\otimes N}} (-1)^{\ell(\hat{w})} \int_{\cal C} d\vec{\xi}
        \,\tilde \chi_{\hat{w};\underline{\vec{b}}}(\tau,\vec \xi) \\  \end{split}
      \end{gather}
      where  the contour $\cal C$ is as described in Definition \ref{dfn:blocks}, the 
      integrand is either $\tilde \chi_{\hat{w};\underline{\vec{b}}}=0$, or 
      there exists a unique  $\vec \kappa_{\hat{w};\vec{\underline{b}}}  \in  \Lambda/D\Lambda$ such that 
      \begin{gather}\label{eqn:integrand_N}
        \begin{split}
          \tilde \chi_{\hat{w};\underline{\vec{b}}}(\tau,\vec \xi) & = {q^\delta \over \Delta(\vec{\xi})^{N-2}}\times \\& \sum_{ \vec\lambda \in \Lambda }q^{\frac{1}{2D}  |\sqrt{m}(D \vec \lambda +\vec \kappa_{\hat{w};\vec{\underline{b}}}  + \varepsilon(N)\vec\rho ) + {m\vec A_{\hat {w}} \over \sqrt{m}}|^2}\sum_{w\in W} (-1)^{Nl(w)} 
          \,  \ex^{\langle w(D \vec \lambda +\vec \kappa_{\hat{w};\vec{\underline{b}}}  + \varepsilon(N)\vec\rho ),  \vec{\xi} \rangle} ,
        \end{split}
      \end{gather}
      where  $\delta$ and $\vec A_{\hat {w}}$ are given as
      \be\label{eqn:delta_A_main}
      \vec A_{\hat w} =  -\sum_{v_i \in V_1} 
      {{\rm sgn}(q_i)\over p_i} w_{v_i}(\vec \rho), ~~
      \delta= 
      \sum_{v\in V_1}{|\vec\rho|^2\over 2} \left({(M_{v_0,v}^{-1})^2\over M_{v_0,v_0}^{-1}} -  M_{v,v}^{-1} \right), 
      \ee
      and 
      \be\varepsilon(N) = 
      \begin{cases}  0 &~ N~ {\rm even}\\
      1 & ~ N~ {\rm odd}\end{cases}.
      \ee
      \end{prop}

      \begin{proof} 
        For a star graph, we can separate the vertices into the central node $v_0$, the end nodes with degree one and the intermediate nods with degree two: 
        \be\label{dfn:threesubsets}
        V= \{v_0\} \cup V_1 \cup V_2.
        \ee
        Integrating over the intermediate vertices, we obtain
        \be
        \int \prod_{v\in V_2} d {\vec{\xi}_v} \ex^{\langle \vec \xi_v, \vec \ell_v\rangle} = \prod_{v\in V_2} \delta_{\vec \ell_v, \vec 0} ,
        \ee
        while integrating over the end vertices gives
        \be
        \int \prod_{v\in V_1} d {\vec{\xi}_v} \,\Delta(\vec{\xi}_v) \,  \ex^{\langle \vec \xi_v, \vec \ell_v\rangle} = \prod_{v\in V_1} \left(\sum_{w\in W} (-1)^{l(w)} \delta_{\vec \ell_v, -w(\vec\rho)} \right),
        \ee
        where we have made use of the denominator identity (\ref{eqn:denom_id}).
        We are then left with an integral over the central node. Writing $\vec \xi=\vec\xi_{v_0}$ and 
        \[
          \underline{\vec{\ell}}= \Big(\vec{\ell}_0, \, \underbrace{-w_1(\vec \rho), \cdots, -w_N( \vec\rho)}_{N} , \underbrace{0, \cdots, 0}_{|V|-N-1} \Big)
        , \]
        where the three groups of vectors correspond to the three subsets of the vertices (\ref{dfn:threesubsets}), 
        we get
        \begin{multline}
          \widehat{Z}^G_{\underline{\vec{b}}} =C_\Gamma^G(q)  \int_{\cal C} d\vec{\xi}\, {1\over \Delta(\vec{\xi})^{N-2}}\\\times
          \sum_{w,w_1,\dots,w_N \in {W}} (-1)^{l(w_1)+\dots+l(w_N)} \sum_{ \vec{\ell_0}\in {S}_{w,w_1,w_2,\dots,w_N} } q^{-\frac{1}{2}||  \underline{\vec{\ell}}||^2} 
          \ex^{\langle \vec{\ell}_{0},  \vec{\xi} \rangle} 
        \end{multline}
        where we define the set
        \be\label{dfn:theset}
        {S}_{w,w_1,w_2,\dots,w_N;\vec{\underline{b}}} := \left\{ \vec \ell_0 \,\Bigg\lvert \Big(\vec{\ell}_0, \, \underbrace{-w_1(\vec \rho), \cdots, -w_N( \vec\rho)}_{N} , \underbrace{0, \cdots, 0}_{|V|-N-1} \Big) \in \Gamma_{M,G} + w(\underline{\vec b}) \right\}.
          \ee
          Note that 
          \be
          \vec{\ell_0}\in  {S}_{w,w_1,w_2,\dots,w_N;\vec{\underline{b}}} \Leftrightarrow 
          w^{-1}(\vec{\ell_0})\in  {S}_{1,w^{-1}w_1,w^{-1}w_2,\dots,w^{-1}w_N;\vec{\underline{b}}}, 
          \ee
          so we can rewrite the integral as 
          \begin{gather}\begin{split}
            \widehat{Z}^G_{\underline{\vec{b}}} =C_\Gamma^G(q)   \int_{\cal C} d\vec{\zeta}\, {1\over \Delta(\vec{\xi})^{N-2}}
            \sum_{\hat w \in {W}^{\otimes N}} (-1)^{\ell(\hat w)} \sum_{ \vec{\ell}_0\in {S}_{1,\hat w};\underline{\vec b} }\,\sum_{w\in W}(-1)^{Nl(w)} q^{-\frac{1}{2}||  \underline{\vec{\ell}}||^2} 
            \ex^{\langle w(\vec{\ell}_0),  \vec{\xi} \rangle} ,
          \end{split}
        \end{gather}
        where we have introduced the notation $\hat{w}$ to denote $(w_1,\dots,w_N)$, and define $$\ell(\hat{w}): = \sum_{v\in V_1} l(w_v) =\sum_{i=1}^N l(w_i).$$
        Explicitly, 
        we have
        \be
        ||  \underline{\vec{\ell}}||^2 = M_{v_0,v_0}^{-1} |\vec \ell_0|^2 - 2 \sum_{v\in V_1} M_{v_0,v}^{-1} \langle w_v(\vec\rho), \vec \ell_0 \rangle +\sum_{v,v'\in V_1} M_{v,v'}^{-1}\langle w_v (\vec\rho)  ,  w_{v'} (\vec\rho) \rangle. 
        \ee
        Combined with Lemma \ref{lem:Midentity}, this leads to 
        \begin{gather}
          \label{hatZ1}
          \begin{split}
            \widehat{Z}^G_{\underline{\vec{b}}}(\tau) &=C_\Gamma^G(q) (-1)^N \sum_{\hat{w} \in {W}^{\otimes N}} (-1)^{\ell(\hat{w})} \int_{\cal C} d\vec{\xi}
          \,\tilde \chi_{\hat{w};\underline{\vec{b}}}(\tau,\vec \xi) \\  \end{split},
        \end{gather}
        with the integrand given by 
        \begin{gather}
          \label{hatZ12}
          \begin{split}
            \tilde \chi_{\hat{w};\underline{\vec{b}}}(\tau,\vec \xi) &: = {q^\delta \over \Delta(\vec{\xi})^{N-2}}  \sum_{ \vec{\ell}\in {S}_{1,\hat{w};\vec{\underline{b}}} }q^{-\frac{1}{2} M_{v_0,v_0}^{-1} |\vec \ell+{\vec A_{\hat {w}} }|^2 }\sum_{w\in W} (-1)^{Nl(w)} 
            \,  \ex^{\langle w(\vec{\ell}),  \vec{\xi} \rangle} , 
          \end{split}
        \end{gather}
        where
        \be\label{dfn:A}
        \vec A_{\hat {w} } ={-1\over M_{v_0,v_0}^{-1}} \sum_{v\in V_1} M_{v_0,v}^{-1}  w_v(\vec\rho), ~~
        \delta= 
        \sum_{v\in V_1}{|\vec\rho|^2\over 2} \left({(M_{v_0,v}^{-1})^2\over M_{v_0,v_0}^{-1}} -  M_{v,v}^{-1} \right).
        \ee
        Moreover, as shown in Lemma \ref{lem:Midentity}, we can write the above in terms of the Seifert data $M(b;\{q_i/p_i\}_i)$ for the plumbed manifold as in \eqref{eqn:delta_A_main}. (See also Lemma \ref{lem:Midentity} for an alternative expression for $\delta$.)

        Next, we observe from the form of the vectors $\underline{\vec{b}}$ (\ref{dfn:setofb}) that the set (\ref{dfn:theset}) can be expressed as  
        \be	\label{eqn:setD}
        {S}_{1,\hat{w};\vec{\underline{b}}}
        =  \left\{ \vec \lambda +\varepsilon(N)\,  \vec\rho \,\Bigg\lvert M^{-1}\Big(\vec \lambda + \vec \lambda_{\hat w;v_0}, \, \vec \lambda_{\hat w;v'}, \vec\lambda_{\hat w;v''}, \dots \Big)^T \in  \Lambda^{\otimes |V|} \right\} 
          \ee	
          where  $\vec\lambda_{\hat w;v}\in \Lambda$ for all $v\in V$. More precisely, we have 
          \[\vec \lambda_{\hat w;v} =\begin{cases}- \varepsilon(N)\, \vec \rho -\vec b_{v_0}  ,& v=v_0  \\ 
            -w_i (\vec \rho) - \vec b_{v_i} ,& v=v_i \in V_1 \\ 
            - \vec b_{v} ,& {\rm otherwise} 
          \end{cases}
          \]
          for the central node, end nodes and the other nodes respectively.
          In terms of the root basis $$\vec \lambda = \sum_{\vec \alpha_k \in \Phi_s} \l^{(k)} \vec \alpha_k , $$
          the above set is given by 
          \be\label{eqn:set_non_sphere}
          {S}_{1,\hat{w};\vec{\underline{b}}}
          =  \left\{ \vec \lambda +\varepsilon(N) \vec\rho \,\Bigg\lvert 
              d_v \l^{(k)} + c_v^{(k)} \equiv 0 \xmod{D} \text{  for all }v\in V, \vec \alpha_k \in \Phi_s
            \right\},
            \ee
            where 
            $$
            c_v^{(k)} = D\left\langle \sum_{v'\in V} M_{v,v'}^{-1} \vec\lambda_{\hat w;v'}, \vec \omega_k\right\rangle 
            $$
            and the condition $D{1\over {\textfrak{e} \,p_i}}\in \ZZ$, when combined with (\ref{M0v_inv}),
            ensures that 
            \[ d_v := DM^{-1}_{v_0,v} \] 
            is an integer. 
            For a given choice of $\vec \alpha_k$ and $v\in V$, the condition  
            \be\label{enq:mod}
            d_v \l^{(k)} + c_v^{(k)} \equiv 0 \xmod{D}
            \ee
            on $\l^{(k)}$
            has a unique solution 
            in $\ZZ/D_v\ZZ$
            if 
            \[ D/D_v := {\rm{g.c.d.}}(d_v,D)\] is the greatest common divisor of each pair in the triplet $(d_v,c_v^{(k)}, D)$, and no solution otherwise.   
            Since $${\rm{l.c.m.}}(\{D_v\}_{v\in V})= {D\over {\rm{g.c.d.}}(\{d_v\}_{v\in V})}=D$$ from the definition of $D$, we conclude that (\ref{enq:mod}) for all $v\in V$ either has no solution or has a unique solution in $\l^{(k)}\in \ZZ/D\ZZ$. 
            As a result, for given $\vec{\underline{b}}$ and $\hat w$  either there exists a unique $\vec \kappa_{\hat{w};\vec{\underline{b}}}\in \Lambda /D\L$ such that
            \be\label{dfn:kappa}
            {S}_{1,\hat{w};\vec{\underline{b}}} = \{\vec \kappa_{\hat{w};\vec{\underline{b}}} + D\vec \lambda +\varepsilon(N) \vec\rho :\vec \lambda \in \Lambda \}, 
            \ee
            or ${S}_{1,\hat{w};\vec{\underline{b}}}=\emptyset$.  
            Put into (\ref{hatZ12}), we obtain (\ref{eqn:integrand_N})  in the first case and zero in the second case.

          \end{proof}

          \vspace{15pt}
          \noindent
          {\bf  Specialisation: Integral Homology Spheres}
          \vspace{5pt}

          Now, in Proposition \ref{prop:main_generalSeif},  we restrict to the graphs $\Gamma$ with one central node and $N$ legs with a unimodular plumbing matrix. 
          The only choice (up to a trivial shift with elements in $\Gamma_{M,G}$) for $\underline{\vec b}$ in \eqref{dfn:setofb} is given by  
          \be\label{hom_sphere_b}
          {\underline{\vec b}}_0 = (\varepsilon(N)\vec \rho,\underbrace{\vec \rho,\dots ,\vec \rho}_{N} , \underbrace{0, \cdots, 0}_{|V|-N-1}  ). 
          \ee
          It then follows that $ \vec \kappa_{\hat{w};\vec{\underline{b}}_0}=0$ and the set (\ref{dfn:theset}) is given by
          \be\label{hom_sphere_set}
          {S}_{w,w_1,w_2,\dots,w_N;\vec{\underline{b}}_0}\, = \{\vec \alpha +\varepsilon(N) \vec\rho : \vec \alpha\in \Lambda\}, 
          \ee
          independent of $w$ and $\hat w$. 
          Together with $w(\vec A_{\hat w}) = \vec A_{w\hat w} $, where $w$ acts diagonally on $\hat w$, we see that 
          \be\label{eq:diagonal_action}
          (-1)^{\ell(\hat w)}\chi_{ \hat{w};\vec{\underline{b}}_0}(\tau,\vec \xi)
          = (-1)^{l(w\hat w)}\chi_{w\hat{w};\vec{\underline{b}}_0}(\tau,\vec \xi) .
          \ee
          As a result, choosing any representative ${\cal W}$ of the coset $W^{\otimes N}/W$, one can express 
          the only non-trivial homological block as
          \be\label{eqn:zhat_sphere}
          \widehat{Z}^G_{0}(\tau) := \widehat{Z}^G_{\underline{\vec{b}}_0}(\tau) = C_\Gamma^G(q) |W| \sum_{\hat{w} \in {\cal W}} (-1)^{\ell(\hat{w})} \int_{\cal C} d\vec{\xi}
          \,\tilde \chi_{\hat{w};\underline{\vec{b}}_0}(\tau,\vec \xi) \ee
          with
          \be
          \label{eq:chitildebriesk}
          \tilde \chi_{\hat{w};\underline{\vec{b}}_0}(\tau,\vec \xi) = {q^\delta \over \Delta(\vec{\xi})^{N-2}}  \sum_{ \vec\alpha \in \Lambda+\varepsilon({N})\vec \rho}q^{\frac{1}{2} |\sqrt{m}\vec \alpha +{m\vec A_{\hat w}\over \sqrt{m}}|^2} \sum_{w\in W} (-1)^{Nl(w)} 
          \,  \ex^{\langle w(\vec \alpha ),  \vec{\xi} \rangle}.
          \ee
          Since $\vec A_{w\hat w}=w(\vec A_{\hat w})$, 
          a convenient choice of $W$ is given by those $\hat w$ with the corresponding $\vec A_{\hat w}\in P^+$.

          \subsection{Seifert Manifolds with Three Exceptional Fibers}
          \label{subsec:3fiber}

          Given  a simply-laced Lie group $G$, we study the integrands $\tilde \chi_{{\hat{w}};\underline{\vec{b}}}$ 
          and the resulting invariants $\widehat Z^G_{\underline{\vec{b}}}(X_\Gamma)$ for arbitrary negative three-leg graphs, leading to Seifert manifolds with three exceptional fibres. 
          As seen in Proposition \ref{prop:main_generalSeif}, the integrand is a finite sum  of specific $q$-series labelled by the set $W^{\otimes 3}$. 
          We find that these $q$-series  closely resemble the  characters of the $(1,m)$ triplet algebra {log-${\cal V}_{\bar \Lambda}$}, for a given positive integer $m$ which we will specify shortly. Subsequently, it follows from the relation (\ref{singlet_triplet}) between the singlet and triplet characters that the $\widehat Z^G_{\underline{\vec{b}}}(X_\Gamma)$ closely resembles a particular linear combination of the characters of the $(1,m)$ singlet algebra {log-${\cal V}^0_{\bar \Lambda}$}. 
          Taking this observation as a starting point, we will establish the various relations between $\tilde \chi_{{\hat{w}};\underline{\vec{b}}}$,  $\widehat Z^G_{\underline{\vec{b}}}(X_\Gamma)$ and log VOA characters summarised in Table  \ref{tab:3sinfibers}. 

          To start, we first establish the following. 
          For $N=3$, we can rewrite the integrand \eqref{eqn:integrand_N} in Proposition \ref{prop:main_generalSeif} in terms of the Lie algebra characters \eqref{highest_weight_module} as 
          \begin{gather}\label{integrand_3star2}\begin{split}
            \tilde \chi_{\hat{w};\underline{\vec{b}}}(\tau,\vec \xi)
& =  {q^{\delta}\over \Delta(\vec \xi)}  
\sum_{\substack{{\vec{\tilde \lambda}} \in \vec \kappa_{\hat{w};\vec{\underline{b}}}  +D \Lambda }}  q^{\frac{1}{2D} | \vec \mu_{\hat w} +\sqrt{m}  ( \vec{\tilde{\lambda}}  +\vec\rho )-{1\over \sqrt{m}} \vec{\rho}|^2} \sum_{w\in W}(-1)^{l(w)}
\,  \ex^{\langle w({\vec{\tilde \lambda}} +\vec \rho) ,  \vec{\xi} \rangle}\\ 
&= q^{\delta}  
\sum_{w\in W} (-1)^{l(w)}  \sum_{\substack{\vec\rho+\vec{\tilde \lambda}  \in P^+ \\ \vec\rho+\vec{\tilde \lambda} \in  w^{-1}(\vec \rho + \vec \kappa_{\hat{w};\vec{\underline{b}}} ) + D\Lambda}} \chi_{\vec{\tilde \lambda} }^{\mathfrak g}(\vec \xi) \,q^{\frac{1}{2D} | \vec \mu_{\hat w} +\sqrt{m} w( \vec{\tilde{\lambda}}  +\vec\rho )-{1\over \sqrt{m}} \vec{\rho}|^2} ,
          \end{split}
        \end{gather}
        where 
        \be
        \label{eq:Awdef}
        \sqrt{m}\vec\mu_{\hat w}= \vec\rho +m \vec A_{\hat w}=\vec\rho -D \sum_{v\in V_1} M_{v_0,v}^{-1}  w_v(\vec\rho).
        \ee
        In terms of the notation \eqref{eqn:decompose_lambda} analogous to that of the log VOA modules, we can write 
        \be\label{eqn:s}
        \vec\mu_{\hat w} = {1\over\sqrt{m}}(\vec\rho - \vec s)~,~~ \vec s = \sum_i s_i \vec\omega_i = D\sum_{v\in V_1} M_{v_0,v}^{-1}  w_v(\vec\rho) = -m\sum_{v_i \in V_1} 
        {{\rm sgn}(q_i)\over p_i} w_{v_i}(\vec \rho). 
        \ee

        Note that the expression in (\ref{integrand_3star2}) has the same structure as the triplet character (\ref{eqn:char}) multiplied by $\eta(\tau)^{{\rm rank} G}$, with the only difference being that in the $\widehat Z$-integrand $\tilde \chi_{\hat{w};\underline{\vec{b}}}$ (\ref{integrand_3star2}) we restrict the sum to $D\Lambda$ instead of $\Lambda$. Performing the contour integral in order to obtain the invariant $\widehat Z_{\underline{\vec{b}}}^{G}$  (\ref{hatZ1}), we obtain an answer which again shares the same structure of the singlet characters (\ref{singlet_char}). 
        The structure of the restricted lattice sum  in $\tilde \chi_{\hat{w};\underline{\vec{b}}}$  for a given $\underline{\vec{b}}$ leads to the following interesting result: fixing $\hat{w}$, a sum of $\tilde \chi_{\hat{w};\underline{\vec{b}}}$ over a specific class of $\underline{\vec{b}}$ coincides with, up to an overall factor and a rescaling $\tau\mapsto D\tau$, a generalised character of the triplet algebra {log-${\cal V}_{\bar \Lambda}$}. 

        \begin{thm}\label{thm:combining_into_characters}
          Let $X_\Gamma$ and $D$ be as  in Proposition \ref{prop:main_generalSeif}, with $N=3$.
          Fix a simply-laced Lie group $G$. 
          Let $\underline{\vec{b}}_\ast$, $\hat{w}_\ast\in W^{\otimes 3}$ be such that  $\tilde{\chi}_{\hat{w}_\ast;\underline{\vec{b}}_\ast}\neq  0$. 
          Then 
          \be\label{eqn:integrand_triplet_sum}
          {q^{-D\delta}\over \eta^{\rm rank G}}\,\sum_{\substack{\underline{\vec{b}} =\underline{\vec{b}}_\ast + (\Delta\vec b ,0,0,\dots,0) \\ \Delta\vec b\in \Lambda/D\Lambda}} \tilde \chi_{{\hat{w}_\ast};\underline{\vec{b}}}(D\tau,\vec\xi)  = \chi_{\vec{\mu}_{\hat w_\ast}}(\tau,\vec\xi) 
          \ee
          is given by a (generalised) character  (\ref{eqn:char}) with $\vec{\mu}_{\hat w_\ast}$ as given in (\ref{eq:Awdef}). 

        \end{thm}

        \begin{proof}
          From Proposition \ref{prop:main_generalSeif}, we see that the non-vanishing of 
          $\tilde{\chi}_{\hat{w}_\ast;\underline{\vec{b}}_\ast}$ implies  that there exists a unique $\vec \kappa_{\hat{w}_\ast;\vec{\underline{b}}_\ast}\in \Lambda /D\L$ such that
          \begin{gather}
            \begin{split}
              \tilde \chi_{\hat{w}_\ast;\underline{\vec{b}}_\ast}(\tau,\vec \xi) =  {q^{\delta}\over \Delta(\vec \xi)}  
              \sum_{\substack{{\vec{\tilde \lambda}} \in \vec \kappa_{\hat{w}_\ast;\vec{\underline{b}}_\ast}  +D \Lambda }}  q^{\frac{1}{2D} | \vec \mu_{\hat w} +\sqrt{m}  ( \vec{\tilde{\lambda}}  +\vec\rho )-{1\over \sqrt{m}} \vec{\rho}|^2} \sum_{w\in W}(-1)^{l(w)}
              \,  \ex^{\langle w({\vec{\tilde \lambda}} +\vec \rho) ,  \vec{\xi} \rangle}. 
            \end{split}\end{gather}	
            From	\eqref{eqn:setD} it is clear that the same holds for $\vec{\underline{b}}=\underline{\vec{b}}_\ast + (\Delta\vec b,0,0,\dots,0)$ with
            \be
            \vec\kappa_{\hat{w}_\ast;\vec{\underline{b}}} = \vec\kappa_{\hat{w}_\ast;\vec{\underline{b}}_\ast} + \Delta\vec b . 
            \ee
            It then follows that 
            \begin{multline}
              \sum_{\substack{\underline{\vec{b}} =
              \underline{\vec{b}}_\ast + (\Delta\vec b ,0,0,\dots,0) \\ \Delta\vec b\in \Lambda/D\Lambda}} \tilde \chi_{\hat{w}_\ast;\underline{\vec{b}}}(\tau,\vec\xi)
              =\\  {q^{\delta}\over \Delta(\vec \xi)}  
              \sum_{\substack{{\vec{\tilde \lambda}} \in \vec\kappa_{\hat{w}_\ast;\vec{\underline{b}}} + \Lambda }}  q^{\frac{1}{2D} | \vec \mu_{\hat w} +\sqrt{m}  ( \vec{\tilde{\lambda}}  +\vec\rho )-{1\over \sqrt{m}} \vec{\rho}|^2} \sum_{w\in W}(-1)^{l(w)}
              \,  \ex^{\langle w({\vec{\tilde \lambda}} +\vec \rho) ,  \vec{\xi} \rangle}
            \end{multline}
            which leads to the theorem
            when comparing to the generalised character  (\ref{eqn:char}). 
          \end{proof}

          Restricting to the case with $D=1$, we obtain the following result for the pseudo-spherical Seifert manifolds. 
          \begin{thm}
            \label{3fiber_sphere}
            Fix a simply-laced Lie group $G$. Consider a  three-leg graph  corresponding to a negative Seifert manifold $X_\Gamma =M(b;\{q_i/p_i\}_{i=1,2,3})$ with three exceptional fibers which has integral inverse Euler number, and  let 
            \[ m=  {1\over |\textfrak{e}(X_\Gamma)|}\in \ZZ. 
            \]
            If ${m\over p_i}\in \ZZ$ for $p=1,2,3$, then 
            the integrands of the homological invariants (\ref{hatZ1})
            are equal, up to an overall rational $q$-power and the factor $\eta^{{\rm rank}G}$, to a virtual generalised character (\ref{eqn:char}) of the $(1,m)$ triplet algebra {log-${\cal V}_{\bar \Lambda}(m)$} with given $G$. More precisely, we have either $\chi_{\underline{\vec{b}},\hat{w}} =0$, or
            \be\label{eqn:integrand_triplet}
            {q^{-\delta}\over \eta^{\rm rank G}}\,\tilde \chi_{\hat{w};\underline{\vec{b}}} = \chi_{\vec{\mu}_{\hat w}} 
            \ee
            where $\delta$ is given in (\ref{eqn:delta_A_main}),  $\chi_{\vec{\lambda}'}$  is given as in (\ref{eqn:char}) with 
            \be\label{eqn:which_character}
            \vec{\lambda}'=\vec \mu_{\hat w}={1\over \sqrt{m}}\left(\vec \rho+m\vec A_{\hat {w}}\right) \ee
            (cf. \eqref{eqn:decompose_lambda}, (\ref{dfn:A}), (\ref{eqn:s})).  
          \end{thm}

          \begin{proof}
            Note that $D=1$ in the notation of 
            Proposition \ref{prop:main_generalSeif}. 
            As a result, from the expression (\ref{eqn:set_non_sphere})  we have 
            \be
            {S}_{1,\hat{w};\vec{\underline{b}}} = \left\{\vec \lambda +\vec\rho :\vec \lambda \in \Lambda \right\},
            \ee
            if 
            \be
            \sum_{v'\in V} (M^{-1})_{v, v'} \vec \lambda_{\hat w;v'} \in \Lambda\ee
            for all $v\in V$ \footnote{Note that this is always true when $X_\Gamma$ is a intgeral homological sphere.}, and ${S}_{1,\hat{w};\vec{\underline{b}}}=\emptyset$ otherwise. 
            In the latter case, we have simply $\tilde \chi_{\hat{w};\underline{\vec{b}}} =0$. In the former case we have
            \begin{gather}
              \begin{split}
                \tilde \chi_{\hat{w};\underline{\vec{b}}}(\tau,\vec \xi) & = {q^\delta \over \Delta(\vec{\xi})}  \sum_{ \vec\alpha \in \Lambda} q^{\frac{1}{2} m |\vec\rho+ \vec\alpha+{\vec A_{\hat {w}} }|^2 }\sum_{w\in W} (-1)^{l(w)} 
                \,  \ex^{\langle w(\vec{\rho}+ \vec\alpha),  \vec{\xi} \rangle} \\
                                                                                                             &= 
                                                                                                             {q^\delta \over \Delta(\vec{\xi})}  \sum_{ \vec\alpha \in \Lambda} q^{\frac{1}{2}  |\sqrt{m}\vec\alpha+Q_0\vec\rho +\vec\mu_{\hat w} |^2} \sum_{w\in W} (-1)^{l(w)} 
                                                                                                             \,  \ex^{\langle w(\vec\rho+ \vec\alpha),  \vec{\xi} \rangle}, 
              \end{split}
            \end{gather}
            where 
            \(\vec \mu_{\hat w}\)
            is given as in (\ref{eqn:which_character}). Comparing with (\ref{eqn:char}) establishes the result. 
          \end{proof}	

          From the result of Theorem \ref{3fiber_sphere} and performing the contour integral (\ref{hatZ1}), we obtain the following corollary. 
          \begin{cor}\label{cor:psudo}
            For $X_\Gamma$ and $m$ as in Theorem \ref{3fiber_sphere}, the homological invariants $\widehat{Z}^G_{\underline{\vec{b}}}(\tau;X_\Gamma)$  are, up to an overall factor, given by a virtual generalised character of the $(1,m)$ singlet  algebra {log-${\cal V}_{\bar \Lambda}^0(m)$} as
            \[ C_\Gamma^G(q)^{-1} \,{q^{-\delta}\over \eta^{\rm rank G}}\, \widehat{Z}^G_{\underline{\vec{b}}}(X_\Gamma) \in 
              \left\{ \sum_{\vec{\mu}} a_{\vec{\mu}} \chi^0_{\vec{\mu}} \middle| a_{\vec{\mu}}\in \ZZ \right\}. 
            \] 
          \end{cor}

          From the relation between the generalised characters and the actual characters when $G=SU(2)$  or $G=SU(3)$, and from the special properties of the $G=SU(2)$ characters, we can further conclude 

          \begin{cor}\label{cor:A1A2}~
            \begin{enumerate}
              \item 
                Consider $G=SU(2)$.  
                For any three-leg graph corresponding to a negative Seifert manifold $X_\Gamma$, and for any $\underline{\vec{b}}$, there exists a function  $\widehat{Z}'^G_{\underline{\vec{b}}}(\tau;X_\Gamma)$ on the upper half plane such that 
                \[ C_\Gamma^G(q)^{-1} \,{q^{-\delta}\over \eta^{\rm rank G}}\, \widehat{Z}'^G_{\underline{\vec{b}}}(X_\Gamma) \in 
                  \left\{ \sum_{\vec{\mu}} a_{\vec{\mu}} \chi^0_{\vec{\mu}} \middle| a_{\vec{\mu}}\in \ZZ \right\}, 
                \] 
                and 
                \[\widehat{Z}'^G_{\underline{\vec{b}}}(\tau;X_\Gamma) =\widehat{Z}^G_{\underline{\vec{b}}}(\tau;X_\Gamma)+ \text{finite polynomial in }q. \] 
              \item
                For $X_\Gamma$ and $m$ as in Theorem \ref{3fiber_sphere}, the homological invariants $\widehat{Z}^G_{\underline{\vec{b}}}(\tau;X_\Gamma)$  are given by characters (\ref{singlet_triplet}) of the $(1,m)$ singlet  algebra {log-{${\cal V}^0_{\bar \Lambda}(m)$}}  as
                \[ C_\Gamma^G(q)^{-1} \,{q^{-\delta}\over \eta^{\rm rank G}}\, \widehat{Z}^G_{\underline{\vec{b}}}(\tau;X_\Gamma) \in 
                  \left\{ \sum_{\vec{\mu}} a_{\vec{\mu}} \chi^0_{\vec{\mu}} \middle| a_{\vec{\mu}}\in\ZZ \right\}, 
                \] 
                when $G=SU(3)$.
            \end{enumerate}
          \end{cor}

          \begin{proof}
            To prove {\it 1.}, integrating 
            (\ref{integrand_3star2})  for $G=SU(2)$ we obtain 
            \begin{gather}
              q^{-\delta}\int_{\cal C} d\vec{\xi}
              \,\tilde \chi_{\hat{w};\underline{\vec{b}}}(\tau,\vec \xi) = \Psi_{mD,r}(\tau) - \sum_{k\equiv r \xmod{2mD}} q^{k^2\over 4mD} \left({\rm sgn}\left(k-{m-s\over 2m}\right) -{\rm sgn}(k)\right) 
            \end{gather}
              where $r=-s+m+2m\kappa$ for $\vec \kappa_{\hat{w};\vec{\underline{b}}}  = \kappa \vec \alpha$, with $\vec\alpha$ denoting the simple root of $A_1$. Comparing with (\ref{A1_outofrange}) gives the statement. Similarly, the statement {\it 2.} is a consequence of the identity \eqref{eqn:charID} and  Corollary \ref{cor:psudo}. 
          \end{proof}
          It would be interesting to investigate the relation between the generalised characters and the actual characters for general $G$.

          \subsection{Seifert Manifolds with Four Exceptional Fibers}\label{subsec:4fibers}

          For negative Seifert manifold $X_\Gamma =M(b;\{q_i/p_i\}_{i=1,\dots,4})$ with four exceptional fibers,  Proposition \ref{prop:main_generalSeif} gives that the integrand (\ref{hatZ12})  \begin{gather}
            \begin{split}
              \tilde \chi_{\hat{w};\underline{\vec{b}}}(\tau,\vec \xi) &: = {q^\delta \over \Delta(\vec{\xi})^{2}}  \sum_{ \vec{\ell}\in {S}_{1,\hat{w};\vec{\underline{b}}} }q^{-\frac{1}{2} M_{v_0,v_0}^{-1} |\vec \ell+{\vec A_{\hat {w}} }|^2 }\sum_{w\in W} 
              \,  \ex^{\langle w(\vec{\ell}),  \vec{\xi} \rangle} 
            \end{split}
          \end{gather}
          either vanishes because ${S}_{1,\hat{w};\vec{\underline{b}}}=\emptyset$, or it reads
          \begin{gather}\label{integrand_4star2}\begin{split}
            \tilde \chi_{\hat{w};\underline{\vec{b}}}(\tau,\vec \xi)
&
= q^{\delta}  
\sum_{ \vec{\tilde \lambda} \in   \vec \kappa_{\hat{w};\vec{\underline{b}}} 
+ D\Lambda } 
q^{\frac{1}{2D} | \vec \mu_{\hat w} +\sqrt{m}  \vec{\tilde{\lambda}} |^2} 
\sum_{w\in W} \frac{\ex^{\langle\vec \xi,w\vec{\tilde \lambda}\rangle} }{\Delta^2(\vec\xi)}
          \end{split}
        \end{gather}
        with 
        $$\sqrt{m}\vec\mu_{\hat w}=m \vec A_{\hat w}=-D \sum_{v\in V_1} M_{v_0,v}^{-1}  w_v(\vec\rho) =  -\sum_{v_i \in V_1} 
        {{\rm sgn}(q_i) m\over p_i} w_{v_i}(\vec \rho)$$
        and 
        $\kappa_{\hat{w};\vec{\underline{b}}}$ given as in 
        (\ref{dfn:kappa}).

        \vspace{15pt}

        In particular, for $G=SU(2)$ we can write  
        \[\vec \kappa_{\hat{w};\vec{\underline{b}}} = \kappa_{\hat{w};\vec{\underline{b}}}\,\vec\alpha\in \Lambda ,~~ 
          \sqrt{m} \vec \mu_{\hat w}  = m \vec A_{\hat w} =: {\mu_{\hat w}} \vec\omega
        \]
        with $ \kappa_{\hat{w};\vec{\underline{b}}},~ {\mu_{\hat w}}  \in \ZZ$,  
        and have 
        \be\label{eqn:1_4legged_A1}
        \tilde \chi_{\hat{w};\underline{\vec{b}}}(\tau, \xi) =  q^{\delta}  
        \sum_{n\in \ZZ} q^{\frac{(2m(nD+\kappa_{\hat{w};\vec{\underline{b}}})+\mu_{\hat w})^2}{4Dm}  } 
        \frac{z^{2Dn+ 2\kappa_{\hat{w};\vec{\underline{b}}}}+ z^{-(2Dn+2 \kappa_{\hat{w};\vec{\underline{b}}})} }{(z-z^{-1})^2}.
        \ee
        Moreover, note that the Cauchy principal value integral gives 
        \begin{multline}
          {\rm p.v.}\int {dz\over z} \left( \frac{z^{2k}+ z^{-2k}}{(z-z^{-1})^2} \right) 
          = {\rm p.v.}\int {dz\over z} \left( \frac{z^{2k}-2+ z^{-2k}}{(z-z^{-1})^2} \right)  \\= {\rm CT}_{z} \left( \frac{z^{2k}-2+ z^{-2k}}{(z-z^{-1})^2} \right)  = |k|
        \end{multline}
        and we can therefore adjust the integrand in the following way: 

        \begin{multline}
          \widehat{Z}^{SU(2)}_{\underline{\vec{b}}}(\tau) =C_\Gamma^G(q) \sum_{\hat{w} \in {W}^{\otimes 4}} (-1)^{\ell(\hat{w})} \int_{\cal C} d{\xi}
          \,\tilde \chi_{\hat{w};\underline{\vec{b}}}(\tau, \xi)\\=C_\Gamma^G(q)  \sum_{\hat{w} \in {W}^{\otimes 4}} (-1)^{\ell(\hat{w})} \int_{\cal C} d{\xi}
          \,\tilde \chi'_{\hat{w};\underline{\vec{b}}}(\tau, \xi),
        \end{multline}
        where 
        \begin{gather}
          \begin{split}
            \tilde \chi'_{\hat{w};\underline{\vec{b}}}(\tau,\vec \xi) &: = {q^\delta \over \Delta(\vec{\xi})^{2}}  \sum_{ \vec{\ell}\in {S}_{1,\hat{w};\vec{\underline{b}}} }q^{-\frac{1}{2} M_{v_0,v_0}^{-1} |\vec \ell+{\vec A_{\hat {w}} }|^2 }\sum_{w\in W} 
            \,  (\ex^{\langle w(\vec{\ell}),  \vec{\xi} \rangle}-1) 
          \end{split}
        \end{gather}
        which satisfies  $\tilde \chi'_{\hat{w};\underline{\vec{b}}} =\tilde \chi_{\hat{w};\underline{\vec{b}}} =0$
        when ${S}_{1,\hat{w};\vec{\underline{b}}} =\emptyset$, and otherwise
        \be\label{eqn:1_4legged_A1_modified}
        \tilde \chi'_{\hat{w};\underline{\vec{b}}}(\tau, \xi) =  q^{\delta}  
        \sum_{n\in \ZZ} q^{\frac{(2m(nD+\kappa_{\hat{w};\vec{\underline{b}}})+\mu_{\hat w})^2}{4Dm}  } 
        \frac{z^{2Dn+ 2\kappa_{\hat{w};\vec{\underline{b}}}}-2+ z^{-(2Dn+2 \kappa_{\hat{w};\vec{\underline{b}}})} }{(z-z^{-1})^2}. 
        \ee
        Clearly, we have $\tilde \chi'_{\hat{w};\underline{\vec{b}}} = \tilde \chi'_{-\hat{w};\underline{\vec{b}}}$ where $-\hat{w}$ denotes multiplying  $\hat{w}\in W^{\otimes 4}$ by the non-trivial element of $W\cong \ZZ_2$ diagonally, in accordance with 
        (\ref{eq:diagonal_action}). 
        Comparing with the characters (\ref{eqn:ppprime_tripletchar1}) for the $(p,p')$ triplet model, we see that the form of the integrand $\chi'_{\hat{w};\underline{\vec{b}}}(\tau, \xi)$ is tantalisingly close to that of the $(p,p')$ characters. 

        In the following we will give a proof that this relation to   $(p,p')$ triplet model  always holds when $X_\Gamma =M(b;\{q_i/p_i\}_{i=1,\dots,4})$ is a weakly-negative Brieskorn sphere, and in \S\ref{sec:examples}
        give other non-spherical  examples for which this happens.

        \begin{thm}\label{thm:4fibchimatching}
          Fix the simply-laced Lie group $G=SU(2)$. Consider a   four-leg graph that corresponds to a negative Seifert Brieskorn sphere $X_\Gamma =M(b;\{q_i/p_i\}_{i=1,...,4})$ with four exceptional fibers. 
          Then we have the following identity for the integrand of the homological blocks 
          \begin{multline}\label{pprime_integrand}
            \sum_{\hat{w} \in {W}^{\otimes 4}} (-1)^{\ell(\hat{w})} 
            \tilde \chi'_{\hat{w};\underline{\vec{b}_0}}(\tau, \xi) =
            \\ q^{\delta} \eta(\tau)\sum_{(w_1,w_2,w_3)\in W^{\otimes 3}} (-1)^{\ell({w_1})+\ell({w_2})+\ell({w_3})}\varepsilon_{w_1,w_2,w_3}
            {\rm ch}^{+}_{r,s_{w_1,w_2,w_3}}(\tau,\xi) ,
          \end{multline}
          where $ {\rm ch}^{+}_{r,s}$ is the characters (\ref{eqn:ppprime_tripletchar1}) of the  $(p,p')$ triplet model {log-${\cal V}_{\bar \Lambda}(p,p')$}  with $p=p_4$, $p'=p_1p_2p_3$, and 
          \be\label{s_value_4fiber} r = 1, ~~s_{w_1,w_2,w_3} = \left|p'\sum_{i\neq 4} (-1)^{l(w_i)} {{\rm sgn}(q_i) \over p_i}\right| ,\ \varepsilon_{w_1,w_2,w_3}={\rm sgn}\left(q_{4}\sum_{i\neq 4} (-1)^{l(w_i)} {{\rm sgn}(q_i) \over p_i}\right) \ee
          and similarly for all permutations of $(p_1,p_2,p_3,p_4)$. Beyond these four possible pairs of $(p,p')$ and their images under $p\leftrightarrow p'$, there are no other choices of $(p,p')$ algebras for which the relation  (\ref{pprime_integrand}) between the homological block integrands and triplet characters holds.

        \end{thm}
        \begin{proof}
          Since $D=1$, we have $\vec{\underline{b}}= \vec{\underline{b}}_0$ and 
          $\kappa_{\hat{w};\vec{\underline{b}}}=0$ 
          for all $\hat{w}\in W^{\otimes 4}$ in  
          (\ref{eqn:1_4legged_A1_modified}), and we have
          \be
          \tilde \chi'_{\hat{w};\underline{\vec{b}_0}}(\tau, \xi) =  q^{\delta}  
          \sum_{n\in \ZZ} q^{\frac{(2mn+\mu_{\hat w})^2}{4m}  } 
          \frac{z^{2n}-2+ z^{-2n} }{(z-z^{-1})^2}.
          \ee
          So, if $p, p' \in \ZZ_+$ with $m=pp'$ and given a pair $(\hat w, \hat w')$ such that 
          $(-1)^{\ell(\hat w)+\ell(\hat w') +1}=1$, we can find $r$, $s$ such that 
          \be
          \mu_{\hat w} =  \pm(ps+p'r) , ~~\mu_{\hat w'}  = \pm(ps-p'r)
          \ee
          leading to:
          \be
          (-1)^{\ell(\hat w)}\tilde \chi'_{\hat{w};\underline{\vec{b}_0}}(\tau, \xi)  
          +(-1)^{\ell(\hat{w}')}\tilde \chi'_{\hat{w}';\underline{\vec{b}_0}}(\tau, \xi) = 
          q^{\delta} \eta(\tau) (-1)^{\ell(\hat w)} {\rm ch}^{+}_{r,s}(\tau,\xi) . 
          \label{eq:4fchitchmatch}
          \ee
          Using equations \eqref{M0v_inv}, \eqref{eqn:M00} and the expression of $m$ in terms of the plumbing data, $m=-DM^{-1}_{v_0,v_0}$, we can express $\mu_{\hat{w}}$ and $\mu_{\hat{w}'}$ in terms of the Seifert data as:
          \begin{equation}
            \mu_{\hat{w}} = - \sum_{i=1}^{4}{{\rm sgn}(q_i)m\over p_i}(-1)^{l(w_i)},\ \mu_{\hat{w}'} = - \sum_{i=1}^{4}{{\rm sgn}(q_i)m\over p_i}(-1)^{l(w_i')}.
          \end{equation}
          For $p=p_4$, we have $p'=p_1 p_2 p_3$ and $p{\lvert}{m\over p_i}$ since for $i=1,2,3$, $(p_4,p_i)=1$. As a result, we can choose  $w_i = w_i'$ for $i=1,2,3$ and $w_4= {\id} = - w_4'$, for which case 
          we have 
          \be 
          \mu_{\hat w} = -\varepsilon_{w_1,w_2,w_3}{\rm sgn} (q_4) \left(ps_{w_1,w_2,w_3}+\varepsilon_{w_1,w_2,w_3}(-1)^{\ell(w_4)}p'r\right)
          \ee
          with 
          \be
          \label{eq:rs4f}
          r = 1, ~~ s_{w_1,w_2,w_3} = -{1\over p}\sum_{i\neq 4} (-1)^{l(w_i)} {{\rm sgn}(q_i) m\over p_i}  .
          \ee 

          Apart from a symmetry in the exchange of $p$ with $p'$ one can show that no other splitting of $m$ into $p$ and $p'$ will offer a similar result.
          To see this, note that a pairing of $\hat w$ and $\hat w'$ satisfying 
          $(-1)^{\ell(\hat w)+\ell(\hat w') +1} =1$ must have $w_i = w'_i$ for three or for one $i\in \{1,2,3,4\}$, and the two cases are in fact equivalent since 
          $\tilde \chi'_{\hat{w};\underline{\vec{b}}} = \tilde \chi'_{-\hat{w};\underline{\vec{b}}}$. As a result, we must have either 
          $p\lvert p_4, m/p\lvert p_1p_2p_3$ or a permutation of $(p_1,p_2,p_3,p_4)$ or swapping $p$ and $p'$, from which we conclude that the solutions given in the theorem are the only possibilities. 
        \end{proof}

        Finally, performing the contour integral of \eqref{pprime_integrand} and using the relation (\ref{shift_B}) between the generalised and the true singlet characters, we obtain the following Corollary.
        \begin{cor}\label{cor:4fib}
          Consider $G=SU(2)$.  For any $X_\Gamma$ as in Theorem \ref{thm:4fibchimatching}, let $p, p'$, $s_{w_1,w_2,w_3}$ as in Theorem \ref{thm:4fibchimatching}.  There exists $ \widehat{Z}'^{SU(2)}_{\underline{\vec{b}_0}}(\tau)$  such that 
          \be
          \widehat{Z}'^{SU(2)}_{\underline{\vec{b}_0}}(\tau)  = C_\Gamma^G(q)  q^\delta   \eta(\tau)\sum_{(w_1,w_2,w_3)\in W^{\otimes 3}} (-1)^{\ell({w_1})+\ell({w_2})+\ell({w_3})}
          {\rm ch}^{+,0}_{1,s'_{w_1,w_2,w_3}}(\tau),
          \ee
          where $ {\rm ch}^{+,0}_{1,s'_{w_1,w_2,w_3}}$ are characters of the singlet model {log-${\cal V}^0_{\bar \Lambda}(p,p')$} with 
          $s'_{w_1,w_2,w_3} \equiv \pm s_{w_1,w_2,w_3}~(p')$ 
          and 
          \be
          \widehat{Z}'^{SU(2)}_{\underline{\vec{b}_0}}(\tau) - \widehat{Z}^{SU(2)}_{\underline{\vec{b}_0}}(\tau) =q^\delta \sum_{r\in \ZZ/2pp'} a_r \Psi_{m,r}(\tau)
          \ee
          for some $a_r\in {\mathbb C}$. 
        \end{cor}

        \subsection{\texorpdfstring{$\widehat{Z}$}{\^Z}-invariants with Line Operators}
        \label{subsec:lineop}

        As discussed in \S4 of \cite{Gukov:2017kmk}, one can also consider  homological blocks when Wilson operators, corresponding to half-BPS  line operators in the 3d ${\cal N}=2$ SCFT, are incorporated. Here we consider the homological blocks, modified by Wilson operators $W_{\vec \nu_{v_\ast}}$ associated to a  node $v_\ast \in V$ in the plumbing graph, corresponding to a highest weight representation with highest weight $\vec \nu \in \Lambda^\vee$: 
        \begin{gather}\label{dfn:ZhatW}\begin{split}
&\widehat{Z}^G_{\underline{\vec{b}}}(X_\Gamma, W_{\vec \nu_{v_\ast}};\tau): =\\& C_\Gamma^G(q) \int_{\cal C} d \underline{\vec{\xi}}\,\left( \prod_{v \in V} \Delta(\vec{\xi}_v)^{2 - \deg v} \right)
\chi_{\vec \nu}(\vec \xi_{v_\ast})\sum_{w \in {W}} \sum_{\underline{\vec{\ell}} \in \Gamma_{M,G} +w(\underline{\vec{b}}) } q^{-\frac{1}{2}||  \underline{\vec{\ell}}||^2} 
\left( \prod_{v' \in V} \ex^{\langle \vec{\ell}_{v'},  \vec{\xi}_{v'} \rangle} \right)
        \end{split}\end{gather}
        where  
        \be\label{LieAlgcharacter}
        \chi_{\vec \nu}(\vec \xi) = {1\over {\Delta(\vec \xi)}}\sum_{w \in W}(-1)^{l(w)} \ex^{\langle \vec \xi, w(\vec \rho+ \vec \nu)\rangle} = \sum_{\vec \sigma\in P^+} m^{(\vec \nu)}_{\vec\sigma}  \sum_{w \in W}  \ex^{\langle \vec \xi, w(\vec\sigma)\rangle} 
        \ee
        is the character of the representation of $G$ with highest weight $\vec \nu$, where $m^{(\vec \nu)}_{\vec\sigma}$ is the multiplicity of the weight $\vec\sigma$ in the highest weight module with highest weight $\vec \nu$.
        They are building blocks of the half index of the three-dimensional theory with line operators included, and reduce to the homological blocks without Wilson operators when one sets the highest weight $\vec\nu=0$. 

        In what follows, as in the rest of the paper, we mainly focus on Seifert manifolds, and consider Wilson operators associated with the end nodes, the central node, and the intermediate nodes of the plumbing graph. We will see that in each of the three cases,
        the Wilson operator leads to 
        a different modification of the relation to log VOA characters discussed in the earlier part of the section. 
        \vspace{15pt}
        \newpage
        \noindent{\bf Wilson Operator at an End Node}\\

        First consider including a Wilson operator associated with an end node, say  $v_1\in V_1$.     While integrating over the end 
        vertices with $v\neq v_1$ gives
        \be
        \int d {\vec{\xi}_v} \,\Delta(\vec{\xi}_v) \,  \ex^{\langle \vec \xi_v, \vec \ell_v\rangle} = \left(\sum_{w\in W} (-1)^{l(w)} \delta_{\vec \ell_v, -w(\vec\rho)} \right),
        \ee
        integrating over $\vec\xi_{v_1}$ gives 
        \be
        \int d {\vec{\xi}} \,\Delta(\vec{\xi}) \chi_{\vec \nu}(\vec \xi)  \ex^{\langle \vec \xi, \vec \ell_{}\rangle} = \int  d {\vec{\xi}_{}}\sum_{w\in W} (-1)^{l(w)}  \ex^{\langle \vec \xi_v, \vec \ell 
      _v+ w(\vec \rho +\vec \nu )\rangle} = \sum_{w\in W} (-1)^{l(w)}  \delta_{\vec \ell_v, -w(\vec\rho+\nu )} . 
        \ee
        In summary,  define for all $v\in V_1$ 
        \be
        \vec \rho_v = \begin{cases} \vec \rho + \vec\nu & ~{\rm if}~v=v_1 \\ \vec \rho & ~{\rm if}~v\neq v_1  \end{cases}, 
        \ee  
        then we have $ \sum_{w\in W} (-1)^{l(w)}  \delta_{\vec \ell_v, -w(\vec\rho_v)}$ as the result of integration of $\vec \xi_v$ for all end nodes $v\in V_1$. 

        As a result, it is easy to check that the following statement, which is completely analogous to Proposition \ref{prop:main_generalSeif},  holds for $\widehat{Z}^G_{\underline{\vec{b}}}(W_{\vec \nu_{v_1}})$.
        Namely, 
        \begin{gather}\begin{split}
          \widehat{Z}^G_{\underline{\vec{b}}}(X_\Gamma, W_{\vec \nu_{v_1}};\tau)&=C_\Gamma^G(q) \sum_{\hat{w} \in {W}^{\otimes N}} (-1)^{\ell(\hat{w})} \int_{\cal C} d\vec{\xi}
        \,\tilde \chi_{\hat{w};\underline{\vec{b}}}(\tau,\vec \xi) \\  \end{split}
      \end{gather}
      and the 
      integrand is 	either $\tilde \chi_{\hat{w};\underline{\vec{b}}}=0$, or 
      there exists a unique  $\vec \kappa_{w,\hat{w};\vec{\underline{b}}}  \in  \Lambda^\vee/D\Lambda$ such that 
      \begin{gather}\label{eqn:integrand_N_Wilson}
        \begin{split}
          \tilde \chi_{\hat{w};\underline{\vec{b}}}(\tau,\vec \xi) & = {q^\delta \over \Delta(\vec{\xi})^{N-2}}\times \\& \sum_{ \vec\lambda \in \Lambda }q^{\frac{1}{2D}  |\sqrt{m}(D \vec \lambda +\vec \kappa_{\hat{w};\vec{\underline{b}}}  + \varepsilon(N)\vec\rho ) + {m\vec A_{\hat {w}} \over \sqrt{m}}|^2}\sum_{w\in W} (-1)^{Nl(w)} 
          \,  \ex^{\langle w(D \vec \lambda +\vec \kappa_{\hat{w};\vec{\underline{b}}}  + \varepsilon(N)\vec\rho ),  \vec{\xi} \rangle} ,
        \end{split}
      \end{gather}
      with  $\delta$ and $\vec A_{\hat {w}}$ as given in (\ref{eqn:delta_A_main}), but now with $\vec \rho$ replaced by $\vec\rho_v$, namely 
      \begin{gather}\label{dfn:AW}\begin{split}
        \vec A_{\hat {w} } &={-1\over M_{v_0,v_0}^{-1}} \sum_{v\in V_1} M_{v_0,v}^{-1}  w_v(\vec\rho_v) = -\sum_{v_i \in V_1} 
        {{\rm sgn}(q_i)\over p_i} w_{v_i}(\vec \rho_{v_i}) 
        , \\
        \delta &=   \sum_{\substack{v\in V_1}}
        {|\vec\rho_v|^2\over 2}  \left({(M_{v_0,v}^{-1})^2\over M_{v_0,v_0}^{-1}} -  M_{v,v}^{-1} \right) =\sum_{i=1}^N {|\vec\rho_{v_i}|^2\over 2}\left( {1\over p_i q_i} - {\theta^{(i)}_{\ell_i-1}\over \theta^{(i)}_{\ell_i}}\right) ,
      \end{split}\end{gather}
      and $\vec \kappa_{\hat{w};\vec{\underline{b}}}  \in  \Lambda^\vee/D\Lambda$ such that  

      \begin{gather}\notag\begin{split}
        {S}_{1,\hat w;\vec{\underline{b}}} &= \left\{ \vec \ell_0 \,\Bigg\lvert \Big(\vec{\ell}_0, \, -w_1(\vec \rho+\vec \nu),\underbrace{ -w_2(\vec \rho), \cdots, -w_N( \vec\rho)}_{N-1} , \underbrace{0, \cdots, 0}_{|V|-N-1} \Big) \in \Gamma_{M,G} +\underline{\vec b} \right\} \\ 
                                                     & = \{\vec \kappa_{\hat{w};\vec{\underline{b}}} + D\vec \lambda +\varepsilon(N) \vec\rho :\vec \lambda \in \Lambda \}.
        \end{split}\end{gather}
        Consequently, we have the following proposition.  

        \begin{prop}
          \label{prop:Wil-End}
          Let $v_\ast\in V_1$ be one of the  end nodes.
          \begin{enumerate}
            \item In the case of $N=3$, 
              Theorem \ref{thm:combining_into_characters} and Theorem \ref{3fiber_sphere},  and as a result Corollaries  \ref{cor:psudo}-\ref{cor:A1A2} all hold for  $\widehat{Z}^G_{\underline{\vec{b}}}(X_\Gamma,W_{\vec \nu_{v_\ast}};\tau)$, with the only difference being given by \eqref{dfn:AW}.
            \item 
              In the case of $N=4$, 
              Theorem  \ref{thm:4fibchimatching}  and as a result Corollary  \ref{cor:4fib} hold for  $\widehat{Z}^G_{\underline{\vec{b}}}(X_\Gamma,W_{\vec \nu_{v_\ast}};\tau)$, with the only difference being given by \eqref{dfn:AW}, and as a result
              $$
              s_{w_1,w_2,w_3} = -{1\over p}\sum_{i\neq 4} (-1)^{l(w_i)} {\rm sgn}(q_i) (1+\nu)
              { m\over p_i}, 
              $$
              where the highest weight is given by $\vec\nu = \nu \vec\omega$.
          \end{enumerate}
        \end{prop}

        \vspace{15pt} 
        \noindent{\bf Wilson Operator at the Central Node}\\

        Next we consider  a Wilson operator associated with the central node $v_0$. 
        The integral over $\vec \xi_{v_0}$ in (\ref{dfn:ZhatW}) reads
        \be
        \int d\vec \xi  {1\over \Delta(\vec \xi)^{N-2} } \ex^{\langle \vec \xi, \vec \ell\rangle}   \chi_{\vec \nu}(\vec \xi)  
        = 
        \sum_{\vec \sigma \in P^+} m^{(\vec \nu)}_{\vec\sigma} \int d\vec \xi  {1\over \Delta(\vec \xi)^{N-2} } \ex^{\langle \vec \xi, \vec \ell\rangle} \sum_{w\in W}  \ex^{\langle \vec \xi, w(\vec\sigma)\rangle} 
        \ee
        where we  have dropped the subscript in $\vec \xi_{v_0}$. 
        As a result, we see that the statement in Proposition \ref{prop:main_generalSeif} is modified in a very simple way by the inclusion of a Wilson operator associated to the central node: 
        \begin{gather}\label{eqn:zhat_wilson_central}\begin{split}
          \widehat{Z}^G_{\underline{\vec{b}}}(X_\Gamma, W_{\vec \nu_{v_0}};\tau) &=C_\Gamma^G(q)\sum_{\vec \sigma\in P^+} m^{(\vec \nu)}_{\vec\sigma}  \sum_{\hat{w} \in {W}^{\otimes N}} (-1)^{\ell(\hat{w})} \int_{\cal C} d\vec{\xi}
        \,  \sum_{w\in W} \ex^{\langle \vec \xi, w(\vec\sigma)\rangle} \,\tilde \chi_{\hat{w};\underline{\vec{b}}}(\tau,\vec \xi)   \end{split}
  \end{gather}
  where $\tilde \chi_{\hat{w};\underline{\vec{b}}}$ is given as in (\ref{eqn:integrand_N}). 
  The two changes in the integrand on  right-hand side upon including Wilson operators are 1) a sum over the weights $\vec\sigma$ that appear in the corresponding highest weight module, and 2)  a multiplication by a factor $\ex^{\langle \vec \xi, \vec\sigma\rangle}$.
  These changes alter but do not destroy the form of the relation between the homological blocks and the generalised singlet and triplet characters, in the case of negative Seifert manifolds with three singular fibers, and we have the following proposition. 

  \begin{prop}\label{prop:Wil-Center} Let $X_\Gamma$, $D$, $\chi_{\vec \mu_{\hat w}}$ be as in Theorem \ref{thm:combining_into_characters}. Then  
    \begin{equation}\label{eqn:Prop_zhat_wilson_central}
      \widehat{Z}^G_{\underline{\vec{b}}}(X_\Gamma, W_{\vec \nu_{v_0}};\tau) =C_\Gamma^G(q) \sum_{\hat{w} \in {W}^{\otimes N}} (-1)^{\ell(\hat{w})} \int_{\cal C} d\vec{\xi}
      \tilde \chi_{\hat{w};\underline{\vec{b}}}^{(\vec\nu)}(\tau,\vec \xi) , 
    \end{equation}
    where a particular sum of the  integrands is given by a polynomial in $z_i$ times a generalised character of the log VOA algebra {log-${\cal V}_{\bar \Lambda}$} :  
    \be
    {q^{-D\delta}\over \eta^{\rm rank G}}\,\sum_{\substack{\underline{\vec{b}} =\underline{\vec{b}}_\ast + (\Delta\vec b ,0,0,\dots,0) \\ \Delta\vec b\in \Lambda/D\Lambda}} \tilde \chi^{(\vec\nu)}_{{\hat{w}_\ast};\underline{\vec{b}}}(D\tau,\vec\xi)  =\sum_{\vec \sigma\in P^+} m^{(\vec \nu)}_{\vec\sigma} 
    \sum_{w\in W} \ex^{\langle \vec \xi, w(\vec\sigma)\rangle} \chi_{\vec{\mu}_{\hat w_\ast}}(\tau,\vec\xi) ,
    \ee
    and analogously for Theorem \ref{3fiber_sphere}. 
    As a result,
    given that the original homological block satisfies 
    \be
    C_\Gamma^G(q)^{-1} \,{q^{-\delta}\over \eta^{\rm rank G}}\, \widehat{Z}^G_{\underline{\vec{b}}}(X_\Gamma;\tau) = \sum_{\vec{\mu}} a_{\vec{\mu}} \chi^0_{\vec{\mu}} 
    \ee
    for some $a_{\vec \mu}\in \ZZ$ as  in  Corollary \ref{cor:psudo}, 
    the homological block with Wilson operator is given by 
    \be\label{eqn:shift_wilson}
    C_\Gamma^G(q)^{-1} \,{q^{-\delta}\over \eta^{\rm rank G}}\, \widehat{Z}^G_{\underline{\vec{b}}}(X_\Gamma, W_{\vec \nu_{v_0}};\tau)=
    \sum_{\vec \sigma\in P^+} m^{(\vec \nu)}_{\vec\sigma}  \sum_{w \in W} 
    \sum_{\vec{\mu}} 
    a_{\vec{\mu}} \chi^0_{\vec\mu-\sqrt{m} w(\vec\sigma)}. 
    \ee
  \end{prop}

  Note that this is precisely the ``shifting" phenomenon that has been observed for the special case of a Lens space example  in \cite{Gukov:2017kmk}.

  \vspace{15pt}
  \noindent{\bf Wilson Operator at an Intermediate Node}\\

  Finally we will consider the case when a Wilson operator associated to an intermediate node in a star graph, say $v_{int} \in V_2$, is added. 
  In this case we have 
  \be
  \int d\vec \xi   \ex^{\langle \vec \xi, \vec \ell\rangle}   \chi_{\vec \nu}(\vec \xi)  
  =  \sum_{\vec \sigma \in P^+} m^{(\vec \nu)}_{\vec\sigma}  \sum_{w'\in W}
  \int d\vec \xi  \ex^{\langle \vec \xi, \vec \ell + w'(\vec\sigma)\rangle}   =   \sum_{\vec \sigma \in P^+} m^{(\vec \nu)}_{\vec\sigma}  \sum_{w'\in W} \delta_{\vec \ell_v, -w'(\vec\sigma)} .
  \ee
  As a result, the relevant sets are now
  \begin{multline}\nonumber
    {S}_{w,w',\hat w;\vec{\underline{b}}} : = \\ \left\{ \vec \ell_0 \,\Bigg\lvert \Big(\vec{\ell}_0,\underbrace{-w_1(\vec \rho),  -w_2(\vec \rho), \cdots, -w_N( \vec\rho)}_{N} , \underbrace{0, \cdots, 0, -w'(\vec\sigma),0,\cdots,0}_{|V|-N-1} \Big) \in \Gamma_{M,G} +w(\underline{\vec b}) \right\},  \\ 
    \end{multline}
    where $ -w'(\vec\sigma)$ is the vector  corresponding to the vertex $v_{int}$, 
    satisfying
    \be
    {S}_{w,w',\hat w;\vec{\underline{b}}} \cong {S}_{w_\ast w,w_\ast w',w_\ast \hat w;\vec{\underline{b}}}, 
    \ee
    with the isomorphism given by  $\vec \ell_0\mapsto w_\ast (\ell_0)$. 
    Putting everything together, we get 
    \begin{equation}\label{eq:ZhatWmid}
      \widehat{Z}^G_{\underline{\vec{b}}}(X_\Gamma,W_{\vec \nu_{v_1}};\tau)=C_\Gamma^G(q) \sum_{\hat{w} \in {W}^{\otimes N}} (-1)^{\ell(\hat{w})} \sum_{w'\in W} \int_{\cal C} d\vec{\xi}
      \,\tilde \chi_{\hat{w},w';\underline{\vec{b}}}(\tau,\vec \xi) 
    \end{equation}
    where  $\tilde \chi_{\hat{w},w';\underline{\vec{b}}}$  vanishes when ${S}_{1,w',\hat w;\vec{\underline{b}}} =\emptyset$,  and is given by 
    \begin{multline}
      \tilde \chi_{\hat{w},w';\underline{\vec{b}}}(\tau,\vec \xi)   = {q^{\delta_{\hat w,w'}} \over \Delta(\vec{\xi})^{N-2}} \\ \times\sum_{ \vec\lambda \in \Lambda }q^{\frac{1}{2D}  |\sqrt{m}(D \vec \lambda +\vec \kappa_{\hat{w},w';\vec{\underline{b}}}  + \varepsilon(N)\vec\rho ) + {m\vec A_{\hat {w},w'} \over \sqrt{m}}|^2}\\\times\sum_{w\in W} (-1)^{Nl(w)} 
      \,  \ex^{\langle w(D \vec \lambda +\vec \kappa_{\hat{w},w';\vec{\underline{b}}}  + \varepsilon(N)\vec\rho ),  \vec{\xi} \rangle}
    \end{multline}
    where $\kappa_{\hat{w},w';\vec{\underline{b}}} \in \Lambda^\vee/D\Lambda$ satisfying
    \be
    {S}_{1,w',\hat w;\vec{\underline{b}}}  = 	 \{\vec \kappa_{\hat{w},w';\vec{\underline{b}}} + D\vec \lambda +\varepsilon(N) \vec\rho :\vec \lambda \in \Lambda \}. 
    \ee
    Moreover, compared to the case without Wilson operators (\ref{dfn:A}),  the data for the homological blocks are modified as 
    \begin{align}\label{eq:AwWilMid}
      \vec A_{\hat {w},w'}  & = \vec A_{\hat {w}} -  { M_{v_0,v_{int}}^{-1} \over M_{v_0,v_0}^{-1}}w'(\vec\sigma) \\
      \delta_{\hat w,w'}& = \delta +{|\vec\sigma|^2\over 2}  \left({(M_{v_0,v_{int}}^{-1})^2\over M_{v_0,v_0}^{-1}} -  M_{v_{int},v_{int}}^{-1} \right) \notag \\&\hspace{3cm}+ \sum_{v\in V_1}\langle w_v(\vec \rho), w'(\vec\sigma)\rangle \left({M_{v_0,v}^{-1}M_{v_0,v_{int}}^{-1}\over M_{v_0,v_0}^{-1}} -  M_{v_{int},v}^{-1} \right)
    \end{align}
    in the notation of Proposition \ref{prop:main_generalSeif}. 
    As a result, we see that Theorem \ref{thm:combining_into_characters} and Theorem \ref{3fiber_sphere} hold analogously for $\tilde \chi_{\hat{w},w';\underline{\vec{b}}}(\tau,\vec \xi) $ in this case, with the modification of the data as given above. Note that the $q$-power $\delta_{\hat w,w'}$ is no longer independent of $\hat w$ and $w'$ due to the last term. Consequently, the  homological block is no longer given by a sum of log VOA characters up to an overall factor. For instance, the statement of Corollary \ref{cor:psudo} gets modified into 
    \be\label{eqn:intermediate_Wilson}
    {1\over C_\Gamma^G(q)} \,{1\over \eta^{{\rm rank}G}} \widehat{Z}^G_{\underline{\vec{b}}}(W_{\vec \nu_{v_1}};\tau)\in 
    \left\{ \sum_{\vec{\mu}}
    q^{\delta_{\vec\mu} }\, a_{\vec{\mu}}  \chi^0_{\vec{\mu}} \middle| a_{\vec{\mu}}\in \ZZ , ~\delta_{\vec\mu}\in \RR \right\}, 
    \ee
    and similarly for Corollary \ref{cor:A1A2}.

    \section{Quivers, Nahm Sums, and Fermionic Characters}
    \label{sec:fermionic}

    The main theme of this paper is the identification of $q$-series invariants of 3-manifolds with VOA characters. In this section, we show how this identification can be used to produce new fermionic forms of characters for logarithmic VOAs:
    \be
    \chi (q) = 
    \sum_{d_i\geq 0} \frac{1}{(q)_{\boldsymbol{d}}}\,q^{ \frac{1}{2} \boldsymbol{d} \cdot C \cdot \boldsymbol{d}+ (\text{terms linear in } \boldsymbol{d})}.
    \label{fermform}
    \ee
    We will focus on $\widehat Z^G$ with $G=SU(2)$, and often drop the superscript for notational convenience. 

    The main idea is to use the enumerative interpretation of $\widehat Z$-invariants and their connection with quiver (COHA) generating series. This perspective on $\widehat Z$-invariants allows one to write the invariants of knot and link complements in the fermionic form (also known as the quiver form or Nahm sum form). Then, it is easy to see that surgery formulae preserve this form, so that $\widehat Z^G_{\underline{\vec{b}}} (X,\tau)$ for closed 3-manifolds can be expressed as a linear combination of fermionic characters \eqref{fermform}. In fact, this is the same linear combination of characters of log VOAs we saw earlier, so that individual terms can be matched and provide (new) fermionic expressions for (combinations of) log VOA characters.

    The fermionic form of VOA characters has a long history and goes as far back as  the original work of Hans Bethe \cite{Bethe1931}. Its modern form is rooted in the relation between 2d CFTs and vertex algebras, on the one hand, and their massive integrable deformations, on the other hand.\footnote{A hallmark of an integrable QFT produced via deformation of a 2d CFT is the factorized scattering, illustrated in Figure~\ref{fig:scattering} and directly related to the structure of the matrix $C$ in \eqref{fermform}.} Underlying this integrable structure are quantum groups, Bethe ansatz equations, Yangian symmetry, and various other symmetries discovered and studied throughout the 1980s, by the Zamolodchikov brothers \cite{Zamolodchikov:1978xm,Zamolodchikov:1989hfa,Zamolodchikov:1989cf}, by the Leningrad school \cite{Smirnov:1990pr,Reshetikhin:1989qg} where the explicit form \eqref{fermform} appeared in connection with the Kostka polynomials \cite{MR869576,MR869577}, and by many other groups.

    \begin{figure}
      \centering
      \includegraphics[width=0.5\textwidth]{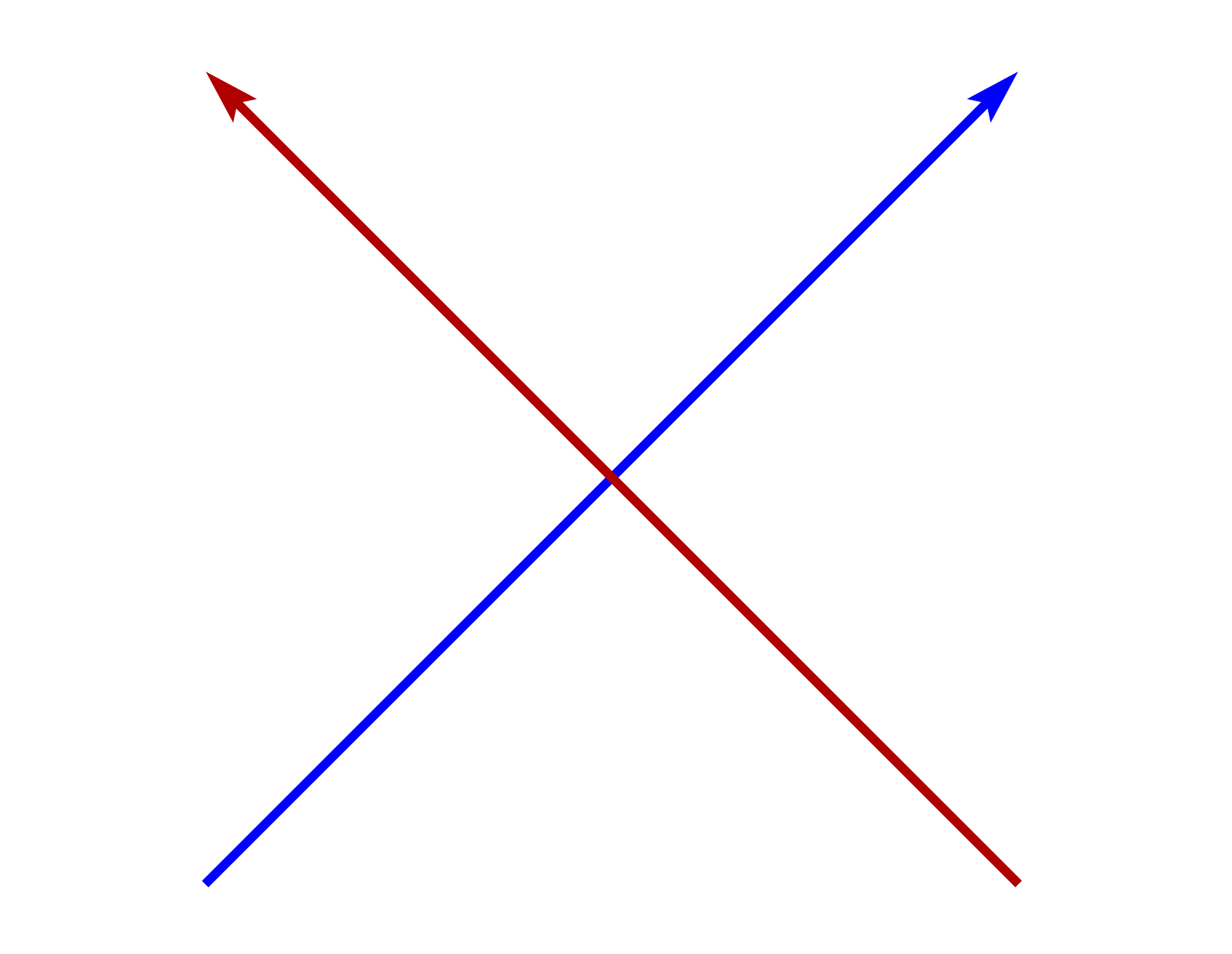}
      \caption{The factorized scattering of quasiparticles in an integrable deformation of 2d CFT. Different types of quasiparticles correspond to rows and columns of the matrix $C$ in \eqref{fermform}.}
      \label{fig:scattering}
    \end{figure}

    The name for the structure of the $q$-series \eqref{fermform} was coined by Barry McCoy and the Stony Brook group in the early 1990s \cite{Kedem:1993ve,Kedem:1992jv,Kedem:1993ze,Dasmahapatra:1993pw}, where various properties of \eqref{fermform} were studied, including the Rogers-Ramanujan type identities as manifestations of the bosonization/fermionization \cite{MR1415414,Berkovich:1997ht}. In all these developments, the matrix $C$ in the $q$-series \eqref{fermform} is one of the main ingredients; in particular, the size of $C$ is equal to the number of quasi-particles in the integrable massive deformation of a CFT.

    On the other hand, the study of finite-size effects and the Bethe ansatz equations mentioned earlier involve dilogarithms and dilogarithm identities \cite{MR1362962,Nahm:1992sx}, which is one of the main conceptual reasons why one should expect these developments from the early 1990s to have direct relation to the 3d-3d correspondence where dilogarithms and dilogarithm identities also play a key role. Building on these developments, in 1995 E.~Frenkel and A.~Szenes \cite{MR1266736} proposed a relation to algebraic K-theory, and in 2004 W. Nahm \cite{Nahm:2004ch} proposed a relation between the fermionic characters \eqref{fermform} and the Bloch group. Sometimes the fermionic expressions \eqref{fermform} are called Nahm sums, though it is not clear whether this term was intended to be the same or different from the original notion introduced by B.~McCoy and others.

    More recently, the fermionic form \eqref{fermform} was also studied in the context of logarithmic CFTs and vertex algebras \cite{Flohr:2006id,Feigin:2007sp,Warnaar:2007mq,Adamovic:2007cs,Feigin:2008sg,Flohr:2013dma}, closely related to the subject of the present paper.

      In a completely different line of developments --- that has origins in the relation between quantum topology and enumerative geometry --- a very interesting correspondence between knots and (equivalence classes of) quivers was proposed about four years ago \cite{Kucharski:2017ogk}. In this correspondence, studied {\it e.g.} in \cite{MR4230391,Panfil:2018sis,Ekholm:2018eee,Ekholm:2020lqy,Kucharski:2020rsp,Jankowski:2021flt}, the combinatorial and algebraic data associated to the quiver encodes wealth of information about the knot. For example, the structure of the quiver matrix $C$ has a close relation to the structure of the (uncolored) superpolynomials and triply-graded knot homology \cite{MR2253002}: the size of $C$ is equal to the number of generators in the reduced superpolynomial / HOMFLY-PT homology, so that the diagonal values of $C$ are given by the homological $t$-degrees of the generators, {\it etc.} One of the main statements in this correspondence is that all HOMFLY-PT polynomials of a knot $K$ coloured by Young tableaux that consists of a single row (column) can be combined in a generating series
    \be
    \sum_{n=0}^{\infty} P_n (a,q) x^n = \sum_{d_1, \ldots, d_m \ge 0} q^{\frac{1}{2} \sum_{i,j} C_{i,j} d_i d_j}
    \prod_{i=1}^m \frac{(-1)^{t_i d_i} q^{l_i d_i} a^{a_i d_i} x^{d_i}}{(q;q)_{d_i}}
    \label{quiverHOMFLY}
    \ee
    where the right-hand side is the motivic DT generating series of the corresponding quiver. Again, the key ingredient is the quiver matrix $C$, accompanied by the vectors $\Bf{t}$, $\Bf{a}$, and $\Bf{l}$.

    One of the conceptual underpinnings of the knot-quiver correspondence, from which it derives its strength and leads to expressions like \eqref{quiverHOMFLY}, has to do with the HOMFLY-PT variable $a$. Namely, instead of working with quantum group invariants of a fixed rank, passing to HOMFLY-PT homology or polynomial invariants allows one to see a much richer structure, associated with enumerative and BPS invariants. In fact, an even richer structure can be uncovered\footnote{The reason for this is that, with both $a$- and $t$-gradings manifest, one can see a rich structure of differentials acting on the space of (refined) BPS states, and this plays an important role in many aspects of the knot-quiver correspondence.} by incorporating another variable $t$, which keeps track of the homological grading, but we will not need it here.

    What will be important to us is that many lessons from the physical interpretation of the HOMFLY-PT homology in terms of BPS states and enumerative invariants extend to $\widehat{Z}$-invariants. In the language of enumerative geometry, this means that the structure of differentials $d_N$ and spectral sequences can be transferred from knot conormal Lagrangian submanifolds in Calabi-Yau geometry to knot complements. The first hints for this were seen \cite{Gukov:2016gkn} already for very simple closed 3-manifolds, such as $X = S^3$. Very recently, it was realized that the knot-quiver correspondence can be extended to $\widehat{Z}$-invariants of knot complements as well \cite{Kucharski:2020rsp,Ekholm:2021irc}.

    In particular, this means that $\widehat{Z}$-invariants of knot complement also can be written in the fermionic form {\it a la} \eqref{fermform} or \eqref{quiverHOMFLY}. Then, it quickly follows that the same is true for closed 3-manifolds obtained via surgery operations. This last step is very simple, so let us start by explaining why the surgery operation preserves the fermionic form. Suppose (in the conventions of \cite{Ekholm:2021irc}):
    \be
    F_K (x,q) = \frac{1}{\eta (\tau)} \sum_{\Bf{d}} q^{\frac{1}{2} \Bf{d}\cdot C_K \cdot \Bf{d} + A_K} \, \frac{(-1)^{\Bf{d} \cdot \Bf{t}} q^{\Bf{d} \cdot \Bf{l}_K} x^{|\Bf{d}|+c} }{(q)_{\Bf{d}}}
    \label{FKquiver}
    \ee
    where $c \in \mathbb{Z}$, and $|\Bf{d}| = \sum_i d_i$. Then, the surgery formula \cite{Gukov:2019mnk} gives:\footnote{Note that in \eqref{eqn:LaplaceT}, the convention of $F_K$ follows \cite{Ekholm:2021irc}. In the convention of \cite{Gukov:2019mnk}, $F_K$ is normalized in a way that its Laplace transformation reads: $\widehat{Z}_0(S^3_{-1/r}(K)) = \mathcal{L}^{(0)}_{-1/r} \left( (x^{\tfrac{1}{2r}} - x^{-\tfrac{1}{2r}}) F_K(x,q) \right)$.}
    \begin{equation}\label{eqn:LaplaceT}
      \widehat{Z}_0(S^3_{-1/r}(K)) = \mathcal{L}^{(0)}_{-1/r} \left( x(x^{\tfrac{1}{2r}} - x^{-\tfrac{1}{2r}})(x^{\tfrac{1}{2}} - x^{-\tfrac{1}{2}}) F_K(q^{-1}x,q) \right),
    \end{equation}
    where $\mathcal{L}^{(0)}_{-1/r}: x^u \mapsto q^{ru^2}$. As a result, we obtain the following formula for $\widehat{Z}_0$:
    \be
    \widehat Z \big( S^3_{-1/r} (K) \big) = \frac{1}{\eta (\tau)} \sum_{\Bf{d} \geq 0} q^{\frac{1}{2} \Bf{d}\cdot C \cdot \Bf{d} + A}
    \Big[ q - q^{2r+2r|\Bf{d}|} - q^{2+2|\Bf{d}|} + q^{3+2r+2(1+r)|\Bf{d}|} \Big]
    \frac{(-1)^{\Bf{d} \cdot \Bf{t}} q^{\Bf{d} \cdot \Bf{l}}}{(q)_{\Bf{d}}}
    \label{Zhatferm}
    \ee
    with
    $$
    C = C_K + 2rE \,, \qquad
    l_i = (l_K)_i + r-2+2rc \,, \qquad
    A = A_K + \frac{r}{4} + \frac{1}{4r} - \frac{3}{2} + rc^2.
    $$
    Here, $E$ is the matrix where every entry is 1. More generally, if
    $$
    F_K (x,q) = \frac{1}{\eta (\tau)} \sum_{\Bf{d}} q^{\frac{1}{2} \Bf{d}\cdot C_K \cdot \Bf{d} + A_K} \, \frac{(-1)^{\Bf{d} \cdot \Bf{t}} q^{\Bf{d} \cdot \Bf{l}_K} x^{\Bf{d}  \cdot \Bf{n} + c } }{(q)_{\Bf{d}}}
    $$
    we can still use \eqref{Zhatferm}, with the more general modified norm, $|\Bf{d}| = \sum_i n_i d_i$, and matrix $E$ whose $ij$ entry is
    \be
    E_{ij} \; = \; n_i n_j.
    \ee

    Various methods for producing the quiver/fermionic form of $F_K (x,q)$ can be found in \cite{Ekholm:2021irc}. They include non-trivial dualities, {\it e.g.} to enumerative geometry or to 3d $\mathcal{N}=2$ theories $T[X]$, as well as more direct diagrammatic techniques, {\it e.g.} based on the $R$-matrix approach \cite{Park:2020edg,Park:2021ufu}. In the rest of this section, we illustrate how, starting with such expressions for knot complements, one can obtain analogous fermionic/quiver forms for $\widehat{Z}$-invariants of closed 3-manifolds, namely the so-called ``small'' surgeries on various knots:
    \begin{equation}
      X \; = \; S^3_{-1/r} (K).
    \end{equation}

    \subsection{New Fermionic Forms Related to {log-\texorpdfstring{${\cal V}_{\bar\Lambda}^0(m)$}{V\_L\^{}0(m)}}}

    A simple infinite family of examples can be obtained by considering surgeries on a torus knot. For concreteness, and to avoid dealing with Spin$^c$ structures, we can consider the so-called ``small surgeries'' on $(s,t)$ torus knots, which all give Brieskorn spheres:
    \begin{equation}
      X \; = \; S^3_{-1/r} (T_{s,t}) \; = \; \Sigma (s,t,rst+1).
    \end{equation}
    In fact, as a warm up, we can start with the simplest knot of all, the right-handed trefoil knot ${\bf 3_1^r} = T_{2,3}$:
    \begin{equation}
      X \; = \; S^3_{-1/r} ({\bf 3_1^r}) \; = \; \Sigma (2,3,6r+1).
    \end{equation}
    For the right-handed trefoil knot, $F_{{\bf 3_1^r}} (x,q)$ can be written in the quiver form \eqref{FKquiver} with \cite{Ekholm:2021irc}:
    $$
    C_{{\bf 3_1^r}} \; = \;
    \begin{pmatrix}
      0 & 1 & 0 & 0 \\
      1 & 0 & 1 & 0 \\
      0 & 1 & 1 & 0 \\
      0 & 0 & 0 & 1
    \end{pmatrix}
    \qquad
    \begin{array}{rcl}
      A & = & 0, \\
      \boldsymbol{t} & = & (0,0,1,1), \\
      \boldsymbol{l}_K & = & (1,2,\frac{3}{2},\frac{3}{2})
    \end{array}
    $$
    Therefore, after the surgery we get a linear combination of fermionic forms \eqref{Zhatferm} with
    \be
    C \; = \;
    \begin{pmatrix}
      2r & 2r+1 & 2r & 2r \\
      2r+1 & 2r & 2r+1 & 2r \\
      2r & 2r+1 & 2r+1 & 2r \\
      2r & 2r & 2r & 2r+1
    \end{pmatrix}
    \label{Crtrefoil}
    \ee
    On the other hand, we know that the same $q$-series invariants can be written as linear combinations of false theta functions and characters of the $(1,m)$ singlet algebra {log-${\cal V}^0_{\bar \Lambda_{A_1}}(m)$} with $m = 36r + 6$ and central charge $c = 1 - \frac{(36r+5)^2}{6r+1}$  \cite{Cheng:2018vpl,Gukov:2019mnk}:
    $$
    \widehat{Z}_0 \big( S^3_{-1/r} (\mathbf{3}_1^r) \big) \; = \;
    \frac{1}{\eta(\tau)}
    \left( 
      \Psi_{36r+6,6r-5} 
      - \Psi_{36r+6,6r+7}
      - \Psi_{36r+6,30r-1}
      + \Psi_{36r+6,30r+11}
    \right)
    $$
    Comparing this expression with \eqref{Zhatferm} and rearranging the sum, we obtain a new fermionic form for the {log-${\cal V}^0_{\bar \Lambda_{A_1}}(m)$} (virtual) character \eqref{A1_outofrange} with $4 \times 4$ matrix \eqref{Crtrefoil}.

    \subsection{Log VOA Characters and Torus Knots}

    Let us consider $(-1/r)$-surgeries on left-handed torus knots $T_{2,2p+1}$. The resultant manifold is a Brieskorn sphere $\Sigma(p_1,p_2,p_3)$ with
    \begin{equation}
      (p_1,p_2,p_3) = (2,2p+1,2(2p+1)r-1).
    \end{equation}
    Its $\widehat{Z}$-invariant  is given by \cite{Gukov:2017kmk,Chung:2018rea}:
    \begin{equation}
      \begin{gathered}
        \widehat{Z}_0\left(S^3_{-1/r}(T_{2,2p+1})\right) \sim \frac{1}{\eta(\tau)} \sum_{r \in S} {\rm sgn}(r) \Psi_{m,r} \\
        m = 2(2p+1)(2(2p+1)r-1), \quad S = \left\{ m + m \sum_{i=1}^3 \frac{\epsilon_i}{p_i} \; | \; \epsilon_i = \pm 1 \right\}, \\
        {\rm sgn}(r) = -\epsilon_1 \epsilon_2 \epsilon_3,
      \end{gathered}
    \end{equation}
    where $\sim$ means ``up to an overall rational power of $q$ and a multiplicative constant.'' 
    The relation between $ \widehat{Z}_0\left(S^3_{-1/r}(T_{2,2p+1})\right) $ and {log-${\cal V}_{\bar \Lambda_{A_1}}^0(m)$} characters is covered in Theorem \ref{3fiber_sphere}. 
    These expressions were anticipated by Hikami \cite{Hikami2004}, in a way that does not explain topological invariance, {\it i.e.} invariance under the Kirby moves.

    For $T_{2,2p+1}$, we have the following $(2p+2)\times(2p+2)$ matrix $C_K$ and auxiliary vectors $\mathbf{t}$ and $\Bf{l}$:
    \begin{equation}
      C_K = \begin{pmatrix} \mathbf{I}_{2p} - \mathbf{D} & \mathbf{-1}^T & \mathbf{0}^T \\
        \mathbf{-1} & 1 & 0 \\
        \mathbf{0} & 0 & 0 
        \end{pmatrix}, \quad \begin{cases} 
        \Bf{t} = {\rm diag}(C_K) \\
        \Bf{l}_K = 1 + \Bf{a} - \frac{1}{2} {\rm diag}(C_K)
      \end{cases}
      \end{equation}
      where $\mathbf{a} = (1,0,1,0,\cdots,1,0,0,0)$. The matrix $\mathbf{I}_{2p}$ is a $2p\times 2p$ identity matrix, and $\mathbf{D}$ is given by $\mathbf{D}_{ij} = {\rm min}(i,j)-1$. Lastly, $\mathbf{-1}$ and $\mathbf{0}$ are $2p$-dimensional row vectors with all entries given by $-1$ resp. $0$.

      To make use of \eqref{Zhatferm}, we must compute $C = C_K + 2r E$, where the matrix $E$ is given by $E_{ij} = n_i n_j$. For left-handed torus knots, the vector $\Bf{n}$ is given by:
      \begin{equation}
        \Bf{n} = (1,1,3,3,\cdots,2p-1,2p-1,1,1).
      \end{equation}
      Let us consider the examples of $p = 1,2$. For $p=1$ (left-handed trefoil), we will need following data to compute the $\widehat{Z}$-invariants via \eqref{Zhatferm}:
      \begin{equation}
        C = \begin{pmatrix} 1 + 2r & 2r & -1+2r & 2r \\ 
          2r & 2r & -1+2r & 2r \\
          -1+2r & -1+2r & 1+2r & 2r \\
          2r & 2r & 2r & 2r
          \end{pmatrix}, \quad \begin{cases} \Bf{t} = (1,0,1,0) \\
          \Bf{l} = (r-\frac{1}{2},r-1,r-\frac{3}{2},r-1)
        \end{cases}.
        \end{equation}

        For $p=2$, left-handed $T_{2,5}$ knot, the data for the fermionic form of the $\widehat{Z}$-invariants are given by:
        \begin{equation}
          \begin{gathered}
            C = \begin{pmatrix} 1+2r & 2r & 6r & 6r & -1+2r & 2r \\
              2r & 2r & -1+6r & -1+6r & -1+2r & 2r \\
              6r & -1+6r & -1+18r & -2+18r & -1+6r & 6r \\
              6r & -1+6r & -2+18r & -2+18r & -1+6r & 6r \\
              -1+2r & -1+2r & -1+6r & -1+6r & 1+2r & 2r \\
              2r & 2r & 6r & 6r & 2r & 2r
            \end{pmatrix}, \\
            \Bf{t} = (1,0,-1,-2,1,0), \quad \Bf{l} = \left( r-\tfrac{1}{2}, r-1, r+\tfrac{1}{2}, r, r-\tfrac{3}{2}, r+1 \right).
            \end{gathered}
          \end{equation}
          In these infinite families of examples we obtain new fermionic forms for linear combinations of {log-${\cal V}^0_{\bar \Lambda_{A_1}}(m)$} characters, suggesting that an actual log VOA associated to $X = S^3_{-1/r}(T_{2,2p+1})$ is likely to be an extension of {log-${\cal V}^0_{\bar \Lambda_{A_1}}(m)$}. While it would be interesting to study this further, we should emphasize that there are many infinite families of 3-manifolds for which we can write the fermionic forms but can not offer any connection to {log-${\cal V}^0_{\bar \Lambda_{A_1}}(m)$} or other familiar log VOAs. We conclude this section by writing explicitly several such families.

          \subsection{Fermionic characters from \texorpdfstring{${\bf 4_1}$}{4\_1}, \texorpdfstring{${\bf 5_2}$}{5\_2}, and \texorpdfstring{${\bf 6_2}$}{6\_2} knots}

          Next, we turn to other classes of knots, twist knots ${\bf 4_1}$, ${\bf 5_2}$, and ${\bf 6_2}$. Small surgeries on these knots produce infinitely many distinct hyperbolic manifolds. Another advantage of this family of examples is that all manifolds that result from small surgeries have $H_1 (X) = 0$, so that there is a unique Spin$^c$ structure, and we do not need to worry about the labels of $\widehat{Z}$-invariants.

          For the ${\bf 4_1}$ knot, its $F_K$ invariant is given by (with $a = q^2$ specialization):
          \begin{equation}
            \begin{gathered}
              F_{\bf 4_1}(x,q) = \sum_{d_1, \cdots, d_6 \geq 0} \left(-q^{\frac{1}{2}}\right)^{\Bf{d}.C_K.\Bf{d}} \prod_{i=1}^6\frac{x_i^{d_i}}{(q)_{d_i}} \\
              x_1 = x_2 = x_3 = q x, \quad x_4 = x_5 = x_6 = q^{\frac{3}{2}}x.
            \end{gathered}
          \end{equation}
          For ${\bf 4_1}$ knot, there are some equivalent choices of the matrix $C_K$ (see \S 4.2.1 of \cite{Ekholm:2021irc}).
          To be explicit, we will make the following choice:
          \begin{equation}
            C_K = \begin{pmatrix} 0 & 0 & 0 & 0 & 0 & 0 \\
              0 & 0 & -1 & -1 & 0 & 0 \\
              0 & -1 & 0 & 0 & 1 & 0 \\
              0 & -1 & 0 & 1 & 1 & 0 \\
              0 & 0 & 1 & 1 & 1 & 0 \\
              0 & 0 & 0 & 0 & 0 & 1
            \end{pmatrix}.
            \end{equation}
            With \eqref{eqn:LaplaceT}, we obtain the following data $(C, \Bf{t}, \Bf{l})$ to compute the $\widehat{Z}$-invariants. 
            \begin{equation}
              \begin{gathered}
                C = \begin{pmatrix} 2r & 2r & 2r & 2r & 2r & 2r \\
                  2r & 2r & 2r-1 & 2r-1 & 2r & 2r \\
                  2r & 2r-1 & 2r & 2r & 2r+1 & 2r \\
                  2r & 2r-1 & 2r & 2r+1 & 2r+1 & 2r \\
                  2r & 2r & 2r+1 & 2r+1 & 2r+1 & 2r \\
                  2r & 2r & 2r & 2r & 2r & 2r+1
                \end{pmatrix} \\
                \Bf{t} = (0,0,0,1,1,1) \\
                \Bf{l} = (r-1,r-1,r-1,r-\tfrac{1}{2},r-\tfrac{1}{2},r-\tfrac{1}{2}).
                \end{gathered}
              \end{equation}
              Similarly, for the ${\bf 5_2}$ knot,
              the matrix $C_K$ can be written as follows:
              \begin{equation}
                C_K = \begin{pmatrix} 
                  0 & 0 & 0 & -1 & 0 & 0 & -1 & 0\\
                  0 & 1 & 0 & 0 & 1 & 1 & -1 & 0\\
                  0 & 0 & 1 & 0 & 1 & 1 & -1 & 0 \\
                  -1 & 0 & 0 & 0 & 1 & 1 & -1 & 0\\
                  0 & 1 & 1 & 1 & 2 & 2 & -1 & 0\\
                  0 & 1 & 1 & 1 & 2 & 3 & -1 & 0\\
                  -1 & -1 & -1 & -1 & -1 & -1 & 1 & 0 \\
                  0 & 0 & 0 & 0 & 0 & 0 & 0 & 0
                \end{pmatrix}.
              \end{equation}
              In this notation, we find:
              \begin{equation}
                \begin{gathered}
                  \Bf{t} = (0,1,1,0,0,1,1,0), \\
                  \Bf{l}_K = (1,\frac{3}{2},\frac{1}{2},2,1,\frac{5}{2},\frac{1}{2},1).
                \end{gathered}
              \end{equation}
              As a result, we obtain the following data $(C,\Bf{t},\Bf{l})$:
              \begin{equation}
                \begin{gathered}
                  C = \begin{pmatrix}
                    2r & 2r & 2r & 2r-1 & 2r & 2r & 2r-1 & 2r \\
                    2r & 2r+1 & 2r & 2r & 2r +1 & 2r +1 & 2r -1 & 2r \\
                    2r & 2r & 2r+1 & 2r & 2r+1 & 2r+1 & 2r-1 & 2r \\
                    2r-1 & 2r & 2r & 2r & 2r+1 & 2r+1 & 2r-1 & 2r \\
                    2r & 2r+1 & 2r+1 & 2r+1 & 2r+2 & 2r+2 & 2r-1 & 2r \\
                    2r & 2r+1 & 2r+1 & 2r+1 & 2r+2 & 2r+3 & 2r-1 & 2r \\
                    2r-1 & 2r-1 & 2r-1 & 2r-1 & 2r-1 & 2r-1 & 2r+1 & 2r \\
                    2r & 2r & 2r & 2r & 2r & 2r & 2r & 2r
                  \end{pmatrix} \\
                  \Bf{t} = (0,1,1,0,0,1,1,0) \\
                  \Bf{l} = (r-1,r-\frac{1}{2}, r-\frac{3}{2}, r, r-1, r+\frac{1}{2},r-\frac{3}{2},r-1).
                \end{gathered}
              \end{equation}
              For the ${\bf 6_2}$ knot, its $F_K$ invariant is given by:
              \begin{equation}
                F_{{\bf 6_2}}(x,q) = -q^{-1}x^{-2} \sum_{\Bf{d}} (-q^{\frac{1}{2}})^{\Bf{d} \cdot C_K \cdot \Bf{d}} \frac{x^{\Bf{n} \cdot \Bf{d}} q^{\Bf{l}_K \cdot \Bf{d}}}{(q)_{\Bf{d}}},
              \end{equation}
              where the matrix $C_K$ and the auxiliary vector $\Bf{l}_K$ are given by:
              \begin{small}
                $$
                \begin{gathered}
                  C_K =
                  \begin{pmatrix}
                    1 & 0 & 1 & 0 & 1 & 0 & 1 & 0 & 1 & 0 & 1 & 0 & 1 & 0 & 0 & 0 \\
                    0 & 0 & 1 & 0 & 1 & 0 & 1 & 0 & 1 & 0 & 1 & 0 & 1 & 0 & 0 & 0 \\
                    1 & 1 & -1 & -1 & -1 & -1 & -1 & 0 & 0 & 0 & -1 & 0 & -1 & -1 & 0 & -1 \\
                    0 & 0 & -1 & 0 & -1 & -1 & -2 & -1 & 0 & 0 & -1 & 0 & -1 & -1 & 0 & -1 \\
                    1 & 1 & -1 & -1 & 2 & 1 & 0 & -1 & 0 & -1 & 1 & 0 & 0 & 0 & 0 & 0 \\
                    0 & 0 & -1 & -1 & 1 & 1 & 0 & -1 & 0 & -1 & 1 & 0 & 0 & 0 & 0 & 0 \\
                    1 & 1 & -1 & -2 & 0 & 0 & 2 & 1 & 0 & 0 & 0 & -1 & 1 & 0 & 0 & 0 \\
                    0 & 0 & 0 & -1 & -1 & -1 & 1 & 1 & 0 & 0 & 0 & -1 & 1 & 0 & 0 & 0 \\
                    1 & 1 & 0 & 0 & 0 & 0 & 0 & 0 & 1 & 0 & 0 & 0 & 0 & 0 & 0 & 0 \\
                    0 & 0 & 0 & 0 & -1 & -1 & 0 & 0 & 0 & 0 & 0 & 0 & 0 & 0 & 0 & 0 \\
                    1 & 1 & -1 & -1 & 1 & 1 & 0 & 0 & 0 & 0 & 1 & 0 & 0 & 0 & 0 & 0 \\
                    0 & 0 & 0 & 0 & 0 & 0 & -1 & -1 & 0 & 0 & 0 & 0 & 0 & 0 & 0 & 0 \\
                    1 & 1 & -1 & -1 & 0 & 0 & 1 & 1 & 0 & 0 & 0 & 0 & 1 & 0 & 0 & 0 \\
                    0 & 0 & -1 & -1 & 0 & 0 & 0 & 0 & 0 & 0 & 0 & 0 & 0 & 0 & 0 & 0 \\
                    0 & 0 & 0 & 0 & 0 & 0 & 0 & 0 & 0 & 0 & 0 & 0 & 0 & 0 & 1 & 0 \\
                    0 & 0 & -1 & -1 & 0 & 0 & 0 & 0 & 0 & 0 & 0 & 0 & 0 & 0 & 0 & 0
                  \end{pmatrix} \\
                  \Bf{l}_K = (\tfrac{1}{2}, 0, \tfrac{1}{2}, 0, 0, -\tfrac{1}{2}, 0, -\tfrac{1}{2}, \tfrac{1}{2}, 0, \tfrac{1}{2}, 0, \tfrac{1}{2}, 0, \tfrac{1}{2}, 0).
                \end{gathered}
                $$
              \end{small}
              \noindent
              After a $(-1/r)$-surgery along the ${\bf 6_2}$ knot, the relevant data $(C,\Bf{t},\Bf{l})$ to compute $\widehat{Z}$ can be explicitly given as:
              $$
              \begin{gathered}
                C = 
                \scalebox{0.52}[0.6]{$
                  \begin{pmatrix}
                    2r+1 & 2r & 4r+1 & 4r & 2r+1 & 2r & 2r+1 & 2r & 2r+1 & 2r & 2r+1 & 2r & 2r+1 & 2r & 2r & 2r \\
                    2r & 2r & 4r+1 & 4r & 2r+1 & 2r & 2r+1 & 2r & 2r+1 & 2r & 2r+1 & 2r & 2r+1 & 2r & 2r & 2r \\
                    4r+1 & 4r+1 & 8r-1 & 8r-1 & 4r-1 & 4r-1 & 4r-1 & 4r & 4r & 4r & 4r-1 & 4r & 4r-1 & 4r-1 & 4r & 4r-1 \\
                    4r & 4r & 8r-1 & 8r & 4r-1 & 4r-1 & 4r-2 & 4r-1 & 4r & 4r & 4r-1 & 4r & 4r-1 & 4r-1 & 4r & 4r-1 \\
                    2r+1 & 2r+1 & 4r-1 & 4r-1 & 2r+2 & 2r+1 & 2r & 2r-1 & 2r & 2r-1 & 2r+1 & 2r & 2r & 2r & 2r & 2r \\
                    2 r& 2 r& 4 r& -1 + 4 r& -1 + 2 r& -1 + 2 r& 1 + 2 r& 1 + 2 r& 2 r& -1 + 2 r& 1 + 2 r& 2 r& 2 r& 2 r& 2 r& 2 r \\
                    1 + 2 r& 1 + 2 r& -1 + 4 r& -2 + 4 r& 2 r& 2 r& 2 + 2 r& 1 + 2 r& 2 r& 2 r& 2 r& -1 + 2 r& 1 + 2 r& 2 r& 2 r& 2 r\\
                    2 r& 2 r& 4 r& -1 + 4 r& -1 + 2 r& -1 + 2 r& 1 + 2 r& 1 + 2 r& 2 r& 2 r& 2 r& -1 + 2 r& 1 + 2 r& 2 r& 2 r& 2 r\\
                    1 + 2 r& 1 + 2 r& 4 r& 4 r& 2 r& 2 r& 2 r& 2 r& 1 + 2 r& 2 r& 2 r& 2 r& 2 r& 2 r& 2 r& 2 r\\
                    2 r& 2 r& 4 r& 4 r& -1 + 2 r& -1 + 2 r& 2 r& 2 r& 2 r& 2 r& 2 r& 2 r& 2 r& 2 r& 2 r& 2 r\\
                    1 + 2 r& 1 + 2 r& -1 + 4 r& -1 + 4 r& 1 + 2 r& 1 + 2 r& 2 r& 2 r& 2 r& 2 r& 1 + 2 r& 2 r& 2 r& 2 r& 2 r& 2 r\\
                    2 r& 2 r& 4 r& 4 r& 2 r& 2 r& -1 + 2 r& -1 + 2 r& 2 r& 2 r& 2 r& 2 r& 2 r& 2 r& 2 r& 2 r\\
                    1 + 2 r& 1 + 2 r& -1 + 4 r& -1 + 4 r& 2 r& 2 r& 1 + 2 r& 1 + 2 r& 2 r& 2 r& 1 + 2 r& 2 r& 1 + 2 r& 2 r& 2 r& 2 r\\
                    2 r& 2 r& -1 + 4 r& -1 + 4 r& 2 r& 2 r& 2 r& 2 r& 2 r& 2 r& 2 r& 2 r& 2 r& 2 r& 2 r& 2 r\\
                    2 r& 2 r& 4 r& 4 r& 2 r& 2 r& 2 r& 2 r& 2 r& 2 r& 2 r& 2 r& 2 r& 2 r& 1 + 2 r& 2 r\\
                    2 r& 2 r& -1 + 4 r& -1 + 4 r& 2 r& 2 r& 2 r& 2 r& 2 r& 2 r& 2 r& 2 r& 2 r& 2 r& 2 r& 2 r 
                \end{pmatrix}$} \\
                \Bf{t} = (1,0,1,0,0,1,0,1,1,0,1,0,1,0,1,0) \\
                \Bf{l} = \scalebox{0.75}{$(r-\frac{3}{2}, r-2, r-\frac{3}{2}, r-2, r-2, r-\frac{5}{2}, r-2, r-\frac{5}{2}, r-\frac{3}{2}, r-2, r-\frac{3}{2}, r-2, r-\frac{3}{2}, r-2, r-\frac{3}{2}, r-2)$}.
              \end{gathered}
              $$
              As we already mentioned earlier, it would interesting to understand which logarithmic VOAs can have (linear combinations of) characters that match these fermionic forms.

              In all of the above examples, the VOAs have the effective central charge
              \begin{equation}
                c_{\text{eff}} \; = \; 1
                \label{ceffone}
              \end{equation}
              which can be obtained directly from the growth of the integer coefficients in the $q$-expansion, {\it cf.} \cite{Ekholm:2021irc}. This result is somewhat surprising since generic vertex algebras, logarithmic or non-logarithmic, can have (and do have!) other values of $c_{\text{eff}}$. When the fermionic forms are available, $c_{\text{eff}}$ can be obtained from the Thermodynamic Bethe Ansatz controlled by matrix $C$, and comes out to be a sum of special values of dilogarithms, evaluated at the roots of the Bethe equations \cite{Zamolodchikov:1989cf,Nahm:1992sx,MR1362962}.

              Therefore, a generic $C$-matrix would produce values of $c_{\text{eff}}$ very different from~1, and it is not completely clear at present why VOA charactes that come from $\widehat{Z}$-invariants \eqref{ZhatVOA} have this property. There must be something special about fermionic forms that come from $\widehat{Z}$-invariants. It seems that in our examples \eqref{ceffone} has to do with the rank of $G=SU(2)$. However, as far as we know, the analogue of \eqref{ceffone} has not been tested for groups of higher rank, nor do we know if it holds for all closed 3-manifolds, even when $G=SU(2)$. All these questions are excellent subject for future work.

              \section{Examples}
\label{sec:examples}

In this section we explicitly demonstrate the relations found in this paper between the $\widehat Z$-invariant, its integrand $\tilde \chi$, and log VOA characters through concrete examples. 
This section is divided in three subsections, reflecting the three cases analyzed in \S \ref{sec:homoblocks}.

In \S \ref{subsec:3fib-ex}, we focus on Seifert manifolds with three exceptional fibers.
We provide a general analysis for integral homology spheres (which we sometimes refer to as ``spherical manifolds" or as ``Brieskorn spheres'') and four explicit examples:  a spherical manifold ($D=|\det(M)|=1$), a pseudo-spherical manifold ($1=D<|\det(M)|$) and two non-spherical cases with $1<D<|\det(M)|$ and $1<D=|\det(M)|$ respectively.
These examples will serve to demonstrate numerically the results of Theorems \ref{thm:combining_into_characters}, \ref{3fiber_sphere}, and thereby Corollaries  \ref{cor:psudo} and \ref{cor:A1A2}.

In \S\ref{subsec:4fib-ex} we focus on Seifert manifolds with four exceptional fibers.
We provide two examples to demonstrate Theorem \ref{thm:4fibchimatching} and thus Corollary \ref{cor:4fib}.
The first example presented in \S \ref{subsec:4fib-ex} is that of a spherical manifold.
With the second example, we also demonstrate that results similar to Theorem \ref{thm:4fibchimatching} and Corollary \ref{cor:4fib} apply to select cases of pseudo-spherical manifolds.

Lastly, in \S \ref{subsec:line-ope-ex} we provide examples of Seifert manifolds with three exceptional fibers with Wilson line operator insertions. As mentioned in \S\ref{subsec:lineop}, for Seifert manifolds the line operators can be associated to the central, intermediate and end nodes of the plumbing graph. 
  Examples with operator insertions at end nodes for spherical, pseudo-spherical and non spherical Seifert manifolds will serve to demonstrate Proposition \ref{prop:Wil-End}. 
Proposition \ref{prop:Wil-Center} is then demonstrated on a spherical and a pseudo-spherical example, by insertion of a Wilson operator at the central node.
We will conclude this section by providing examples of insertions of Wilson line operators at intermediate nodes of a spherical and a non-spherical Seifert manifold and relating the $\widehat{Z}$-invariant to linear combination of log VOA  characters, as in  \eqref{eqn:intermediate_Wilson}.

In this section we will make explicit reference to the Weyl groups of $SU(2)$ and $SU(3)$.
The Weyl group of $SU(2)$ is isomorphic to $\mathbb{Z}_2$.
We will write the elements of Weyl length zero and one respectively as $\id$ and $-\id$.
Elements of the Weyl group of $SU(3)$, isomorphic to $D_3$, will be written in terms of group elements $a, b$, corresponding to reflections with respect to planes orthogonal to the two simple roots $\vec\alpha_1,\ \vec\alpha_2$.
The $SU(3)$ Weyl group is then given by $\{\id, a, b, ab, ba , aba=bab\}$.
Furthermore, we will always use the weight basis when we explicitly write the vectors as tuples. 
For instance,  we write $\vec s = \sum_i s_i \vec{\omega}_i =: (s_1,s_2) $.

All $q$-series and topological quantities were computed using ``pySeifert'': a computational toolkit written using Sage \cite{sagemath,pySeifert}.

\subsection{Seifert Manifolds with Three Exceptional Fibers}\label{subsec:3fib-ex}

In this section we give examples of Seifert manifolds with three exceptional fibers.
After presenting general results for spherical manifolds, we compute $\widehat{Z}$-invariant integrands, $\tilde{\chi}_{\hat{w};\vec{\underline{b}}}$, for spherical, pseudo-spherical and non-spherical manifolds. 
We then verify their relation to the log VOA characters as described in Theorems  \ref{thm:combining_into_characters}, \ref{3fiber_sphere} and Corollaries  \ref{cor:psudo} and \ref{cor:A1A2}.

\vspace{15pt}
\noindent
{\bf General Spherical Examples}

The simplest class of examples we can consider is that of Brieskorn spheres  $X_\Gamma = M(-1;\left\{{q_i}/{p_i}\right\}_{i=1,2,3}) =: \Sigma(p_1,p_2,p_3)$, where we have
\be\label{eqn:Brieskornpq}
\sum_i {q_i\over p_i} = 1-{1\over p_1 p_2 p_3}. 
\ee
In these cases, we have
\begin{equation}
m=-{1\over \textfrak{e}}=p_1p_2p_3,~~ M^{-1}_{v_0,v_i}=\frac{m}{p_i}{\rm sgn}(q_i) ~{\rm for }~ v_i\in V_1.
\end{equation}
As explained in \S \ref{sec:homoblocks}, spherical examples have a single $\vec{\underline{b}}_0$, whose expression can be found in equation \eqref{hom_sphere_b}.
From \eqref{eqn:which_character} and \eqref{eqn:s}, we obtain 
\begin{equation}
\vec{s}=\sum_{i=1}^{3}s_i\vec{\omega}_i = \sum_{k=1}^{3}\frac{p_1p_2p_3}{p_k}{w}_{k}(\vec{\rho})\,  {\rm sgn}(q_k)
\end{equation}
for a given choice of  $\hat w\in W^{\otimes 3}$. 
\begin{table}
\begin{subtable}[h]{\textwidth}
\small
\begin{tabular}{ll}
$(s_{1},s_{2})$ for $\hat w$ with $(-1)^{\ell(\hat{w})}=1$& \\\toprule
$\left(p_{1} p_{2} - p_{1} p_{3} - p_{2} p_{3},\,-2 \, p_{1} p_{2} - p_{1} p_{3} - p_{2} p_{3}\right)$ & $\left(-2 \, p_{1} p_{2} - p_{1} p_{3} - p_{2} p_{3},\,p_{1} p_{2} - p_{1} p_{3} - p_{2} p_{3}\right)$ \\
$\left(p_{1} p_{2} - p_{1} p_{3} - p_{2} p_{3},\,p_{1} p_{2} - p_{1} p_{3} - p_{2} p_{3}\right)$ & $\left(-p_{1} p_{2} + p_{1} p_{3} - p_{2} p_{3},\,-p_{1} p_{2} - 2 \, p_{1} p_{3} - p_{2} p_{3}\right)$ \\
$\left(-p_{1} p_{2} + p_{1} p_{3} - p_{2} p_{3},\,2 \, p_{1} p_{2} - 2 \, p_{1} p_{3} - p_{2} p_{3}\right)$ & $\left(2 \, p_{1} p_{2} + p_{1} p_{3} - p_{2} p_{3},\,-p_{1} p_{2} - 2 \, p_{1} p_{3} - p_{2} p_{3}\right)$ \\
$\left(-p_{1} p_{2} - 2 \, p_{1} p_{3} - p_{2} p_{3},\,-p_{1} p_{2} + p_{1} p_{3} - p_{2} p_{3}\right)$ & $\left(-p_{1} p_{2} - 2 \, p_{1} p_{3} - p_{2} p_{3},\,2 \, p_{1} p_{2} + p_{1} p_{3} - p_{2} p_{3}\right)$ \\
$\left(2 \, p_{1} p_{2} - 2 \, p_{1} p_{3} - p_{2} p_{3},\,-p_{1} p_{2} + p_{1} p_{3} - p_{2} p_{3}\right)$ & $\left(p_{1} p_{2} - p_{1} p_{3} - p_{2} p_{3},\,-2 \, p_{1} p_{2} + 2 \, p_{1} p_{3} - p_{2} p_{3}\right)$ \\
$\left(-2 \, p_{1} p_{2} - p_{1} p_{3} - p_{2} p_{3},\,p_{1} p_{2} + 2 \, p_{1} p_{3} - p_{2} p_{3}\right)$ & $\left(p_{1} p_{2} - p_{1} p_{3} - p_{2} p_{3},\,p_{1} p_{2} + 2 \, p_{1} p_{3} - p_{2} p_{3}\right)$ \\
$\left(p_{1} p_{2} + 2 \, p_{1} p_{3} - p_{2} p_{3},\,-2 \, p_{1} p_{2} - p_{1} p_{3} - p_{2} p_{3}\right)$ & $\left(-2 \, p_{1} p_{2} + 2 \, p_{1} p_{3} - p_{2} p_{3},\,p_{1} p_{2} - p_{1} p_{3} - p_{2} p_{3}\right)$ \\
$\left(p_{1} p_{2} + 2 \, p_{1} p_{3} - p_{2} p_{3},\,p_{1} p_{2} - p_{1} p_{3} - p_{2} p_{3}\right)$ & $\left(-p_{1} p_{2} + p_{1} p_{3} - p_{2} p_{3},\,-p_{1} p_{2} + p_{1} p_{3} - p_{2} p_{3}\right)$ \\
$\left(-p_{1} p_{2} + p_{1} p_{3} - p_{2} p_{3},\,2 \, p_{1} p_{2} + p_{1} p_{3} - p_{2} p_{3}\right)$ & $\left(2 \, p_{1} p_{2} + p_{1} p_{3} - p_{2} p_{3},\,-p_{1} p_{2} + p_{1} p_{3} - p_{2} p_{3}\right)$ \\
\end{tabular}
\vspace{0.2cm}
\end{subtable}
\hfill
\begin{subtable}[H!]{\textwidth}
\small
\begin{tabular}{ll}
$(s_{1},s_{2})$ for $\hat w$ with $(-1)^{\ell(\hat{w})}=-1$& \\\toprule
$\left(-p_{1} p_{2} - p_{1} p_{3} - p_{2} p_{3},\,-p_{1} p_{2} - p_{1} p_{3} - p_{2} p_{3}\right)$ & $\left(-p_{1} p_{2} - p_{1} p_{3} - p_{2} p_{3},\,2 \, p_{1} p_{2} - p_{1} p_{3} - p_{2} p_{3}\right)$ \\
$\left(2 \, p_{1} p_{2} - p_{1} p_{3} - p_{2} p_{3},\,-p_{1} p_{2} - p_{1} p_{3} - p_{2} p_{3}\right)$ & $\left(p_{1} p_{2} + p_{1} p_{3} - p_{2} p_{3},\,-2 \, p_{1} p_{2} - 2 \, p_{1} p_{3} - p_{2} p_{3}\right)$ \\
$\left(-2 \, p_{1} p_{2} + p_{1} p_{3} - p_{2} p_{3},\,p_{1} p_{2} - 2 \, p_{1} p_{3} - p_{2} p_{3}\right)$ & $\left(p_{1} p_{2} + p_{1} p_{3} - p_{2} p_{3},\,p_{1} p_{2} - 2 \, p_{1} p_{3} - p_{2} p_{3}\right)$ \\
$\left(p_{1} p_{2} - 2 \, p_{1} p_{3} - p_{2} p_{3},\,-2 \, p_{1} p_{2} + p_{1} p_{3} - p_{2} p_{3}\right)$ & $\left(-2 \, p_{1} p_{2} - 2 \, p_{1} p_{3} - p_{2} p_{3},\,p_{1} p_{2} + p_{1} p_{3} - p_{2} p_{3}\right)$ \\
$\left(p_{1} p_{2} - 2 \, p_{1} p_{3} - p_{2} p_{3},\,p_{1} p_{2} + p_{1} p_{3} - p_{2} p_{3}\right)$ & $\left(-p_{1} p_{2} - p_{1} p_{3} - p_{2} p_{3},\,-p_{1} p_{2} + 2 \, p_{1} p_{3} - p_{2} p_{3}\right)$ \\
$\left(-p_{1} p_{2} - p_{1} p_{3} - p_{2} p_{3},\,2 \, p_{1} p_{2} + 2 \, p_{1} p_{3} - p_{2} p_{3}\right)$ & $\left(2 \, p_{1} p_{2} - p_{1} p_{3} - p_{2} p_{3},\,-p_{1} p_{2} + 2 \, p_{1} p_{3} - p_{2} p_{3}\right)$ \\
$\left(-p_{1} p_{2} + 2 \, p_{1} p_{3} - p_{2} p_{3},\,-p_{1} p_{2} - p_{1} p_{3} - p_{2} p_{3}\right)$ & $\left(-p_{1} p_{2} + 2 \, p_{1} p_{3} - p_{2} p_{3},\,2 \, p_{1} p_{2} - p_{1} p_{3} - p_{2} p_{3}\right)$ \\
$\left(2 \, p_{1} p_{2} + 2 \, p_{1} p_{3} - p_{2} p_{3},\,-p_{1} p_{2} - p_{1} p_{3} - p_{2} p_{3}\right)$ & $\left(p_{1} p_{2} + p_{1} p_{3} - p_{2} p_{3},\,-2 \, p_{1} p_{2} + p_{1} p_{3} - p_{2} p_{3}\right)$ \\
$\left(-2 \, p_{1} p_{2} + p_{1} p_{3} - p_{2} p_{3},\,p_{1} p_{2} + p_{1} p_{3} - p_{2} p_{3}\right)$ & $\left(p_{1} p_{2} + p_{1} p_{3} - p_{2} p_{3},\,p_{1} p_{2} + p_{1} p_{3} - p_{2} p_{3}\right)$ \\
\end{tabular}
\end{subtable}
\caption{The set of $G=SU(3)$ $\vec s$ values for Brieskorn spheres $\Sigma(p_1,p_2,p_3)$, satisfying \eqref{eqn:Brieskornpq} and $q_i>0$ for all $i=1,2,3$.
}\label{tab:SymbA2Briesk}
\end{table}
Integrands of the $\widehat{Z}$-invariant for Brieskorn spheres enjoy a symmetry property (invariance up to a sign, see \eqref{eq:diagonal_action}) under 
the diagonal action of the Weyl group $\hat w \mapsto w'\hat w$. 
Under diagonal action of the Weyl group, equation \eqref{eqn:s} shows that $\vec s$ is mapped to $w(\vec s)$. 
Hence one can always choose a representative $\vec s$ that lies in the fundamental Weyl chamber. 
Equivalence under diagonal action allows for the reduction of the number of inequivalent $\vec{s}$ from $\lvert W^{\otimes N}\lvert$ to $\lvert W^{\otimes N-1}\lvert$. 

In the $SU(2)$ case we thus only have $|W|^{N-1}=2^2$ inequivalent $\vec s= s\,\vec \omega$, which  can be obtained by fixing one of the components of $\hat{w}$ to any $w\in W$.
One possible choice is to fix $w_1$ to the identity $\id$.
The $s$ values we obtain with this choice are
\begin{align}
    s &\in\left\{-p_2p_3-p_1p_3+p_1p_2,-p_2p_3+p_1p_3-p_1p_2,\right.\nonumber\\
    &\quad\ \left.-p_2p_3-p_1p_3-p_1p_2,-p_2p_3+p_1p_3+p_1p_2\right\},
    \label{eq:SymbolicA1Briesk}
\end{align}
assuming ${\rm sign}(q_i)=1$ for $i=1,2,3$.
In the set \eqref{eq:SymbolicA1Briesk}, the first two $s$ values are computed from $\hat{w}=\left(\id,\id,-\id\right)$ and $\hat{w}=\left(\id,-\id,\id\right)$, which have $(-1)^{\ell(\hat{w})}=-1$, whereas the last two come from $\hat{w}=\left(\id,\id,\id\right)$ and $\hat{w}=\left(\id,-\id,-\id\right)$, which have $(-1)^{\ell(\hat{w})}=1$.

For the $SU(3)$ case, we have $|W|^{N-1}=36$ inequivalent pairs of $\vec s = (s_1,s_2)$. 
Using the diagonal Weyl group action to fix $\hat{w}_{1}$ to be the identity element, the pairs $\vec s = (s_1,s_2)$  corresponding to even resp. odd Weyl length can be found in Table \ref{tab:SymbA2Briesk}.
Alternatively, symmetry under the diagonal action of the Weyl group \eqref{eq:diagonal_action} can be used in most cases to fix $\vec{s}$ so that the its weight components fall into the range $s_i\in\left\{1,2,\dots,m\right\}$ corresponding to the expected range for characters of log VOA representations (as opposed to generalized characters).

This restriction on the magnitude of $s_i$ depends on the relative magnitude of the $p_i$ coefficients.
For $G=SU(2)$ the only manifold not to have this simplification is the Poincaré sphere, $(p_1,p_2,p_3)=(2,3,5)$; for $G=SU(3)$ the property fails on all spheres with at least one $p_i<4$\footnote{The choices given in table \ref{tab:SymbA2Briesk} do not always fall in the range $s_i\in\left\{1,2,\dots,m\right\}$ for all Brieskorn spheres $\Sigma(p_1,p_2,p_3)$.}.

\vspace{15pt}
\noindent
{\bf One spherical example}

We explicitly work out the Brieskorn sphere example $X_\Gamma=M\left( -1; \frac{3}{5}, \frac{2}{7}, \frac{1}{9} \right)=\Sigma(5,7,9)$.
By expanding the continued fractions we can compute the plumbing matrix for this Seifert manifold
\begin{equation}
\frac{5}{3}=2-\frac{1}{3},\quad \frac{7}{2}=4 - \frac{1}{2},\quad 9=9
\end{equation}
so the plumbing matrix is
\begin{equation}
M=\left(\begin{array}{rrrrrr}
-1 & 1 & 0 & 1 & 0 & 1 \\
1 & -2 & 1 & 0 & 0 & 0 \\
0 & 1 & -3 & 0 & 0 & 0 \\
1 & 0 & 0 & -4 & 1 & 0 \\
0 & 0 & 0 & 1 & -2 & 0 \\
1 & 0 & 0 & 0 & 0 & -9
\end{array}\right).
\end{equation}
The plumbing graph of the Brieskorn sphere can also be found in Figure \ref{fig:BriesGraph}.
\begin{figure}
	\centering
	\includegraphics[width=0.5\textwidth]{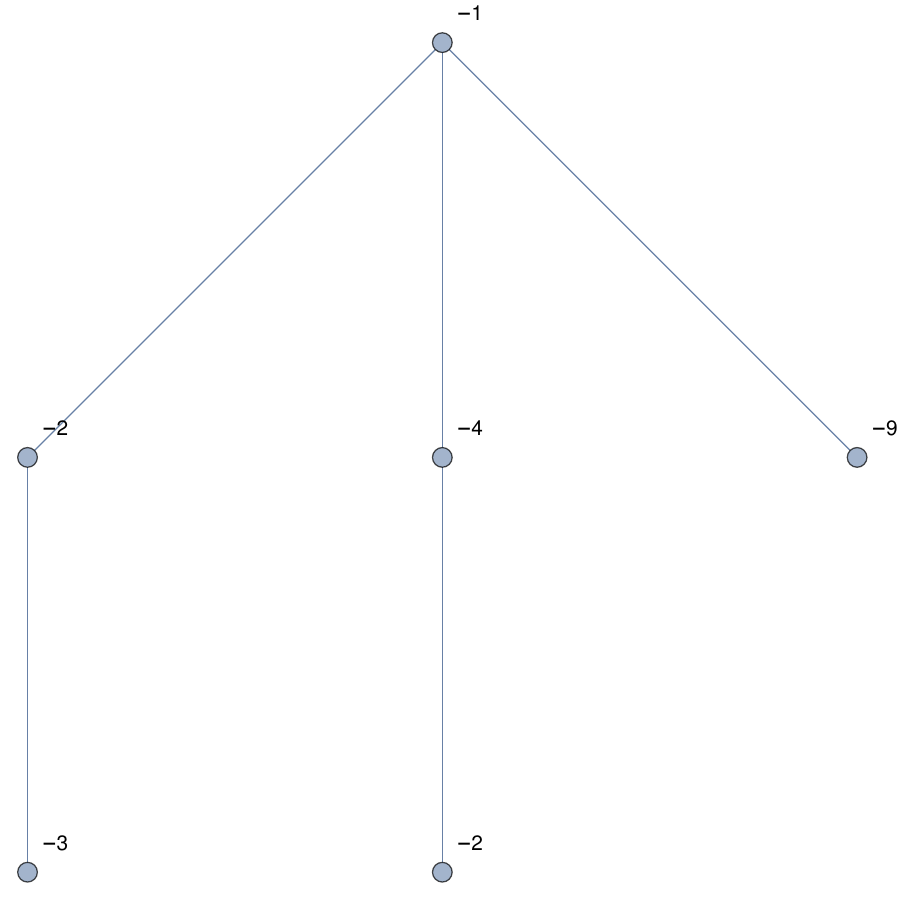}
	\caption{Graph of $M\left( -1,\frac{3}{5},\frac{2}{7},\frac{1}{9} \right)$}
	\label{fig:BriesGraph}
\end{figure}
Of the three legs two are of length two (corresponding to $\left|q_i\right| > 1$) and one  is of length one.
In total this graph has $\left|V\right|=6$ nodes which coincide with the dimensions of the plumbing matrix.

Other important topological data needed to compute the integrand of $\widehat{Z}^G_{\vec{\underline{b}}}$ are $m=-M^{-1}_{v_{0},v_{0}}=\prod_{i=1}^{3}p_i=315$.
\begin{table}
	\centering
	\begin{subtable}{\textwidth}
	\small
	\centering
		\begin{tabular}{lll|lll}
$\hat{w}$ & $(-1)^{\ell(\hat{w})}$ & $\left(s_{1}\right)$ & 
$\hat{w}$ & $(-1)^{\ell(\hat{w})}$ & $\left(s_{1}\right)$ \\\toprule
$\left({\id}, {\id}, {\id}\right)$ & $1$ & $\left(143\right)$ & $\left({\id}, {\id}, -{\id}\right)$ & $-1$ & $\left(73\right)$ \\
$\left({\id}, -{\id}, {\id}\right)$ & $-1$ & $\left(53\right)$ & $\left(-{\id}, {\id}, {\id}\right)$ & $-1$ & $\left(17\right)$ \\
\end{tabular}
	\vspace{0.2cm}
	\end{subtable}
\hfill

\begin{subtable}{\textwidth}
\small
	\centering
\begin{tabular}{lll|lll}
$\hat{w}$ & $(-1)^{\ell(\hat{w})}$ & $\left(s_{1},s_{2}\right)$ & 
$\hat{w}$ & $(-1)^{\ell(\hat{w})}$ & $\left(s_{1},s_{2}\right)$ \\ \toprule
$\left({\id}, {\id}, {\id}\right)$ & $1$ & $\left(143,\,143\right)$ & $\left({\id}, {\id}, a\right)$ & $-1$ & $\left(73,\,178\right)$ \\
$\left({\id}, {\id}, b\right)$ & $-1$ & $\left(178,\,73\right)$ & $\left({\id}, {\id}, \mathit{ab}\right)$ & $1$ & $\left(38,\,143\right)$ \\
$\left({\id}, {\id}, \mathit{ba}\right)$ & $1$ & $\left(143,\,38\right)$ & $\left({\id}, {\id}, \mathit{aba}\right)$ & $-1$ & $\left(73,\,73\right)$ \\
$\left({\id}, a, {\id}\right)$ & $-1$ & $\left(53,\,188\right)$ & $\left({\id}, a, b\right)$ & $1$ & $\left(88,\,118\right)$ \\
$\left({\id}, a, \mathit{ba}\right)$ & $-1$ & $\left(53,\,83\right)$ & $\left({\id}, b, {\id}\right)$ & $-1$ & $\left(188,\,53\right)$ \\
$\left({\id}, b, a\right)$ & $1$ & $\left(118,\,88\right)$ & $\left({\id}, b, \mathit{ab}\right)$ & $-1$ & $\left(83,\,53\right)$ \\
$\left({\id}, \mathit{ab}, {\id}\right)$ & $1$ & $\left(8,\,143\right)$ & $\left({\id}, \mathit{ab}, b\right)$ & $-1$ & $\left(43,\,73\right)$ \\
$\left({\id}, \mathit{ab}, \mathit{ba}\right)$ & $1$ & $\left(8,\,38\right)$ & $\left({\id}, \mathit{ba}, {\id}\right)$ & $1$ & $\left(143,\,8\right)$ \\
$\left({\id}, \mathit{ba}, a\right)$ & $-1$ & $\left(73,\,43\right)$ & $\left({\id}, \mathit{ba}, \mathit{ab}\right)$ & $1$ & $\left(38,\,8\right)$ \\
$\left({\id}, \mathit{aba}, {\id}\right)$ & $-1$ & $\left(53,\,53\right)$ & $\left(a, {\id}, {\id}\right)$ & $-1$ & $\left(17,\,206\right)$ \\
$\left(a, {\id}, b\right)$ & $1$ & $\left(52,\,136\right)$ & $\left(a, {\id}, \mathit{ba}\right)$ & $-1$ & $\left(17,\,101\right)$ \\
$\left(a, b, {\id}\right)$ & $1$ & $\left(62,\,116\right)$ & $\left(a, b, b\right)$ & $-1$ & $\left(97,\,46\right)$ \\
$\left(a, b, \mathit{ba}\right)$ & $1$ & $\left(62,\,11\right)$ & $\left(a, \mathit{ba}, {\id}\right)$ & $-1$ & $\left(17,\,71\right)$ \\
$\left(a, \mathit{ba}, b\right)$ & $1$ & $\left(52,\,1\right)$ & $\left(b, {\id}, {\id}\right)$ & $-1$ & $\left(206,\,17\right)$ \\
$\left(b, {\id}, a\right)$ & $1$ & $\left(136,\,52\right)$ & $\left(b, {\id}, \mathit{ab}\right)$ & $-1$ & $\left(101,\,17\right)$ \\
$\left(b, a, {\id}\right)$ & $1$ & $\left(116,\,62\right)$ & $\left(b, a, a\right)$ & $-1$ & $\left(46,\,97\right)$ \\
$\left(b, a, \mathit{ab}\right)$ & $1$ & $\left(11,\,62\right)$ & $\left(b, \mathit{ab}, {\id}\right)$ & $-1$ & $\left(71,\,17\right)$ \\
$\left(b, \mathit{ab}, a\right)$ & $1$ & $\left(1,\,52\right)$ & $\left(\mathit{aba}, {\id}, {\id}\right)$ & $-1$ & $\left(17,\,17\right)$ \\
\end{tabular}
\end{subtable}
	\caption{$SU(2)$ and $SU(3)$ $\vec s$ values for $M\left( -1;\frac{3}{5},\frac{2}{7},\frac{1}{9} \right)$, where $\mathit{a},\, \mathit{b}$ are reflections to the planes orthogonal to $\vec\alpha_1$ and $\vec\alpha_2$ respectively.}
	\label{tab:Briesk}
\end{table}
Symmetry with respect to the components of $\vec{s}$ is a general property of Brieskorn spheres and a direct consequence of equations  (\ref{dfn:A}) and (\ref{eqn:delta_A_main}).

For each $\hat{w}$, and thus for each pair $(s_{1},\ s_{2})$,  we can compute $\tilde{\chi}_{\hat{w};\vec{\underline{b}}}$ explicitly.
If we set $\hat{w}=\left(\mathds{1},\mathds{1},\mathds{1} \right)$,  using equation \eqref{integrand_3star2}:
\begin{align}
\tilde{\chi}_{\hat{w};\vec{\underline{b}}}(\tau,\vec{\xi})&=q^{95} - 2q^{238} + 2q^{524} - q^{667} + \frac{ \Delta\left(2 \xi_{1}, 2 \xi_{2}\right)}{\Delta\left(\xi_{1}, \xi_{2}\right)}q^{754} + O(q^{1000})
\nonumber\\&=q^{\delta}\eta^2( \tau )\,\chi_{\vec{\lambda}'_{0,0,s_1,s_2}}(\tau,\vec{\xi})
\end{align}
where:
\begin{equation}
\vec{s}= \sum_{i=1}^2 s_i\vec\omega_i= -m\vec{A}_{\hat{w}}=143\vec\rho \, ,~~ \delta = \frac{341}{315}.
\end{equation}
In Table \ref{tab:briespi} the reader can find a representative of the set of equivalent $\vec s$'s for the case $G=SU(3)$ in terms of the $p_i$'s, such that all $s_1, s_2$ are within the range between $1$ and $m$ for $X_\Gamma=\Sigma(5,7,9)$. 
The explicit values of 
 $\vec{s}$'s for $G=SU(2)$ and $G=SU(3)$ can be found in Table \ref{tab:Briesk}.
\begin{table}
\small
	\centering
	\begin{tabular}{ll}
		$\vec s = \left( s_{1},s_{2} \right)$&\\\toprule
		$\left(p_1 p_2+p_3 p_2+p_1 p_3,  p_1 p_2+p_3 p_2+p_1 p_3 \right)$&
		$\left(p_1 p_3+p_2 p_3 , p_1 p_2+p_3 p_2+p_1 p_3\right)$\\ 
		$\left(p_1 p_2+p_3 p_2+p_1 p_3 , p_1 p_3+p_2 p_3 \right)$&
		$\left(-p_1 p_2+p_3 p_2+p_1 p_3 , p_1 p_3+p_2 p_3 \right)$\\
		$\left(p_1 p_3+p_2 p_3 , -p_1 p_2+p_3 p_2+p_1 p_3 \right)$&
		$\left(-p_1 p_2+p_3 p_2+p_1 p_3 , -p_1 p_2+p_3 p_2+p_1 p_3 \right)$\\
		$\left(p_1 p_2+p_3 p_2 , p_1 p_2+p_3 p_2+p_1 p_3 \right)$&
		$\left(p_2 p_3 , p_1 p_2+p_3 p_2+p_1 p_3 \right)$\\
		$\left(p_1 p_2+p_3 p_2 , p_1 p_3+p_2 p_3 \right)$&
		$\left(p_2 p_3-p_1 p_2 , p_1 p_3+p_2 p_3 \right)$\\
		$\left(p_2 p_3 , -p_1 p_2+p_3 p_2+p_1 p_3 \right)$&
		$\left(p_2 p_3-p_1 p_2 , -p_1 p_2+p_3 p_2+p_1 p_3 \right)$\\
		$\left(p_1 p_2+p_3 p_2+p_1 p_3 , p_1 p_2+p_3 p_2 \right)$&
		$\left(p_1 p_3+p_2 p_3 , p_1 p_2+p_3 p_2 \right)$\\
		$\left(p_1 p_2+p_3 p_2+p_1 p_3 , p_2 p_3 \right)$&
		$\left(-p_1 p_2+p_3 p_2+p_1 p_3 , p_2 p_3 \right)$\\
		$\left(p_1 p_3+p_2 p_3 , p_2 p_3-p_1 p_2 \right)$&
		$\left(-p_1 p_2+p_3 p_2+p_1 p_3 , p_2 p_3-p_1 p_2 \right)$\\
		$\left(p_1 p_2+p_3 p_2-p_1 p_3 , p_1 p_2+p_3 p_2 \right)$&
		$\left(p_2 p_3-p_1 p_3 , p_1 p_2+p_3 p_2 \right)$\\
		$\left(p_1 p_2+p_3 p_2-p_1 p_3 , p_2 p_3 \right)$&
		$\left(-p_1 p_2+p_3 p_2-p_1 p_3 , p_2 p_3 \right)$\\
		$\left(p_2 p_3-p_1 p_3 , p_2 p_3-p_1 p_2 \right)$&
		$\left(-p_1 p_2+p_3 p_2-p_1 p_3 , p_2 p_3-p_1 p_2 \right)$\\
		$\left(p_1 p_2+p_3 p_2 , p_1 p_2+p_3 p_2-p_1 p_3 \right)$&
		$\left(p_2 p_3 , p_1 p_2+p_3 p_2-p_1 p_3 \right)$\\
		$\left(p_1 p_2+p_3 p_2 , p_2 p_3-p_1 p_3 \right)$&
		$\left(p_2 p_3-p_1 p_2 , p_2 p_3-p_1 p_3 \right)$\\
		$\left(p_2 p_3 , -p_1 p_2+p_3 p_2-p_1 p_3 \right)$&
		$\left(p_2 p_3-p_1 p_2 , -p_1 p_2+p_3 p_2-p_1 p_3 \right)$\\
		$\left(p_1 p_2+p_3 p_2-p_1 p_3 , p_1 p_2+p_3 p_2-p_1 p_3 \right)$&
		$\left(p_2 p_3-p_1 p_3 , p_1 p_2+p_3 p_2-p_1 p_3 \right)$\\
		$\left(p_1 p_2+p_3 p_2-p_1 p_3 , p_2 p_3-p_1 p_3 \right)$&
		$\left(-p_1 p_2+p_3 p_2-p_1 p_3 , p_2 p_3-p_1 p_3 \right)$\\
		$\left(p_2 p_3-p_1 p_3 , -p_1 p_2+p_3 p_2-p_1 p_3 \right)$&
		$\left(-p_1 p_2+p_3 p_2-p_1 p_3 , -p_1 p_2+p_3 p_2-p_1 p_3 \right)$
	\end{tabular}
	\caption{$\vec s$ values for $G=SU(3)$ and $X_\Gamma=\Sigma(p_1,p_2,p_3)$.}\label{tab:briespi}
\end{table}
Finally, we can demonstrate Corollary \ref{cor:psudo} by summing over all $\tilde{\chi}_{\hat{w};\vec{\underline{b}}}(\tau,\vec \xi)$ and integrating:
\begin{align}
C_\Gamma^G(q)^{-1}\widehat{Z}^{A_2}_{\vec{\underline{b}}_0}(\tau)&=6 q^{95} - {12} q^{118} - {12} q^{126} - {12} q^{142} + 12 q^{144} + 12 q^{158} + {O}\left(q^{160}\right)\nonumber\\
&=\sum_{(s_1,s_2)}(-1)^{\ell_{s_1,s_2}}q^\delta\eta^2(\tau)\chi_{\vec{\mu}_{(s_1,s_2)}}^0(\tau).
\end{align}

\vspace{15pt}
\noindent
{\bf A pseudo-spherical example}

Pseudo-spherical Seifert manifolds share similar features to spherical Seifert manifolds, with the crucial difference that not all $\hat{w}$ contribute to the $\widehat{Z}-$invariant.
Such manifolds have non-unimodular plumbing matrix, but because $M^{-1}_{v_0,v}\in \mathbb{Z}\, ~\forall v\in V$ the lattice dilation factor $D$ is one.
One such example is given by $X_\Gamma=M\left( -1;\frac{1}{2},\frac{1}{3},\frac{1}{9} \right)$.
Its plumbing matrix is:
\begin{equation}
M=\left(
\begin{array}{cccc}
-1 & 1 & 1 & 1 \\
1 & -2 & 0 & 0 \\
1 & 0 & -3 & 0 \\
1 & 0 & 0 & -9 \\
\end{array}
\right), ~~ M^{-1}=\left(
\begin{array}{cccc}
-18 & -9 & -6 & -2 \\
-9 & -5 & -3 & -1 \\
-6 & -3 & -\frac{7}{3} & -\frac{2}{3} \\
-2 & -1 & -\frac{2}{3} & -\frac{1}{3} \\
\end{array}
\right).
\end{equation}
Because $D=1$, this manifold serves as an example of Theorem \ref{3fiber_sphere}.
We can also compute:
$$\det\left( M \right)=3, ~ m = -DM_{v_0,v_0}^{-1}=18,$$
 and $$ \Ck\left( M \right)\cong \left\{\left(0, 0, 0, 0\right), \left(1, 0, -1, -6\right), \left(1, 0, -2, -3\right)\right\},$$
leading to two inequivalent choices of $\vec{\underline{b}}$ for $G=SU(3)$:
\begin{gather}
\begin{split}
\vec{\underline{b}}_0&=\left(\left(-1,\,-1\right), \left(1,\,1\right), \left(1,\,1\right), \left(1,\,1\right)\right)\\
\vec{\underline{b}}_1&=\left(\left(0,\,-1\right), \left(1,\,1\right), \left(-1,\,1\right), \left(-2,\,1\right)\right).
\end{split}
\end{gather}
Because the inverse plumbing matrix contains non integer entries, not all choices of $\hat{w}$ give nonempty $S_{\hat{w};\vec{\underline{b}}}$, so not all $\hat{w}$ contribute to the sum in \eqref{hatZ1}.
For instance, inspecting the explicit form of the inverse plumbing matrix, one sees that the set \eqref{dfn:theset} is never empty when $w_{2}=w_{3}$ for $\vec{\underline{b}}=\vec{\underline{b}}_0$, leading to 36 possible choices of $\hat{w}$. 
We collect the set of corresponding $\vec s$, for $\vec{\underline{b}}=\vec{\underline{b}}_0$ and $\vec{\underline{b}}_1$, in Table \ref{tab:pseudoA2}.

As a very explicit example, we look at $\vec{\underline{b}}=\vec{\underline{b}}_0$ and $\hat{w}=\left( \mathds{1},\mathds{1},\mathds{1} \right)$, for which we have
\begin{align}
 \tilde{\chi}_{\hat{w};\vec{\underline{b}}}(\tau,\vec{\xi})&=q - 2 \, q^{18} + \frac{q^{21} \Delta\left(2\xi_{1}, 2  \xi_{2}\right)}{\Delta\left(\xi_{1}, \xi_{2}\right)} + 2 \, q^{52} - \frac{2 \, q^{55} \Delta\left(2 \xi_{1}, 2  \xi_{2}\right)}{\Delta\left(\xi_{1}, \xi_{2}\right)}
+O\left(q^{58}\right)\nonumber\\&=q^{\delta}(\eta\left( \tau \right))^2\chi_{\vec \lambda'_{0,0,s_1,s_2}}(\tau,\vec{\xi})
\end{align}
where $(s_1,s_2)=\left( 17, 17 \right)$ as in (\ref{eqn:A2triplet}),  and $\delta=\frac{17}{18}$, as in \eqref{eqn:delta_A_main}.

As for the spherical case, we can recover the $\widehat Z$ invariant by summing over all $\hat{w}$ that contribute with a non-empty $\mathcal S_{\hat w; \vec{\underline{b}}}$.
For $\vec{\underline{b}}=\vec{\underline{b}}_0$ we get
\begin{align}
C_\Gamma^G(q)^{-1}\widehat{Z}^{A_2}_{\vec{\underline{b}}_0}(\tau)&=6 q + 12 q^{5}  {-24} q^{6} + 12 q^{13} + 12 q^{16} + 6 q^{17} {-24} q^{18} + \mathcal{O}\left(q^{20}\right)\nonumber\\
&=\sum_{(s_1,s_2)}(-1)^{\ell_{s_1,s_2}}q^\delta\eta^2(\tau)\chi_{\vec{\mu}_{(s_1,s_2)}}^0(\tau)
\end{align}
and for $\vec{\underline{b}}=\vec{\underline{b}}_1$
\begin{align}
C_\Gamma^G(q)^{-1}\widehat{Z}^{A_2}_{\vec{\underline{b}}_1}(\tau)&=q^{1/3}\left({-2} q  {-1} q^{2}  {-2} q^{3}  {-2} q^{4} + 4 q^{5}  {-2} q^{6} + 2 q^{8} + 4 q^{9}\right) + \mathcal{O}\left(q^{31/3}\right) \nonumber\\
&=\sum_{(s_1,s_2)}(-1)^{\ell_{s_1,s_2}}q^\delta\eta^2(\tau)\chi_{\vec{\mu}_{(s_1,s_2)}}^0(\tau).
\end{align}

\begin{table}
\begin{subtable}{\textwidth}
\small
\begin{tabular}{lll|lll}
$\hat{w}$ & $(-1)^{\ell(\hat{w})}$ & $\left(s_{1},s_{2}\right)$ & 
$\hat{w}$ & $(-1)^{\ell(\hat{w})}$ & $\left(s_{1},s_{2}\right)$ \\\toprule

$\left({\id}, {\id}, {\id}\right)$ & $1$ & $\left(17,\,17\right)$ & $\left({\id}, a, a\right)$ & $1$ & $\left(1,\,25\right)$ \\
$\left({\id}, b, b\right)$ & $1$ & $\left(25,\,1\right)$ & $\left({\id}, \mathit{ab}, \mathit{ab}\right)$ & $1$ & $\left(-7,\,17\right)$ \\
$\left({\id}, \mathit{ba}, \mathit{ba}\right)$ & $1$ & $\left(17,\,-7\right)$ & $\left({\id}, \mathit{aba}, \mathit{aba}\right)$ & $1$ & $\left(1,\,1\right)$ \\
$\left(a, {\id}, {\id}\right)$ & $-1$ & $\left(-1,\,26\right)$ & $\left(a, a, a\right)$ & $-1$ & $\left(-17,\,34\right)$ \\
$\left(a, b, b\right)$ & $-1$ & $\left(7,\,10\right)$ & $\left(a, \mathit{ab}, \mathit{ab}\right)$ & $-1$ & $\left(-25,\,26\right)$ \\
$\left(a, \mathit{ba}, \mathit{ba}\right)$ & $-1$ & $\left(-1,\,2\right)$ & $\left(a, \mathit{aba}, \mathit{aba}\right)$ & $-1$ & $\left(-17,\,10\right)$ \\
$\left(b, {\id}, {\id}\right)$ & $-1$ & $\left(26,\,-1\right)$ & $\left(b, a, a\right)$ & $-1$ & $\left(10,\,7\right)$ \\
$\left(b, b, b\right)$ & $-1$ & $\left(34,\,-17\right)$ & $\left(b, \mathit{ab}, \mathit{ab}\right)$ & $-1$ & $\left(2,\,-1\right)$ \\
$\left(b, \mathit{ba}, \mathit{ba}\right)$ & $-1$ & $\left(26,\,-25\right)$ & $\left(b, \mathit{aba}, \mathit{aba}\right)$ & $-1$ & $\left(10,\,-17\right)$ \\
$\left(\mathit{ab}, {\id}, {\id}\right)$ & $1$ & $\left(-10,\,17\right)$ & $\left(\mathit{ab}, a, a\right)$ & $1$ & $\left(-26,\,25\right)$ \\
$\left(\mathit{ab}, b, b\right)$ & $1$ & $\left(-2,\,1\right)$ & $\left(\mathit{ab}, \mathit{ab}, \mathit{ab}\right)$ & $1$ & $\left(-34,\,17\right)$ \\
$\left(\mathit{ab}, \mathit{ba}, \mathit{ba}\right)$ & $1$ & $\left(-10,\,-7\right)$ & $\left(\mathit{ab}, \mathit{aba}, \mathit{aba}\right)$ & $1$ & $\left(-26,\,1\right)$ \\
$\left(\mathit{ba}, {\id}, {\id}\right)$ & $1$ & $\left(17,\,-10\right)$ & $\left(\mathit{ba}, a, a\right)$ & $1$ & $\left(1,\,-2\right)$ \\
$\left(\mathit{ba}, b, b\right)$ & $1$ & $\left(25,\,-26\right)$ & $\left(\mathit{ba}, \mathit{ab}, \mathit{ab}\right)$ & $1$ & $\left(-7,\,-10\right)$ \\
$\left(\mathit{ba}, \mathit{ba}, \mathit{ba}\right)$ & $1$ & $\left(17,\,-34\right)$ & $\left(\mathit{ba}, \mathit{aba}, \mathit{aba}\right)$ & $1$ & $\left(1,\,-26\right)$ \\
$\left(\mathit{aba}, {\id}, {\id}\right)$ & $-1$ & $\left(-1,\,-1\right)$ & $\left(\mathit{aba}, a, a\right)$ & $-1$ & $\left(-17,\,7\right)$ \\
$\left(\mathit{aba}, b, b\right)$ & $-1$ & $\left(7,\,-17\right)$ & $\left(\mathit{aba}, \mathit{ab}, \mathit{ab}\right)$ & $-1$ & $\left(-25,\,-1\right)$ \\
$\left(\mathit{aba}, \mathit{ba}, \mathit{ba}\right)$ & $-1$ & $\left(-1,\,-25\right)$ & $\left(\mathit{aba}, \mathit{aba}, \mathit{aba}\right)$ & $-1$ & $\left(-17,\,-17\right)$ \\
\end{tabular}  
\vspace{0.2cm}
\end{subtable}
\hfill

\begin{subtable}{\textwidth}
\small
\begin{tabular}{lll|lll}
$\hat{w}$ & $(-1)^{\ell(\hat{w})}$ & $\left(s_{1},s_{2}\right)$ & 
$\hat{w}$ & $(-1)^{\ell(\hat{w})}$ & $\left(s_{1},s_{2}\right)$ \\\toprule
$\left({\id}, a, {\id}\right)$ & $-1$ & $\left(5,\,23\right)$ & $\left({\id}, b, \mathit{ab}\right)$ & $-1$ & $\left(17,\,5\right)$ \\
$\left({\id}, \mathit{aba}, \mathit{ba}\right)$ & $-1$ & $\left(5,\,-1\right)$ & $\left(a, a, {\id}\right)$ & $1$ & $\left(-13,\,32\right)$ \\
$\left(a, b, \mathit{ab}\right)$ & $1$ & $\left(-1,\,14\right)$ & $\left(a, \mathit{aba}, \mathit{ba}\right)$ & $1$ & $\left(-13,\,8\right)$ \\
$\left(b, a, {\id}\right)$ & $1$ & $\left(14,\,5\right)$ & $\left(b, b, \mathit{ab}\right)$ & $1$ & $\left(26,\,-13\right)$ \\
$\left(b, \mathit{aba}, \mathit{ba}\right)$ & $1$ & $\left(14,\,-19\right)$ & $\left(\mathit{ab}, a, {\id}\right)$ & $-1$ & $\left(-22,\,23\right)$ \\
$\left(\mathit{ab}, b, \mathit{ab}\right)$ & $-1$ & $\left(-10,\,5\right)$ & $\left(\mathit{ab}, \mathit{aba}, \mathit{ba}\right)$ & $-1$ & $\left(-22,\,-1\right)$ \\
$\left(\mathit{ba}, a, {\id}\right)$ & $-1$ & $\left(5,\,-4\right)$ & $\left(\mathit{ba}, b, \mathit{ab}\right)$ & $-1$ & $\left(17,\,-22\right)$ \\
$\left(\mathit{ba}, \mathit{aba}, \mathit{ba}\right)$ & $-1$ & $\left(5,\,-28\right)$ & $\left(\mathit{aba}, a, {\id}\right)$ & $1$ & $\left(-13,\,5\right)$ \\
$\left(\mathit{aba}, b, \mathit{ab}\right)$ & $1$ & $\left(-1,\,-13\right)$ & $\left(\mathit{aba}, \mathit{aba}, \mathit{ba}\right)$ & $1$ & $\left(-13,\,-19\right)$ \\
\end{tabular}
\end{subtable}

\caption{The set of $\vec s$ for $X_\Gamma= M\left(-1,\frac{1}{2},\frac{1}{3},\frac{1}{9}\right)$, with $\vec{\underline{b}}_0,\
\vec{\underline{b}}=\vec{\underline{b}}_0$, $\vec{\underline{b}}_1$ for the upper and lower sub-table respectively.  
}
\label{tab:pseudoA2}
\end{table}

\vspace{15pt}
\noindent
{\bf A non-spherical example with $1<D<\det{M}$ }

The manifold $X_\Gamma=M\left( -1; \frac{1}{3},-\frac{1}{2},-\frac{1}{2} \right)$ 
has plumbing matrix
\begin{equation}
M=\left(
\begin{array}{cccc}
-1 & 1 & 1 & 1 \\
1 & -3 & 0 & 0 \\
1 & 0 & 2 & 0 \\
1 & 0 & 0 & 2 \\
\end{array}
\right),
\end{equation}
with $\det(M)=20$ and  $D=10<\det(M)$.  
Theorem \ref{thm:combining_into_characters} applies to this example as well, as we will explicitly demonstrate here. 
The Cokernel of the plumbing matrix is 
\begin{align}
&\Ck\left( M \right)=\nonumber\\
&\left\{  \left(1,\,-2,\,1,\,1\right), \left(1,\,-1,\,2,\,2\right), \left(2,\,-1,\,2,\,2\right), \left(0,\,0,\,1,\,1\right), \left(1,\,0,\,1,\,1\right),\right.\nonumber\\
& \left(2,\,-2,\,2,\,2\right), \left(0,\,-1,\,1,\,1\right), \left(1,\,-1,\,1,\,1\right), \left(1,\,0,\,2,\,2\right), \left(0,\,0,\,0,\,0\right),\nonumber\\
&\left(1,\,0,\,1,\,2\right), \left(0,\,0,\,1,\,2\right), \left(1,\,-1,\,2,\,1\right), \left(2,\,-2,\,1,\,2\right), \left(1,\,-2,\,1,\,2\right),\nonumber\\
&\left. \left(1,\,0,\,2,\,1\right), \left(0,\,0,\,2,\,1\right), \left(1,\,-1,\,1,\,2\right), \left(2,\,-2,\,2,\,1\right), \left(1,\,-2,\,2,\,1\right)
\right\},
\end{align}
leading to
\begin{align}
\mathcal{B} =& \{\vec{\underline{b}}_i\}_{i=0,\dots,10}\\
=&\left\{\left(\left(-1,\,-1\right), \left(1,\,1\right), \left(1,\,1\right), \left(1,\,1\right)\right), \left(\left(-1,\,-1\right), \left(1,\,1\right), \left(3,\,0\right), \left(3,\,0\right)\right),\right.\nonumber\\
&\left(\left(-1,\,-1\right), \left(1,\,1\right), \left(3,\,0\right), \left(5,\,-1\right)\right), \left(\left(1,\,-2\right), \left(-3,\,3\right), \left(3,\,0\right), \left(3,\,0\right)\right),\nonumber\\
&\left(\left(1,\,-2\right), \left(-3,\,3\right), \left(3,\,0\right), \left(5,\,-1\right)\right), \left(\left(1,\,-2\right), \left(-1,\,2\right), \left(5,\,-1\right), \left(3,\,0\right)\right),\nonumber\\
&\left(\left(1,\,-2\right), \left(-1,\,2\right), \left(5,\,-1\right), \left(5,\,-1\right)\right), \left(\left(1,\,-2\right), \left(1,\,1\right), \left(3,\,0\right), \left(3,\,0\right)\right),\nonumber\\
&\left(\left(1,\,-2\right), \left(1,\,1\right), \left(3,\,0\right), \left(5,\,-1\right)\right), \left(\left(3,\,-3\right), \left(-3,\,3\right), \left(3,\,0\right), \left(5,\,-1\right)\right),\nonumber\\
&\left.\left(\left(3,\,-3\right), \left(-1,\,2\right), \left(5,\,-1\right), \left(5,\,-1\right)\right)\right\}
\end{align}
As for the pseudo-spherical example, the inverse plumbing matrix contains non integer entries, leading to a smaller set of admissible $\hat{w}$ with non-vanishing contribution.
One such admissible $\hat{w}$ is $\hat{w}=\left( \mathds{1},\mathds{1},\mathds{1} \right)$, which has $\vec{\kappa}_{\hat{w};\vec{\underline{b}}_{0}}=\left( 9,9 \right)$.
We can therefore compute its $\widehat{Z}^G_{\vec{\underline{b}}}$ integrand contribution
\begin{equation}
\tilde{\chi}_{\hat{w};\vec{\underline{b}}_{0}}\left( D\tau,\vec{\xi} \right)=q^{10}+q^{410}\frac{\Delta\left( 9\xi_{1},9\xi_{2} \right)}{\Delta\left( \xi_{1},\xi_{2} \right)}+O\left( q^{510} \right)
\end{equation}
and
\begin{equation}
  q^{D\delta}\eta^2(\tau)\chi_{\vec{\lambda}'_{0,0,s_1,s_2}}(\tau,\xi)=- {q^{-6}}+ {2}{q^{-2}} - \frac{q^{4} \Delta\left(2 \, \xi_{1}, 2 \, \xi_{2}\right)}{\Delta\left(\xi_{1}, \xi_{2}\right)} -2q^{6}+q^{10}+O\left( q^{12} \right)
\end{equation}
where $\delta=-\frac{2}{3}$ and $\vec{s}=4\vec{\rho}$.
As proven in Theorem \ref{thm:combining_into_characters}, we explicitly check that the full character is recovered if we sum over all possible $\vec{\lambda}\in \Lambda/D\Lambda$, which is achieved by summing over $\Delta \vec{b}$.

\vspace{15pt}
\noindent
{\bf A non-spherical example with $1<D=\left|\det{M}\right|$}

We conclude this subsection with the Seifert manifold $X_\Gamma=M( 0; -{1\over2},{2\over7},-{1\over3})$, with  plumbing matrix:
\begin{equation}
M=\left(
\begin{array}{ccccc}
0 & 1 & 0 & 1 & 1 \\
1 & 2 & 0 & 0 & 0 \\
0 & 0 & -2 & 0 & 1 \\
1 & 0 & 0 & 3 & 0 \\
1 & 0 & 1 & 0 & -4 \\
\end{array}
\right)
\end{equation}
and plumbing graph given in Figure \ref{fig:GeneralNonUni}.
\begin{figure}
	\centering
	\includegraphics[width=0.5\textwidth]{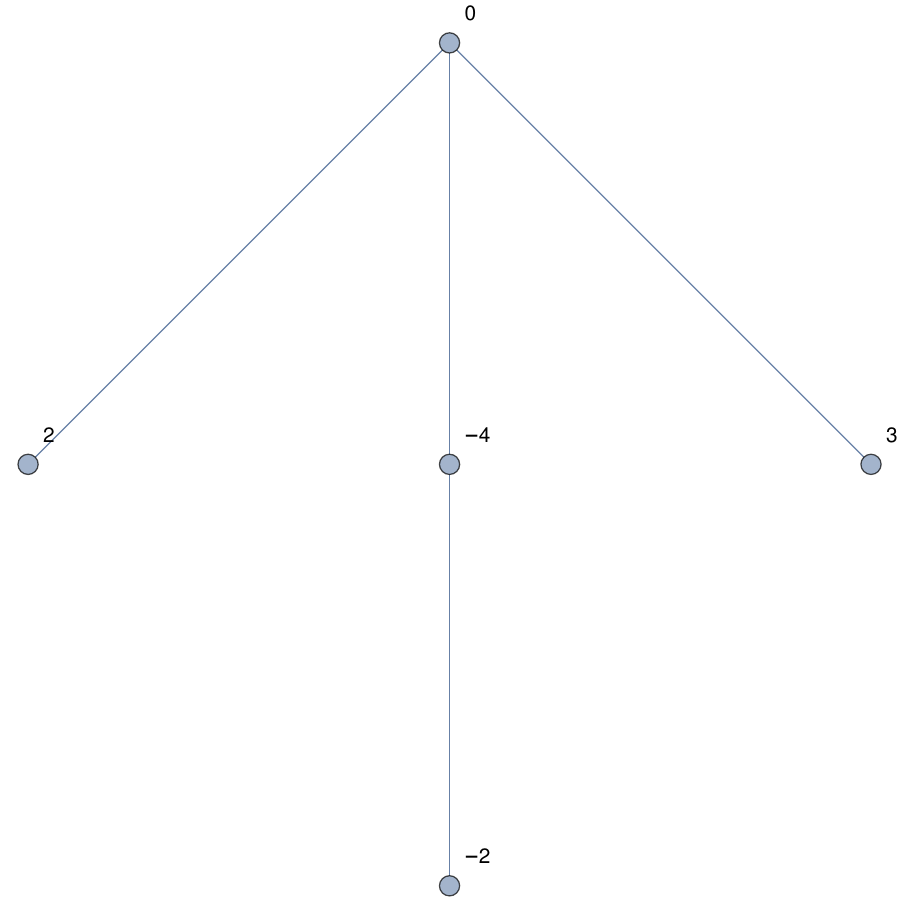}
	\caption{Graph of $M(0;-1/2,2/7,-1/3)$}
	\label{fig:GeneralNonUni}
\end{figure}
The manifold is non-unimodular, with plumbing matrix determinant of ${\rm det}(M)=-23$ so we expect to find nontrivial $\vec{\underline{b}}$ values.
The Cokernel can be computed and from that the independent Spin$^c$ structures, which are collected in Table \ref{tab:spincstruct}.
\begin{table}
  \centering
  \small
  \begin{tabular}{ll}
  $\left(\left(-1,-1\right), \left(1,1\right), \left(0,0\right), \left(1,1\right), \left(1,1\right)\right) $&$ \left(\left(1,-2\right), \left(3,0\right), \left(-4,2\right), \left(1,1\right), \left(3,0\right)\right)$\\
  $\left(\left(1,-2\right), \left(3,0\right), \left(-2,1\right), \left(1,1\right), \left(3,0\right)\right) $&$ \left(\left(1,-2\right), \left(3,0\right), \left(0,0\right), \left(-1,2\right), \left(5,-1\right)\right)$\\
  $\left(\left(1,-2\right), \left(3,0\right), \left(0,0\right), \left(1,1\right), \left(7,-2\right)\right) $&$ \left(\left(1,-2\right), \left(5,-1\right), \left(0,0\right), \left(-1,2\right), \left(3,0\right)\right)$\\
  $\left(\left(3,-3\right), \left(3,0\right), \left(-6,3\right), \left(1,1\right), \left(5,-1\right)\right) $&$ \left(\left(3,-3\right), \left(3,0\right), \left(-4,2\right), \left(-1,2\right), \left(7,-2\right)\right)$\\
  $\left(\left(3,-3\right), \left(5,-1\right), \left(-6,3\right), \left(1,1\right), \left(3,0\right)\right) $&$ \left(\left(3,-3\right), \left(5,-1\right), \left(-4,2\right), \left(-1,2\right), \left(5,-1\right)\right)$\\
  $\left(\left(3,-3\right), \left(5,-1\right), \left(-2,1\right), \left(-1,2\right), \left(5,-1\right)\right)$ & $\left(\left(3,-3\right), \left(5,-1\right), \left(-2,1\right), \left(1,1\right), \left(7,-2\right)\right)$
  \end{tabular}
  \caption{Independent Spin$^c$ structures of $M\left( 0;-\frac{1}{2},\frac{2}{7},-\frac{1}{3} \right)$}
  \label{tab:spincstruct}
\end{table}

From the inverse plumbing matrix we can compute $m=-DM^{-1}_{v_0,v_0}=42$.
As for the other non-spherical case and the pseudo-spherical case, the sets \eqref{dfn:theset} are empty for some choices of $\hat w$. 
For each $\hat{w}$ with non-empty set $S_{\hat w;\vec{\underline{b}}}$ we compute the vector  $\vec{\kappa}_{\hat{w};\vec{\underline{b}}_{0}}$ as in \eqref{dfn:kappa}.
For $\vec{\underline{b}}_0$, $\hat{w}=\left( \id,\id,\id \right)$, we compute $\vec{\kappa}_{\hat{w};\vec{\underline{b}}_{0}}=\left( 22,22 \right)$ and find
\begin{equation}
\tilde{\chi}_{\hat{w};\vec{\underline{b}}_{0}}(D\tau,\vec{\xi})= \frac{\Delta(2\xi_1,2\xi_2)}{\Delta(\xi_1,\xi_2)} q^{298}+\frac{\Delta(21\xi_1,21\xi_2)}{\Delta(\xi_1,\xi_2)}q^{17318}+O(q^{19917})
\end{equation}
and
\begin{equation}
\label{eq:ellipsis}
q^{D\delta}\eta^2\left( \tau \right)\chi_{\vec{\lambda}'_{0,0,s_1,s_2}}(\tau,\vec{\xi})=q^{-2}+\dots+\frac{\Delta(2\xi_1,2\xi_2)}{\Delta(\xi_1,\xi_2)} q^{298}+O(q^{305}) 
\end{equation}
with $\delta=-\frac{11}{42}$ and $s_1=s_2=-29$.
In \eqref{eq:ellipsis} we have put in dots the $q$-powers between $q^{-2}$ and $q^{298}$ to highlight the fact that $\tilde{\chi}_{\hat{w},\vec{\underline{b}}_0}$ is included in the log VOA character $\chi_{\vec{\lambda}'_{0,0,-29,-29}}$.
To recover the full character we sum over all possible $\vec{\lambda}\in \Lambda/D\Lambda$,
\begin{align}
\sum_{\substack{\vec{\underline{b}}=\vec{\underline{b}}_0+\left(\Delta \vec b,0,0,0 \right)\\ \Delta \vec b\in \Lambda/D\Lambda}}\tilde{\chi}_{\hat{w};\vec{\underline{b}}}(D\tau,\vec{\xi})&= -q^{-2} + 2q^{27} -\frac{\Delta(2\xi_1,2\xi_2)}{\Delta(\xi_1,\xi_2)}q^{66} -2q^{85}+q^{114} + O(q^{124})
\nonumber\\
&=q^{D\delta}\eta^2\left( \tau \right)\chi_{\vec{\lambda}'_{0,0,s_1,s_2}}(\tau,\vec{\xi}).
\end{align}

\subsection{Seifert Manifolds with Four Exceptional Fibers}\label{subsec:4fib-ex}

In this subsection we will provide two examples Seifert manifolds with four exceptional fibers to demonstrate the results of  \S\ref{subsec:4fibers}.
The first example, which will be of a spherical manifold, will give a numerical confirmation of Theorem \ref{thm:4fibchimatching} and its Corollary \ref{cor:4fib}.
The second example in this subsection will be of a pseudo-spherical Seifert manifold, 
even though this case is not covered by Theorem \ref{thm:4fibchimatching}. 
This demonstrates that the relation between {log-${\cal V}_{\bar \Lambda}^0(p,p')$} and three-manifolds holds more generally than what is proven in Theorem \ref{thm:4fibchimatching}. 

\vspace{15pt}
\noindent
{\bf A Spherical Example}

Our first example is the spherical manifold $X_\Gamma=M\left(-2, \frac{1}{2}, \frac{2}{3}, \frac{2}{5}, \frac{3}{7}\right)$.  
  We analyse this manifold to demonstrate the relation between $\widehat{Z}^{SU(2)}_{\vec{\underline{b}}}$ and {log-${\cal V}_{\bar \Lambda}^0(p,p')$ characters described in Theorem \ref{thm:4fibchimatching}. 
The plumbing matrix of $X_\Gamma$ is
\begin{equation}
M=\left(\begin{array}{rrrrrrrrr}
-2 & 1 & 1 & 0 & 1 & 0 & 1 & 0 & 0 \\
1 & -2 & 0 & 0 & 0 & 0 & 0 & 0 & 0 \\
1 & 0 & -2 & 1 & 0 & 0 & 0 & 0 & 0 \\
0 & 0 & 1 & -2 & 0 & 0 & 0 & 0 & 0 \\
1 & 0 & 0 & 0 & -3 & 1 & 0 & 0 & 0 \\
0 & 0 & 0 & 0 & 1 & -2 & 0 & 0 & 0 \\
1 & 0 & 0 & 0 & 0 & 0 & -3 & 1 & 0 \\
0 & 0 & 0 & 0 & 0 & 0 & 1 & -2 & 1 \\
0 & 0 & 0 & 0 & 0 & 0 & 0 & 1 & -2
\end{array}\right)
\end{equation}
and, $M_{v_0,v_0}^{-1}=m=p_1p_2p_3p_4=210$.
Because $X_\Gamma$ is a spherical manifold all $\hat{w}$ produce a non-empty $S_{\hat w;\vec{\underline{b}}}$, as defined in \eqref{dfn:theset}.
There are eight different choices of admissible pairs $p,p'$, corresponding to $(p,p')=\left(p_i,\frac{m}{p_i}\right)$ and their permutations $p\leftrightarrow p'$:
\begin{equation}
(p,p')\in \left\{(2,105),(3,70),(5,42),(7,30)\right\}.
\end{equation}
Using  \eqref{eq:rs4f} we can produce all possible $s_{w_1,w_2,w_3}$ for each pair.
Independent $s_{w_1,w_2,w_3}$ are collected in Table \ref{tab:pp'svals}, while the remaining ones are obtained with an extra a minus sign.
\begin{table}
\centering
\begin{tabular}{c|cccc}
$(p,p')$,$\hat{w}$ & $(-\id,-\id,-\id)$ & $(-\id,\id,\id)$ &$(\id,-\id,\id)$ & $(\id,\id,-\id)$ \\\toprule
$(2,105)$ & $71$ & $-1$ & $-29$ & $-41$ \\
$(3,70)$& $59$ & $11$ & $-31$ & $-39$ \\
$(5,42)$& $41$ & $1$ & $-13$ & $-29$ \\
$(7,30)$& $31$ & $-1$ & $-11$ & $-19$ \\
\end{tabular}
\caption{$s_{w_1,w_2,w_3}$ for $X_\Gamma=M\left(-2, \frac{1}{2}, \frac{2}{3}, \frac{2}{5}, \frac{3}{7}\right)$, for different choices of $p,p'$}.\label{tab:pp'svals}
\end{table}

Using the data above one may verify the main claim in Theorem \ref{thm:4fibchimatching}:
\begin{align}
\sum_{\hat{w}\in W^{\otimes 4}}(-1)^{l(w)}\tilde{\chi}^{\prime}_{\hat{w};\vec{\underline{b}}}(\tau,\xi)&=
2\left(q^{145/4}\frac{\Delta(2\xi)-2}{\Delta^2(\xi)}+q^{(261/4)}\frac{\Delta(2\xi)-2}{\Delta^2(\xi)}\right)+O(q^{317/4})\nonumber\\
&=q^{\delta} \eta(\tau)\sum_{w_1,w_2,w_3\in W} (-1)^{\ell({w_1})+\ell({w_2})+\ell({w_3})}
 {\rm ch}^{+}_{r,s_{w_1,w_2,w_3}}(\tau,\xi).
\end{align}

\vspace{15pt}
\noindent
{\bf A Pseudo-Spherical Example}

In pseudo-spherical and non-spherical cases, 
non-integer entries in the inverse plumbing matrix cause the $S_{\hat{w};\vec{\underline{b}}}$ set to be empty for some $\hat{w}$, and the corresponding  $\tilde{\chi}'_{\hat{w};\vec{\underline{b}}}$ to vanish.
In most cases, this results in the absence of pairings between $\widehat{Z}$ integrands with $\hat{w}$ and $\hat{w}'$ such that $w_i=-w_i'$ and $w_j=w_j',\ i\neq j$ which are a key assumption necessary for Theorem \ref{thm:4fibchimatching}.

Nonetheless, for some non-spherical manifolds such pairings do exist, and hence Theorem \ref{thm:4fibchimatching} is also applicable and a set of $(p,p')$ may be found to produce characters satisfying a similar equation to \eqref{eq:4fchitchmatch}.
One such example is the pseudo-spherical Seifert manifold $X_\Gamma = M\left(-1;-\frac{1}{2},\frac{1}{3},\frac{1}{3},\frac{2}{3}\right)$.
The plumbing matrix of $X_\Gamma$ is given by:
\begin{equation}
M=\left(\begin{array}{rrrrrr}
-1 & 1 & 1 & 1 & 1 & 0 \\
1 & 2 & 0 & 0 & 0 & 0 \\
1 & 0 & -3 & 0 & 0 & 0 \\
1 & 0 & 0 & -3 & 0 & 0 \\
1 & 0 & 0 & 0 & -2 & 1 \\
0 & 0 & 0 & 0 & 1 & -2
\end{array}\right).
\end{equation}
Its inverse is:
\begin{equation}
M^{-1}=\left(\begin{array}{rrrrrr}
-6 & 3 & -2 & -2 & -4 & -2 \\
3 & -1 & 1 & 1 & 2 & 1 \\
-2 & 1 & -1 & -\frac{2}{3} & -\frac{4}{3} & -\frac{2}{3} \\
-2 & 1 & -\frac{2}{3} & -1 & -\frac{4}{3} & -\frac{2}{3} \\
-4 & 2 & -\frac{4}{3} & -\frac{4}{3} & -\frac{10}{3} & -\frac{5}{3} \\
-2 & 1 & -\frac{2}{3} & -\frac{2}{3} & -\frac{5}{3} & -\frac{4}{3}
\end{array}\right)
\end{equation}
from which we read $m=6$.
The Cokernel of $M$ is:
\begin{align}
\text{Coker}(M)=\{&\left(0,\,0,\,0,\,0,\,0,\,0\right), \left(1,\,0,\,0,\,-1,\,-1,\,0\right),\nonumber\\
&\left(1,\,0,\,0,\,-2,\,0,\,-1\right), \left(1,\,0,\,-1,\,0,\,-1,\,0\right),\nonumber\\
&\left(1,\,0,\,-2,\,0,\,0,\,-1\right), \left(2,\,0,\,-2,\,-2,\,-1,\,0\right),\nonumber\\
&\left.\left(1,\,0,\,-1,\,-1,\,0,\,-1\right)\right\}
\end{align}
which leads to four possible $\vec{\underline{b}}$:
\begin{align}
\mathcal{B}=&\left\{\vec{\underline{b}}_{i}\right\}_{i=0,\dots,3}\nonumber\\
=&\left\{\left(\left(-2\right), \left(1\right), \left(1\right), \left(1\right), \left(0\right), \left(1\right)\right),\right.\left(\left(0\right), \left(1\right), \left(-1\right), \left(1\right), \left(-2\right), \left(1\right)\right),\nonumber\\
&\left(\left(0\right), \left(1\right), \left(1\right), \left(-1\right), \left(-2\right), \left(1\right)\right),
\left.\left(\left(2\right), \left(1\right), \left(-3\right), \left(-3\right), \left(-2\right), \left(1\right)\right)\right\}.
\end{align}
The only combinations of $\vec{\underline{b}}$ and $\hat{w}$ which give non-empty $S_{\hat{w};\vec{\underline{b}}}$ are:
\begin{align}
    b_0: &\quad W_0 = \left\{(\id,-\id,-\id,-\id),\ (\id,-\id,-\id,-\id)\right\}\nonumber\\
    b_1: &\quad W_1 = \emptyset\nonumber\\
    b_2: &\quad W_2 = \emptyset\nonumber\\
    b_3: &\quad W_3 = \left\{(\id,\id,\id,-\id),\  (-\id,\id,\id,-\id)\right.\nonumber\\
&\quad\left.\qquad\ \ \ (\id,-\id,-\id,\id),\ (-\id,-\id,-\id,\id)\right\}.
\end{align}

Because $S_{\hat{w};\vec{\underline{b}}}$ is empty for $\vec{\underline{b}}=\vec{\underline{b}}_1$ and $\vec{\underline{b}}=\vec{\underline{b}}_2$ the homological blocks for such choices vanish. 
For $\vec{\underline{b}}_0$ and $\vec{\underline{b}}_3$ cases, we can pair $\hat{w}$ to $\hat{w}'$ such that ${w}_1=-{w}'_1$ and ${w}_i={w}'_i$ for $i=2,3,4$, leading to  $(p,p')=(2,3)$ as the only choice.
The $s$ values corresponding to the $\hat w$ that give non-vanishing homological blocks are gives:
\begin{align}
s_{-\id,-\id,-\id} &=3 \nonumber\\
s_{\id,\id,-\id} &=-1 \nonumber\\
s_{-\id,-\id,\id} &=1 \nonumber.
\end{align}
With this data, the analogue of \eqref{pprime_integrand} reads:
\begin{align}
\sum_{\hat{w}\in W_0}\tilde{\chi}'_{\vec{\hat{w}};\underline{b}_0}(\tau,\xi)
&=q^{5/24}\left(-q^{3/8}+q^{27/8}+q^{75/8}\left(1-\frac{\Delta(4\xi)-2}{\Delta^2(\xi)}\right) \right)+O(q^{18})\nonumber\\
    &=q^{\delta} \eta(\tau)\,{\rm ch}^{+}_{1,3}(\tau,\xi)
\end{align}
and:
\begin{align}
\sum_{\hat{w}\in W_{3}}\tilde{\chi}'_{\vec{\hat{w}};\underline{b}_3}(\tau,\xi)&=2q^{5/24} \left(q^{49/24} -q^{121/24} -q^{169/24} +q^{289/24} \right)+O(q^{15})\nonumber\\
    &=q^{\delta} \eta(\tau)\left({\rm ch}^{+}_{1,1}(\tau,\xi)-{\rm ch}^{+}_{1,-1}(\tau,\xi)\right)
\end{align}

\subsection{\texorpdfstring{$\widehat{Z}$}{\^Z}-invariants with Line Operators}
\label{subsec:line-ope-ex}

As discussed in \S\ref{subsec:lineop}, the insertion of Wilson operators allows us to access different log VOA characters  through the $\widehat Z$-invariants.
In this subsection we provide examples of Propositions \ref{prop:Wil-End}, \ref{prop:Wil-Center} and  \eqref{eqn:intermediate_Wilson}.
Without further remark, in this subsection we exemplify the new phenomena when incorporating line operators using the same manifolds as those in \S\ref{subsec:3fib-ex}, and simply referring them as the ``spherical manifold", ``pseudo-spherical manifold" etc.  

\vspace{15pt}
\noindent
{\bf Wilson Operators at an End Node}

Theorems \ref{thm:combining_into_characters} and  \ref{3fiber_sphere} can also be applied to spherical, non-spherical and pseudo-spherical Seifert manifolds when Wilson operators are inserted at end nodes of the plumbing graphs legs.
Such generalizations merely require a substitution of $\delta$ and $\vec{A}_{\hat{w}}$ with the definition in \eqref{dfn:AW}.

Theorem \ref{3fiber_sphere} will apply to the spherical and pseudo-spherical cases.
Let $\vec{\nu} = (1,4)$ be the highest weight of the $A_2$ representation.
For the spherical case we find
\begin{align}
    \vec{A}_{\hat{w}} &= {-1\over M_{v_0,v_0}^{-1}} \sum_{v\in V_1} M_{v_0,v}^{-1} w_v(\vec\rho_v) = {206\over315}\vec\omega_1+ {79\over63}\vec\omega_2\nonumber\\
    \delta &= \frac{1853}{315}\nonumber
\end{align}
from which we can compute $\sqrt{m}\mu=\rho+m\vec{A}_{\hat{w}}=207 \vec\omega_1 + 396 \vec\omega_2$ and $s_1 = -206$, $s_2 = -395$.
Hence we can compute, for $\hat{w}=(\id,\id,\id)$
\begin{align}
\tilde{\chi}_{\hat{w};\vec{\underline{b}}_0}(\tau,\xi) &= - q^{16} + q^{222} - \frac{q^{360} \Delta\left(2 \, \xi_{1}, 2 \, \xi_{2}\right)}{\Delta\left(\xi_{1}, \xi_{2}\right)} + q^{411} +O(q^{772})\\
&=q^\delta\eta^2(\tau)\chi_{\vec{\lambda}_{0,0,s_1,s_2}}(\tau,\xi)
\end{align}
which provides a numerical confirmation of Theorem \ref{3fiber_sphere} for spherical manifolds.
A similar result can also be obtained for the pseudo-spherical case.
With the same highest weight $\vec \nu$ we get $\vec{A}_{\hat{w}}={13\over 9}\vec\omega_1 + {53\over18}\vec\omega_2$, $\delta = 125/18$,  $\sqrt{m}\mu = 39\omega_1+42\omega_2$ and, therefore, for $\hat{w}=(\id,\id,\id)$
\begin{align}
\tilde{\chi}_{\hat{w};\vec{\underline{b}}_0}&=
- \frac{q^{11} \Delta\left(2  \xi_{1}, 2  \xi_{2}\right)}{\Delta\left(\xi_{1}, \xi_{2}\right)}
- \frac{q^{12} \Delta\left(4  \xi_{1}, \xi_{2}\right)}{\Delta\left(\xi_{1}, \xi_{2}\right)} 
- \frac{q^{22} \Delta\left(3 \xi_{1}, 3  \xi_{2}\right)}{\Delta\left(\xi_{1}, \xi_{2}\right)} 
- q^{36} + O(q^{38})\nonumber\\
&=q^{\delta}\eta^2\left( \tau \right)\chi_{\vec{\lambda}'_{0,0,s_1,s_2}}(\tau,\vec{\xi})\nonumber.
\end{align}

For non-spherical cases Theorem \ref{thm:combining_into_characters} is similarly generalized. 
With $\vec \nu = (1,4)$, $\hat{w}=(\id,\id,\id)$ we get, for the first non-spherical example, $\vec{A}_{\hat{w}}=-1/30\vec \omega_1 + 1/15\vec\omega_2$, $\delta = 10/3$ and $\sqrt{m}\mu=4/5\omega_1 + 7/5\omega_2$, and
\begin{align}
\tilde{\chi}_{\hat{w};\vec{\underline{b}}_0}&= -\frac{\Delta(7\xi_1,\xi_2)}{\Delta(\xi_1,\xi_2)}q^{7357/50}+O(q^{13457/50}).
\end{align}
The triplet character is
\begin{align}
\chi_{\vec{\lambda}'_{0,0,s_1,s_2}}(\tau,\vec{\xi})&=q^{1947/50}+\dots+ \frac{\Delta(4\xi_1,10\xi_2)}{\Delta(\xi_1,\xi_2)}q^{17307/50}+O(q^{17407/50}).
\end{align}
 Summing over $\Delta \vec b$, we recover the full character
\begin{align}
\sum_{\substack{\vec{\underline{b}}=\vec{\underline{b}}_0+\left(\Delta \vec b,0,0,0 \right)\\ \Delta \vec b\in \Lambda/D\Lambda}}\tilde{\chi}_{\hat{w};\vec{\underline{b}}}(D\tau,\vec{\xi})&= q^{1947/50}- 2q^{1957/50}
+2q^{1977/50}-q^{1987/50}+O(q^{2827/50})
\nonumber\\
&=q^{D\delta}\eta^2\left( \tau \right)\chi_{\vec{\lambda}'_{0,0,s_1,s_2}}(\tau,\vec{\xi})
\end{align}
as per Theorem \ref{thm:combining_into_characters}.

\vspace{15pt}
\noindent
{\bf Wilson Operators at the Central Node}\\

When a Wilson operator is inserted at the central node of the plumbing graph, integrands of  $\widehat{Z}^G_{\vec{\underline{b}}}$ obtain an extra polynomial in $z_i$ (cf. \eqref{eqn:zhat_wilson_central}) as a multiplicative factor.
As argued in Proposition \ref{prop:Wil-Center}, for spherical and pseudo-spherical Seifert manifolds with three exceptional fibers,  the $\widehat{Z}$ invariant can still be expressed as a linear combination of singlet characters \eqref{singlet_char}.

  Consider the spherical example described above with a Wilson operator of highest weight $\vec{\nu} = (3,3)$ inserted at the central node of the plumbing graph.
Up to Weyl group action, weights of nonzero multiplicity in the highest weight module of highest weight $\vec{\nu}$ are $\vec{\sigma} \in \left\{(2,2),(3,0),(0,3),(1,1),(0,0)\right\}$ with multiplicities $m^{(\vec\nu)}_{\vec\sigma}$ given by  $1,1,1,2,3$ respectively.
Using \eqref{eqn:zhat_wilson_central} we can compute
\begin{align}
\widehat{Z}^{A_2}_{\vec{\underline{b}}_0}(X_\Gamma, W_{\vec \nu_{v_0}};\tau) &=108 q^{95} - {216} q^{118} - {216} q^{126} - {216} q^{142} + \mathcal{O}\left(q^{143}\right)\nonumber\\
&=q^{\delta}\eta(\tau)^2\sum_{\vec \sigma\in P^+} m^{(\vec \nu)}_{\vec\sigma}  \sum_{w \in W}
  \sum_{\hat{w}\in W^{\otimes 3}} 
 \chi^0_{\vec\mu_{\hat{w}}-\sqrt{m} w(\vec\sigma)}.
\end{align}
For pseudo-spherical examples not all $\hat{w}$ contribute to a nonempty $S_{\hat{w};\vec{\underline{b}}}$, and therefore not all $\vec{\mu}_{\hat{w}}$ contribute to the sum over singlet characters.
The $\widehat{Z}$-invariant is:
\begin{align}
\widehat{Z}^{A_2}_{\vec{\underline{b}}_0}(X_\Gamma, W_{\vec \nu_{v_0}};\tau) &=
108 q + 216 q^{5} - 432 q^{6} + 216 q^{13} + {216} q^{16} + \mathcal{O}\left(q^{17}\right)\nonumber\\
&=q^{\delta}\eta(\tau)^2\sum_{\vec \sigma\in P^+} m^{(\vec \nu)}_{\vec\sigma}  \sum_{w \in W}
  \sum_{\hat{w}\in W^{\otimes 3}} a_{\mu_{\hat{w}}}
 \chi^0_{\vec\mu_{\hat{w}}-\sqrt{m} w(\vec\sigma)}
\end{align}
where 
\begin{equation}
a_{\mu_{\hat{w}}} = \begin{cases} 
0 & \text{if } S_{\hat{w};\vec{\underline{b}}}=\emptyset\\
1 & \text{otherwise}
\end{cases}.
\end{equation}

\vspace{15pt}
\noindent
{\bf Wilson Operators at an Intermediate Node}\\

Similar to the case of Wilson operators at an end node, the insertion of a Wilson operator at an intermediate node results in a modification of $\vec A_{\hat{w}}$ and $\delta$.
This introduces an extra dependence of $\vec A_{\hat{w}}$ and $\delta$ on a weight $\vec \sigma$, and on $w'\in W$
Using the definitions in equation \eqref{eq:AwWilMid}, in the spherical case, for $\hat{w}=(\id,\id,\id)$, $w' = a$, and $\vec\sigma = (1,4)$, we get $\vec{A}_{\hat{w},w'}= {332\over315}\vec\omega_1 -{802\over315}\vec\omega_2$, $\delta_{\hat{w},w'}=1916/315$, and
\begin{align}
\tilde{\chi}_{\hat{w},w';\vec{\underline{b}}_0}(\tau,\xi) &= -q^{173/5}- \frac{\Delta(2\xi_1,2\xi_2)}{\Delta(\xi_1,\xi_2)}q^{888/5} + q^{1833/5} + q^{2523/5}  +O(q^{3263/5})\nonumber\\
&=q^\delta\eta^2(\tau)\chi_{\vec{\lambda}_{0,0,s_1,s_2}}(\tau,\xi).
\end{align}
Hence we can rewrite each integrand in \eqref{eq:ZhatWmid} as a triplet character, and, by so doing, we can write the $\widehat{Z}$ invariant as a linear combination of singlet characters, multiplied by additional individual rational factors of $q$.

We consider $X_\Gamma = M (-1; \frac{2}{3},-\frac{1}{2},-\frac{1}{2})$, a non-spherical Seifert manifold with plumbing matrix
\begin{equation}
M=\left(\begin{array}{rrrrr}
-1 & 1 & 0 & 1 & 1 \\
1 & -2 & 1 & 0 & 0 \\
0 & 1 & -2 & 0 & 0 \\
1 & 0 & 0 & 2 & 0 \\
1 & 0 & 0 & 0 & 2
\end{array}\right)
\end{equation}
and with $D=8$ and $m = 6$.
In the case where a Wilson operator, of highest weight $\vec{\nu} = (1,4)$ is inserted at the mid point of the first leg.
The computation of $\tilde{\chi}$ directly follows from \eqref{eq:AwWilMid}, from which, with $\hat{w}=(\id,\id,\id)$ and $w'=a$  gives
\begin{equation}
\tilde{\chi}_{\hat{w},w';\vec{\underline{b}}_0}(\tau,\xi) = - \frac{\Delta(3\xi_1,6\xi_2)}{\Delta(\xi_1,\xi_2)}q^{739/8} + O(q^{1667/8})\\
\end{equation}
whereas the triplet character is
\begin{equation}
q^\delta(\eta(\tau))^2\chi_{\vec{\lambda}_{0,0,s_1,s_2}}(\tau,\xi)= - \frac{\Delta(2\xi_1,2\xi_2)}{\Delta(\xi_1,\xi_2)}q^{387/8} + \dots - \frac{\Delta(3\xi_1,6\xi_2)}{\Delta(\xi_1,\xi_2)}q^{739/8} +O(q^{747/8}).
\end{equation}
To recover the full character we sum over $\Delta \vec b\in \Lambda/D\Lambda$

\begin{align}
\sum_{\substack{\vec{\underline{b}}=\vec{\underline{b}}_0+\left(\Delta\vec b,0,0,0 \right)\\ \Delta \vec b\in \Lambda/D\Lambda}} \tilde{\chi}_{\hat{w};\vec{\underline{b}}}(D\tau,\vec{\xi}) &= -q^{387/8}\frac{\Delta(2\xi_1,2\xi_2)}{\Delta(\xi_1,\xi_2)}-q^{395/8}\frac{\Delta(\xi_1,4\xi_2)}{\Delta(\xi_1,\xi_2)}-q^{407/8}+O(q^{423/8})
\nonumber\\
&=q^{D\delta}\eta^2\left( \tau \right)\chi_{\vec{\lambda}'_{0,0,s_1,s_2}}(\tau,\vec{\xi}).
\end{align}

\section{Discussions and Future Directions}

We close this work with a list of questions and suggestions for future research:

\begin{itemize}

    \item Quantum spectral curves like \eqref{qspectral} are ubiquitous in $SL(2,\mathbb{C})$ Chern-Simons theory and in 3d-3d correspondence. However, their role in vertex algebra is less clear, aside from the obvious fact that such $q$-difference equations encode dependence of VOA characters on the fugacity associated with a symmetry of the VOA. In particular, even the classical limit of the quantum curve that we worked out for the triplet algebra requires a better understanding and interpretation, from the VOA perspective.

    \item Following \cite{Gadde:2013wq}, in the earlier work \cite{Cheng:2018vpl} and in \S \ref{sec:physics} of this paper we offered a physical explanation for the departure from the classical modular properties of the BPS half-indices in 3d $\mathcal{N}=2$ theories with 2d $(0,2)$ boundary conditions. Since characters of logarithmic vertex algebras exhibit similar deviations from traditional modularity, it is perhaps not surprising to find many relations of the form \eqref{ZhatVOA}, which is indeed one of the main results of this work. While we were able to build a large dictionary between $q$-series invariants of families of 3-manifolds and characters of log VOAs, it would be interesting to develop new tools that allow us to access other aspects of a log VOA, {\it e.g.} other conformal blocks, directly from the data of a 3-manifold. We hope this can be achieved by a further study of the 3d-3d correspondence and better understanding of the relation between the category of log VOA modules and the algebraic structure $\text{MTC} [X]$ of line operators in 3d theory $T[X]$, along the lines of \cite{Gukov:2016gkn,Cheng:2018vpl,Costello:2018swh,Gukov:2020lqm,Costantino:2021yfd,Creutzig:2021ext}.

    \item The relations between non-traditional (``exotic'') forms of modularity, logarithmic vertex algebras, and $q$-series invariants of 3-manifolds also involve quantum groups at generic values of $q$ and their various specializations:
\begin{equation}
\begin{tikzcd}
& {\text{log VOAs}} \ar[dd, leftrightarrow] \ar[dl, leftrightarrow] \ar[dr, leftrightarrow] &  \\
\widehat Z \text{-invariants} \ar[dr, leftrightarrow] \ar[rr, leftrightarrow] & ~~~~ & {\text{quantum} \atop \text{groups}} \ar[dl, leftrightarrow] \\
& {\text{quantum} \atop \text{modular forms}}  &  
\end{tikzcd}
\label{theweb}
\end{equation}
    In this work we focused on the upper and the left, and the lower to a lesser extent,  nodes in this diagram. It is however important to stress that the relation to quantum groups also plays an important role in these connections, see {\it e.g.} \cite{Park:2020edg,Park:2021ufu} and the upcoming work \cite{toappear}.
    
    \item The relation discussed in the present work between three-manifold invariants and log VOAs is far from being one-to-one.
    In particular, in Proposition \ref{prop:main_generalSeif} and Theorem \ref{3fiber_sphere}, the log VOA in question only depends on
    $m =  -D M_{v_0,v_0}^{-1}$. 
    A natural question is hence whether there is an extension of the {log-${\cal V}_{\bar \Lambda}$} algebra such that $\widehat Z^G(X_\Gamma)$ is related to the algebra in an even closer way? 
     In \cite{Cheng:2018vpl} we discussed the Weil representation attached to $\widehat Z^{SU(2)}(X_\Gamma)$ when 
     $X_\Gamma$ is a negative Seifert manifold with three exceptional fibers. We believe that this is a crucial property that provides important hints for the search of such an extended algebra ${\cal V}^G_{X_\Gamma}$.   
    \item 
    In Theorem  \ref{thm:combining_into_characters}, 
    we see that the integrands of a specific  combination of $\widehat Z^G_{\underline{\vec b}}$, with the summand labelled by $\Lambda/D\Lambda$, are given by log VOA (generalised) characters. While $\widehat Z^G_{\underline{\vec b}}$ with different generalised {Spin$^c$} structures are independent topological invariants, it has been noticed   that sometimes one has to combine different  $\underline{\vec b}$ 
     to recover various known topological invariants. See \cite{Gukov:2020frk} for interesting examples.  In this sense, what we have found in this work is an analogous phenomenon.  In this regard, natural questions include the following. What is the topological meaning of the parameter $D$, and in particular, what is the meaning for the manifold to be ``pseudo-spherical", namely to have $D=1$? For the case of $D>1$, does an individual $\widehat Z^G_{\underline{\vec b}}$ with a given ${\underline{\vec b}}$ have an interpretation in terms of the log VOA?       
     
    \item The results in \S\ref{subsec:3fiber} relate the  $\widehat{Z}^G$ invariants and log VOA characters for \textit{all} simply-laced gauge groups $G$. The relation in particular holds for $G = SU(N)$ for all positive integers $N$. 
      It is therefore natural to consider the large-$N$ behaviour of  the  {log-${\cal V}_{\bar \Lambda_{SU(N)}}(m)$} models.
    More specifically, an effective variable $a = q^N$ is expected to play a role in the homological blocks $\widehat Z^{SU(N)}$ \cite{Gukov:2016gkn,Ekholm:2020lqy,Ekholm:2021irc}.
    In the large-$N$ limit, then, could there be a triply graded version of the {log-${\cal V}_{\bar \Lambda}$} model, with an additional $a$-grading? We expect the answer to be in affirmative and the corresponding log-VOA to be an analogue of the triplet algebra where a finite-dimensional symmetry is replaced by an infinite-dimensional symmetry {\it a la} Yangian.
    
	 \item 
	 It has been conjectured  that the homological invariants $\widehat Z$ are closely related to quantum modular forms  (defined by Zagier \cite{zagier2010quantum})  in some way \cite{Cheng:2018vpl}.  See \cite{Cheng:2018vpl,bringmann2020quantum,bringmann2020higherdepth} for earlier results for the $G=SU(2)$ case. At the same time, the quantum modular properties of the log VOA characters have been an active area of research \cite{MR3624911,BM1,MR3919493,BM2021}, and this immediately leads to some results on the quantum modular properties of $\widehat Z^{G}$ for $G\neq SU(2)$. 
     An in-depth analysis of the  modular properties of $\widehat Z^{G}$ for $G= SU(3)$ will appear in an upcoming paper  \cite{quantum_mod_rank2}. It would be interesting to further develop the triangular relation between quantum modular forms, log VOAs, and homological blocks, as depicted in \eqref{theweb}.
     
\item In this work, we mainly focus on negative Seifert manifolds for the sake of concreteness. It would be very interesting to explore the VOAs corresponding to other weakly negative plumbed three-manifolds. In particular, it would be very interesting if one could construct the VOAs with a procedure reflecting the operations on plumbing graphs, such as connecting weighted graphs into a larger one. 

     \item In the present work we focus on {\em weakly negative plumbed three-manifolds}, and in particular negative Seifert manifolds. 
     In \cite{Cheng:2018vpl}, it was proposed that the role of (higher rank) false theta functions in these cases will be replaced by (higher depth) mock modular forms in the case of plumbed manifolds that are not weakly negative, and in particular {\em positive} Seifert manifolds. A natural question is thus to find the VOAs connected to the homological blocks of such three-manifolds, in a systematic manner analogous to the results presented in the present paper. 
     
	 \item Understanding of the relations in \eqref{theweb} can be greatly facilitated by the fermionic form of log VOA characters and $\widehat{Z}$-invariants. In this paper, we only made some initial steps in this direction, leaving many interesting questions to future work. For example, the appearance of classical and quantum dilogarithms on both sides of the 3d-3d correspondence suggests many connections to cluster algebras which, while natural, so far did not appear in the study of $\widehat{Z}$-invariants of log VOA. The relation between cluster algebras and $\widehat{Z}$ TQFT is also expected because the latter provides a non-perturbative definition of $SL(2,\mathbb{C})$ Chern-Simons theory, whereas the relation between cluster algebras and log VOAs is natural in view knot-quiver correspondence and recent work \cite{Cecotti:2015lab,Jankowski:2021flt}.
	 
	 \item The quiver/fermionic form of $q$-series invariants also offers new ways of addressing long-standing questions in logarithmic vertex algebras. For example, it offers a fresh new perspective on the ``semi-classical'' limit ($m \to \infty$) of the {log-${\cal V}_{\bar \Lambda}^0(m)$} model  and going from the ``positive zone'' in Kazhdan-Lusztig correspondence to the ``negative zone,'' on which we plan to report elsewhere. The counterpart of this question in quantum topology is the reversal of orientation on $X$ and the relation between $\widehat{Z} (X)$ and $\widehat{Z} (-X)$.
	 
	 \item The quiver/fermionic forms discussed in this paper are {\it virtual} characters of the familiar vertex algebras like triplet and singlet log VOAs. In other words, these fermionic formulas are linear combinations of the characters of irreducible modules. As was mentioned in the fourth bullet point of this section, this seems to suggest that logarithmic VOAs associated to 3-manifolds are extensions of these familiar algebras, at least in simple examples. This point should be important for upgrading the dictionary between 3-manifolds and characters to actual VOAs.
	 
\end{itemize}

\section*{Acknowledgements}
We would like to thank Kathrin Bringmann, John Cardy, Thomas Creutzig, Tobias Ekholm, Pavel Etingof, Dennis Gaitsgory, Angus Gruen, Antun Milas, Piotr Kucharski, Sunghyuk Park, Du Pei, Nicolai Reshetikhin, Vyacheslav Rychkov, Marko Stosic, Piotr Sulkowski and Alexander Zamolodchikov for helpful discussions.
The work of M.C. and D. P. is supported by an NWO vidi grant (number 016.Vidi.189.182).
The work of S.C. is supported by the US Department of Energy under
grant DE-SC0010008.
The work of S.G.~is supported by the U.S.~Department of Energy, Office of Science, Office of High Energy Physics, under Award No.~DE-SC0011632, and by the National Science Foundation under Grant No.~NSF DMS 1664227.
The work of B. F. has been funded within the framework of the HSE University Basic Research Program. The work of S.M.H. is supported by the National Science and Engineering Council of
Canada and the Canada Research
Chairs program.  This research was initiated at the American Institute of Mathematics (AIM) as part of the SQuaRE meeting.

\appendix
\section{Plumbing Matrix for Seifert Manifolds}
\label{sec:proofMvanish}

In this appendix we establish a few properties of the plumbing matrix for Seifert manifolds.

\begin{lem}	
	\label{lem:Midentity}~~
	\begin{enumerate}
		\item 
The following holds for the plumbing matrix $M$ for all Seifert manifolds. 
\be\label{eq:Midentity2}
M_{v_0,v_0}^{-1}  M_{v,v'}^{-1}  -M_{v_0,v}^{-1}  M_{v_0,v'}^{-1}  =0 \;\;\text{ for all  }v,v'\in V_1, v\neq v' , 
\ee

\item 
The $q$ power $\delta$ defined in (\ref{dfn:A}) can be expressed 
In terms of the plumbing data of the Seifert manifolds with Seifert data $M(b;\{q_i/p_i\}_i)$ as
\be\label{eqn:delta}
\delta = {|\vec\rho|^2\over 2}  \sum_{i=1}^N \left( {1\over p_i q_i} - {\theta^{(i)}_{\ell_i-1}\over \theta^{(i)}_{\ell_i}}\right) ,
\ee
where ${\theta^{(i)}_{k}}$s are defined in \eqref{dfn:theta}. 

\item
	The vector $\vec A_{\hat w}$ defined in (\ref{dfn:A}) is given in terms of the Seifert data $M(b;\{q_i/p_i\}_i)$ as 
\begin{equation}\label{Arel}
\vec A_{\hat w} = \sum_{v_i \in V_1} 
{{\rm sgn}(q_i)\over p_i} w_{v_i}(\vec \rho)
\end{equation}	
In the above, $v_i$ denotes the end point of the $i$-th leg with Seifert data $q_i/p_i$ and $\textfrak{e} = b+\sum_i {q_i\over p_i}$ is the  Euler number of the Seifert manifold. 

	\end{enumerate}	

\end{lem}

\begin{proof}
	  For a block matrix of the form,
	 $$M= \begin{pmatrix} A & B \\ C &D \end{pmatrix},$$
	 the inverse takes the form
	 \begin{gather}
	 \begin{split}
	 \label{eqn:brblock}
	 M^{-1}&=
	 \begin{pmatrix}  (A-BD^{-1}C)^{-1} & -(A-BD^{-1}C)^{-1}BD^{-1} \\  -D^{-1}C(A-BD^{-1}C)^{-1} &  D^{-1} + D^{-1}C(A-BD^{-1}C)^{-1}BD^{-1} \end{pmatrix}
	 \\
	 &= \begin{pmatrix} A^{-1} + A^{-1} B (D-CA^{-1}B)^{-1}CA^{-1} & -A^{-1}B(D-CA^{-1}B)^{-1} \\ - (D-CA^{-1}B)^{-1}CA^{-1} & (D-CA^{-1}B)^{-1} \end{pmatrix}\, .
	 \end{split}
	 \end{gather}
	 
	 For a $N$-star graphs, let $a$ be the weight of the central node, denoted $v_0$; let $\ell_i+1$ be the length of the $i$th leg of the plumbing graph, for $i \in \{1,\ldots, N\}$, $L:= \sum_{i=1}^N \ell_i$, and let $a^{(i)}_m$ be the weight of the $m$th node along leg $i$, denoted $v^{(i)}_m$. We order the vertices such that $v_0^{(i)}$ is adjacent to the central node and $v_{\ell_i}^{(i)}$ is the end point of the $i$-th leg.
	 
	 We can then write the linking matrix as a block matrix of the form,
	 \be
	 M=\left(\begin{array}{@{}c|c@{}} M_0 &B \\\hline B^T & X\end{array} \right ). 
	 \ee
	 In the above, $M_0$ is an $(N+1) \times (N+1)$ matrix capturing the central node and the $N$ nodes adjacent to it,
	 \be 
	 M_0=\begin{pmatrix}
	 	a & 1 & 1 &\cdots & 1 \\
	 	1 & a^{(1)}_0 & 0& \cdots &0 \\
	 	1 & 0 & a^{(2)}_0 & \cdots & 0\\
	 	\vdots & \vdots & \vdots & \ddots&\vdots\\
	 	1 &0&0 & \cdots & a^{(N)}_0
	 \end{pmatrix},
	 \ee
	 and 
	 $B$ is an $(N+1) \times L$ matrix with entries $B_{0\beta}=0$, 
	 \be
	 B_{i\beta}=\begin{cases} 1 & ~{\rm when}~\beta = 1+ \sum_{j=1}^{i-1} \ell_j \\ 0 & ~{\rm otherwise} \end{cases}
	 \ee
	 for $i\in\{1,\ldots, N\}$, $\beta\in \{1,\dots, L\}$, 
	 capturing the linking between nodes $v^{(i)}_0$ and $v^{(i)}_1$ for $i \in {1, \ldots, N}$. Finally, $D$ is square matrix of size $L\times L$ and is  block-diagonal  of the form
	 \be
	 X=\begin{pmatrix}
	 	X_1 & 0&\cdots &0 \\
	 	0 &X_2 &\cdots &0 \\
	 	\vdots & \vdots & \ddots & \vdots\\
	 	0 & 0 &\cdots & X_N
	 \end{pmatrix},
	 \ee
	 where $X_i$ is an $\ell_i \times \ell_i$ tridiagonal linking matrix
	 \be
	 X_i=\begin{pmatrix}
	 	a^{(i)}_1 & 1 & 0 & \cdots&0 \\
	 	1 & a^{(i)}_2 & 1 & \cdots&0 \\
	 	0& 1 & a^{(i)}_3 & \cdots&0\\
	 	\vdots&\vdots&\vdots&\ddots&\vdots\\
	 	0& &  & 1& a^{(i)}_{\ell_i}
	 \end{pmatrix},
	 \ee
	 for the $i$th leg.
	 
	 Equation (\ref{eqn:brblock}) states that
	 \be\label{prooflem1_eq1}
	 M^{-1}=\left(\begin{array}{@{}c|c@{}} \tilde M^{-1} &-\tilde M^{-1}BX^{-1} \\\hline -X^{-1}B^T\tilde M^{-1}&X^{-1}+ X^{-1}B^{T}\tilde M^{-1} B X^{-1}\end{array} \right ). 
	 \ee
	 where we have defined the $(N+1) \times (N+1)$ matrix
	 $$\tilde M := M_0-B X^{-1} B^T.$$

	 Given the explicit form of the matrices $B$ and $X$, it is straightforward to compute that the $(N+1) \times (N+1)$ matrix $B X^{-1} B^T$ is diagonal and takes the form
	 \be\label{eqn:tildeM_1}
	 (BX^{-1}B^T)_{ij}=\delta_{i,j}(X_i^{-1})_{11}
	 , ~~ i,j\in \{0,\ldots, N\}.
	 \ee

	 It follows directly that
	 \be\label{eqn:tildeM_3}
	 \tilde M =\begin{pmatrix}
	 	b & 1 & 1 & \cdots &1 \\
	 	1 & A_1 & 0 &\cdots &0 \\
	 	1 & 0 & A_2 & \cdots &0\\
	 	\vdots & \vdots & \vdots & \ddots &\vdots \\
	 	1 & 0 & 0& \cdots &A_N
	 \end{pmatrix}
	 \ee
	 where 
	 \be\label{dfn:Ai} A_i= a_0^{(i)} - (X_i^{-1})_{11}.\ee
	 Explicitly, we have 
	 \be\label{eqn:expression_tildeM1}
	 \tilde M^{-1}_{0,0} ={1\over A},~ \tilde M^{-1}_{0,i} =  
	 - {1\over A A_i}, ~\tilde M^{-1}_{i,j}= {1 \over A A_i A_j} ~{\rm for}~i\neq j
	 \ee
	 and 
	 \be
	 (\tilde M^{-1})_{i,i}= {1\over AA_i}\left({b} -\sum_{j\neq i}{1\over A_j} \right)
	 \ee
	 with 
	 \begin{equation}
	 \label{dfn:AA}
	 A := {{\rm det}(\tilde M)\over \prod_{k=1}^N A_k}= b-\sum_k {1\over A_k} .
	 \end{equation}

	At the same time, from (\ref{prooflem1_eq1}) we have
	 \begin{gather}\label{eqn:M0i}
	 \begin{split}
	M_{v_0,v_0}^{-1}= \tilde M^{-1}_{0,0}, ~~ M^{-1}_{v_0,v^{(i)}_{k}} = - \tilde M^{-1}_{0,i} (X^{-1}_i)_{1,k}  
	 \end{split}
	 \end{gather}
	 and 
	 \begin{gather}\label{eqn:Mij}
	 \begin{split}
	 M^{-1}_{v^{(i)}_{\ell_i},v^{(j)}_{\ell_j}} &=\left(X^{-1}+ X^{-1}B^{T}\tilde M^{-1} B X^{-1}\right)_{v^{(i)}_{\ell_i},v^{(j)}_{\ell_j}} \\ &= \tilde M^{-1}_{i,j} (X^{-1}_i)_{1,\ell_i}   (X^{-1}_i)_{1,\ell_j}  + (X^{-1}_i)_{{\ell_i},{\ell_i}} \delta_{i,j} . 
	 \end{split}
	 \end{gather}

	 The above expressions 
	   and (\ref{eqn:expression_tildeM1}) immediately leads to the identity  
	    (\ref{eq:Midentity2}). 
	    
	 It will be illuminating to express the above quantities in terms of the Seifert data (as opposed to plumbing data) directly. To do so, let us compute the entries of $X^{-1}_i$. To avoid an overload of indices, we will momentarily suppress the index $i$ labelling the different legs of the graph when the context is clear.

	Note that $X$ is a tridiagonal matrix. Using the recursion formula for its inverse, we obtain
	\begin{equation}\label{eqn:Xinverse}
	(X^{-1})_{k,k'}= (-1)^{k+k'} {\theta_{k-1}\theta'_{\ell-k'}\over {\rm det}(X)}
	\end{equation}
	for $k\leq  k'$, 
	where $\theta_k$ and $\theta'_k$ are integers satisfying the  recursion formula 
	\begin{gather}\label{dfn:theta}
	\begin{split}
	\theta_k = a_k \theta_{k-1} - \theta_{k-2}  ,   \\ 
	\theta'_k = a_{\ell+1-k} \theta'_{k-1} - \theta'_{k-2}  , ~
	\end{split}
	\end{gather}
	for $k =-1, 0,1,\dots, \ell $, 
	and the initial conditions are 
	\[
	\theta_0= \theta'_0=1, ~~ \theta_{-1}=\theta'_{-1}= 0.  
	\]
	In particular, we have ${\rm det}(X_i) = \theta_{\ell_i}= \theta'_{\ell_i}$. 
	
	Note that 
	
	\[a_{\ell-k} - \cfrac{1}{a_{\ell-k+1} - \cfrac{1}{a_{\ell-k+2} - \cfrac{1}{ \ddots - \cfrac{1}{a_{\ell} } }}} = {\theta'_{k+1}\over\theta'_{k}},  \]
	which leads to
	\[
	(X^{-1})_{11} = {\theta'_{\ell-1} \over \theta'_\ell} = \frac{1}{a_{1} - \cfrac{1}{a_{2} - - \cfrac{1}{ \ddots - \cfrac{1}{a_{\ell} } }}}
	\] and
	similarly
		\begin{equation}\label{eqn:Xellell}
		(X^{-1})_{\ell,\ell} = {\theta_{\ell-1}\over \theta_{\ell}} =  \frac{1}{a_{\ell} - \cfrac{1}{a_{\ell-1} - - \cfrac{1}{ \ddots - \cfrac{1}{a_{1} } }}}.
		\end{equation}

	Restoring the leg-index $i$, we have $a_k=a_k^{(i)}$ and $X=X_i$ in the above, and  similarly $\theta_k=\theta^{(i)}_k$. 
	We then obtain from (\ref{dfn:Ai})
	\begin{equation}\label{eqn:contFrac}
A_i = 	(\tilde M)_{i,i} = a^{(i)}_0 - \frac{1}{a^{(i)}_{1} - \cfrac{1}{a^{(i)}_{2} - - \cfrac{1}{ \ddots - \cfrac{1}{a^{(i)}_{\ell_i} } }}} =-{p_i\over q_i}
	\end{equation}
 \cite{Gukov:2016gkn}
	in terms of the Seifert data $M(b;\{q_i/p_i\}_i)$ of the plumbed manifold, with $p_i\geq 1$, and we can write $\tilde M$ entirely in terms of the Seifert data as 
	\be\label{eqn:tildeM_2}
	\tilde M =\begin{pmatrix}
		b & 1 & 1 & \cdots &1 \\
		1 & -p_1/q_1 & 0 &\cdots &0 \\
		1 & 0 & -p_2/q_2 & \cdots &0\\
		\vdots & \vdots & \vdots & \ddots &\vdots \\
		1 & 0 & 0& \cdots &- p_N/q_N.
	\end{pmatrix}
	\ee
	
	In particular we have from (\ref{dfn:AA}) 
	\be\label{eqn:M00}
	A =
	{{\rm det}(\tilde M)\over \prod_{k=1}^N \tilde M_{k,k} }= 
	b +\sum_k {q_k\over p_k} =\textfrak{e} = {1\over M^{-1}_{v_0,v_0}}\ee
	is given by the orbifold Euler characteristic of the Seifert manifold. 
	    
Using (\ref{eq:Midentity2}),  (\ref{eqn:M0i}) and (\ref{eqn:Mij}) we obtain
\be
\delta = {|\vec \rho|^2\over 2} 
\sum_{i} {(M_{v_0,v^{(i)}_{\ell_i}}^{-1})^2 \over M_{v_0,v_0}^{-1}} - 
M^{-1}_{v^{(i)}_{\ell_i},v^{(i)}_{\ell_i}}
={|\vec \rho|^2\over 2}  \sum_i {((X_i^{-1})_{1,\ell_i})^2\over AA_i}\left(-b+\sum_k {1\over A_k} \right) - 
(X_i^{-1})_{\ell_i,\ell_i}	   \ee
which leads to (\ref{eqn:delta}) upon using  (\ref{eqn:Xellell}) and (\ref{eqn:thetaell}).

To prove (\ref{Arel}), 
 we choose the continued fraction expression  (\ref{eqn:contFrac}) for $-p_i/q_i$ that is effective, namely the   $\{a^{(i)}_{k}\}$ satisfying $a^{(i)}_{k}\leq -2$ for all $k>0$. 
 We will see that the final expression we have derived using this particular choice is invariant under the 3d Kirby moves\footnote{Note that $q^\delta$, on the other hand, depends explicitly on the plumbing graph and is not necessarily invariant under Kirby moves. However, it should be invariant under Kirby moves when combined with $C_\Gamma^G(q)$.  }, which clarifies that making this choice does not lead to a loss of generality. 
With this choice,  it is easy to see that  
$\theta_k = (-1)^{k} |\theta_k|$ and similarly for $\theta'_k$, and in particular 
\be\label{eqn:thetaell} \theta'_\ell  = \theta_\ell=  (-1)^{\ell} |q|. \ee

To sum up,  (\ref{eqn:expression_tildeM1}),  (\ref{eqn:M0i}) and (\ref{eqn:Xinverse})
lead  to 
\begin{equation}\label{M0v_inv}
M^{-1}_{v_0,v^{(i)}_{\ell_i}} = - \tilde M^{-1}_{0,i} (X^{-1}_i)_{1,\ell_i}  = {\rm sgn}(q_i) {1\over \chi p_i}
\end{equation}
which together with (\ref{eqn:M00}) immediately leads to (\ref{Arel}).

\end{proof}

\section{The \texorpdfstring{$A_2$}{A\_2} Character Identity}
\label{charc_id}

Let $\chi^{0}_{a,b;s_1,s_2}$ be generalized $A_2$ characters defined as in \eqref{eqn:chiMin}. In this appendix we prove the identity: 
\begin{lem}\label{lem:charID}
\be\label{eqn:charID_B}
\chi^{0}_{0,0;s_1,s_2+m}= 3 \chi^{0}_{1,0;s_1,s_2}
- \chi^{0}_{0,0;m-s_1-s_2,s_2} + \chi^{0}_{0,0;s_1+s_2,m-s_2}. 
\ee
where $s_1, s_2$ are defined as in \eqref{eqn:labelMods}. 
\end{lem}

Note that the identity in Lemma~\ref{lem:charID} holds for all $s_1,s_2\in \ZZ$ and is not subjected to the condition $s_1, s_2 \in \{ 1,2, \cdots, m \}$, which is natural from the point of view of representation theory of the singlet algebras.  
Does this imply the existence of some interesting isomorphism among modules of log VOA? We will leave the question for future work and simply prove the algebraic identity here.

\subsection{Notation}
For sake of brevity, let us fix $(s_1,s_2)$ and define:
\begin{multline}
(( a,b )) := \sum_{m_1,m_2,m_3 \geq 0} q^{m(m_1^2 + m_2^2 + m_3^2 - m_1 m_2 + m_2 m_3 + m_1 m_3)} \\
\times q^{a m_1+b m_2+ (a+b) m_3} q^{(s_1^2+s_2^2+s_1s_2)/3m},
\end{multline}
and its restriction to $m_2 =0$:
\begin{equation}
((a,b))_0 := \sum_{m_1,m_3 \geq 0} q^{m(m_1^2 + m_3^2 + m_1 m_3)} q^{a m_1 + (a+b) m_3} q^{(s_1^2+s_2^2+s_1s_2)/3m}.
\end{equation}

The above $q$-series have several useful properties, which follow from the definitions of $((a,b))$ and $((a,b))_0$:
\begin{align}
((a,b)) &= ((b,a)) \label{eqn:parenSym} \\
((a,b))_0 &= ((a+b,-b))_0 \label{eqn:paren0} \\
((a-m,b+2m)) &= q^{-b-m}((a,b)) - q^{-b-m} ((a,b))_{0} \label{eqn:parenShift1} \\
((a+m,b-2m)) &= q^{b-m}((a,b)) + ((a+m,b-2m))_{0} \label{eqn:parenShift2}
\end{align}

In terms of $((\cdot,\cdot))$, they are simplified into:
\begin{equation}\label{eqn:chi0simple}
\begin{aligned}
\eta^2 \chi^{0}_{0,0;s_1,s_2}(\tau) &= q^{m+s_2}(( m-s_1, m+s_1+s_2 )) + q^{m+s_1}((m+s_1+s_2,m-s_2)) \\
&\quad + q^{m-s_1-s_2}((m-s_2,m-s_1)) - q^{m+s_1+s_2}((m+s_1,m+s_2)) \\
&\quad - q^{m-s_2}((m-s_1-s_2,m+s_1)) - q^{m-s_1}((m+s_2,m-s_1-s_2)),
\end{aligned}
\end{equation}
\begin{equation}
\begin{aligned}
\eta^2 \chi^{0}_{1,0;s_1,s_2}(\tau) &= q^{\frac{m-2s_1-s_2}{3}} ((-s_2,m-s_1)) + q^{\frac{m+ s_1 - s_2}{3}} ((s_1+s_2,m-s_2)) \\
&\quad +q^{\frac{m+s_1+2s_2}{3}}((-s_1,m+s_1+s_2))-q^{\frac{m-2s_1-s_2}{3}}((s_2,m-s_1-s_2)) \\
&\quad -q^{\frac{m+s_1-s_2}{3}}((-s_1-s_2,m+s_1))-q^{\frac{m+s_1+2s_2}{3}}((s_1,m+s_2)).
\end{aligned}
\end{equation}

\subsection{Proof of the character identity}
Let us first rewrite \eqref{eqn:charID_B} in a more symmetric way,
\begin{multline}\label{eqn:charIDsym}
\Big( \eta^2 \chi^{0}_{0,0;s_1,s_2+m}(\tau) - \eta^2 \chi^{0}_{1,0;s_1,s_2}(\tau) \Big) = \Big(-\eta^2 \chi^{0}_{0,0;m-s_1-s_2,s_2}(\tau) +  \eta^2 \chi^{0}_{1,0;s_1,s_2}(\tau) \Big)  \\
+ \Big(  \eta^2 \chi^{0}_{0,0;s_1+s_2,m-s_2}(\tau) + \eta^2 \chi^{0}_{1,0;s_1,s_2}(\tau) \Big).
\end{multline}
We will simplify the differences term-by-term.

The LHS can be written in $((a,b))$ notation as follows. By doing so, we observe that $\eta^2 \chi^{0}_{0,0;s_1,s_2+m}(\tau)$ and $\eta^2 \chi^{0}_{1,0;s_1,s_2}(\tau)$ have the following term in common
\begin{equation}
q^{\frac{m-2s_1-s_2}{3}} ((-s_2,m-s_1))-q^{\frac{m+s_1-s_2}{3}}((-s_1-s_2,m+s_1))
\end{equation}
which thus cancel out. As a result, the LHS is written as:
\begin{equation}
\begin{aligned}
\eta^2 \chi^{0}_{0,0;s_1,s_2+m}(\tau) &-\eta^2 \chi^{0}_{1,0;s_1,s_2}(\tau) = \\
&q^{\frac{7m + s_1 + 5 s_2}{3}} (( m-s_1, 2m+s_1+s_2)) + q^{\frac{4m +4s_1 + 2 s_2}{3}}((2m+s_1+s_2,-s_2)) \\
&+q^{\frac{m-2s_1-s_2}{3}}((s_2, m-s_1-s_2)) + q^{\frac{m+s_1+2s_2}{3}}((s_1, m+s_2)) \\
&- q^{\frac{7m + 4s_1+5s_2}{3}}((m+s_1, 2m+s_2 )) - q^{\frac{4m-2s_1+2s_2}{3}}((2m+s_2,-s_1-s_2)) \\
&-q^{\frac{m+s_1+2s_2}{3}}((-s_1, m+s_1+s_2)) - q^{\frac{m+s_1-s_2}{3}}((s_1+s_2, m-s_2)).
\end{aligned}
\end{equation}

Note that all $((a,b))$-terms appear pairwise under the shifting properties, \eqref{eqn:parenShift1} and \eqref{eqn:parenShift2}; for instance,  the terms $((m-s_1,2m+s_1+s_2))$ and $((-s_1,m+s_1+s_2))$ can be simplified using the shift property (\ref{eqn:parenShift2}):
\begin{equation}
\begin{aligned}
q^{\frac{7m + s_1 + 5 s_2}{3}}((m - s_1,2m+s_1+s_2)) &= q^{\frac{4m -2s_1 + 2 s_2}{3}}((2m - s_1,s_1+s_2)) \\
&\qquad -q^{\frac{4m -2s_1 + 2 s_2}{3}}((2m - s_1,s_1+s_2))_0 \\
- q^{\frac{m + s_1 + 2 s_2}{3}}((-s_1,m+s_1+s_2)) &= - q^{\frac{4m -2 s_1 + 2 s_2}{3}}((2m-s_1,s_1+s_2)) \\
&\qquad - q^{\frac{m + s_1 + 2 s_2}{3}} ((m+s_1+s_2,-s_1))_0
\end{aligned}
\end{equation}
where we have used symmetricity of $((a,b))$ when necessary. 
Note  that the common $((2m-s_1,s_1+s_2))$ term cancels out. After accounting for all such cancellations under the shifting properties, we are left with:
\begin{equation}\label{eqn:charID_LHS_simple}
\begin{aligned}
\eta^2 (\chi^{0}_{0,0;s_1,s_2+m}(\tau) &- \chi^{0}_{1,0;s_1,s_2}(\tau) )  \\
=&- q^{\frac{4m -2 s_1 + 2 s_2}{3}}((2m-s_1,s_1+s_2))_0
- q^{\frac{m + s_1 + 2 s_2}{3}} ((m+s_1+s_2,-s_1))_0 \\
&- q^{\frac{m +s_1 - s_2}{3}}((s_1+s_2,m-s_2))_0 + q^{\frac{m-2s_1-s_2}{3}}((s_2,m-s_1-s_2))_0 \\
&+q^{\frac{4m + 4s_1+2s_2}{3}}((2m+s_1, s_2 ))_0 +q^{\frac{m+s_1+2s_2}{3}}((m+s_2,s_1))_0  \\
=~ & q^{\frac{4m + 4s_1+2s_2}{3}}((2m+s_1, s_2 ))_0  + q^{\frac{m-2s_1-s_2}{3}}((s_2,m-s_1-s_2))_0 \\
&- q^{\frac{m +s_1 - s_2}{3}}((s_1+s_2,m-s_2))_0 - q^{\frac{4m -2 s_1 + 2 s_2}{3}}((2m-s_1,s_1+s_2))_0 , 
\end{aligned}
\end{equation}
where in the last equality we have used 
$((a,b))_0 = ((a+b,-b))_0$ to write 
$((a,b))_0 = ((a+b,-b))_0$:
\begin{equation}
- q^{\frac{m + s_1 + 2 s_2}{3}} ((m+s_1+s_2,-s_1))_0+q^{\frac{m+s_1+2s_2}{3}}((m+s_2,s_1))_0 = 0.
\end{equation}

We will repeat the above analysis for the RHS of \eqref{eqn:charIDsym}. Following the same procedure as in LHS, the first term of \eqref{eqn:charIDsym} becomes:
\begin{equation}
\begin{aligned}
\eta^2 (- \chi^{0}_{0,0;m-s_1-s_2,s_2}(\tau) &+ \chi^{0}_{1,0;s_1,s_2}(\tau)) = \\
&- q^{\frac{4m-2s_1+2s_2}{3}}((2m+s_2,-s_1-s_2))_0 - q^{\frac{m + s_1- s_2}{3}}((m+s_1,-s_1-s_2))_0  \\
&+ q^{\frac{4m-2s_1-4s_2}{3}}((2m-s_2,-s_1))_0 + q^{\frac{m+s_1+2s_2}{3}}((m+s_1+s_2,-s_1))_0
\end{aligned}
\end{equation}
and the second term of \eqref{eqn:charIDsym} becomes:
\begin{equation}
\begin{aligned}
\eta^2 ( \chi^{0}_{0,0;s_1+s_2,m-s_2}(\tau)&+ \chi^{0}_{1,0;s_1,s_2}(\tau)) = \\
& q^{\frac{m-2s_1- s_2}{3}} (( m-s_1, -s_2))_0 + q^{\frac{4m+4s_1+2 s_2}{3}} (( 2m+s_1+s_2, -s_2))_0 \\
&-q^{\frac{4m-2s_1-4 s_2}{3}} (( 2m-s_1-s_2, s_1))_0 -q^{\frac{m+s_1+2 s_2}{3}} (( m + s_2, s_1))_0 \end{aligned}
\end{equation}

Now we are ready to complete the proof. First, collect both terms on the RHS, we get 
\begin{equation}
\begin{aligned}
RHS &= q^{\frac{m-2s_1- s_2}{3}} (( m-s_1, -s_2))_0 + q^{\frac{4m+4s_1+2 s_2}{3}} (( 2m+s_1+s_2, -s_2))_0 \\ 
&+ q^{\frac{4m-2s_1-4s_2}{3}}((2m-s_2,-s_1))_0 + q^{\frac{m+s_1+2s_2}{3}}((m+s_1+s_2,-s_1))_0 \\
&- q^{\frac{4m-2s_1+2s_2}{3}}((2m+s_2,-s_1-s_2))_0 - q^{\frac{m + s_1- s_2}{3}}((m+s_1,-s_1-s_2))_0 \\
&-q^{\frac{4m-2s_1-4 s_2}{3}} (( 2m-s_1-s_2, s_1))_0 -q^{\frac{m+s_1+2 s_2}{3}} (( m + s_2, s_1))_0\\
&= q^{\frac{m-2s_1- s_2}{3}} (( m-s_1, -s_2))_0 + q^{\frac{4m+4s_1+2 s_2}{3}} (( 2m+s_1+s_2, -s_2))_0 \\ 
&- q^{\frac{4m-2s_1+2s_2}{3}}((2m+s_2,-s_1-s_2))_0 - q^{\frac{m + s_1- s_2}{3}}((m+s_1,-s_1-s_2))_0 
\end{aligned}
\end{equation}
where we have again used 
 $((a,b))_0 = ((a+b,-b))_0$
 in the last equality. Comparing with \eqref{eqn:charID_LHS_simple}, we find that LHS = RHS in \eqref{eqn:charIDsym}. This completes the proof.

\bibliography{higherrank_bib.bib}
\bibliographystyle{ytphys.bst}

\end{document}